\documentclass[11pt,reqno]{amsart}
\usepackage[utf8]{inputenc}
\usepackage[T1]{fontenc}
\usepackage{charter} 

\usepackage{amsmath, amssymb, amsthm, amsfonts}

\usepackage{geometry}
\usepackage{fancyhdr}
\usepackage{microtype}
\usepackage{enumitem}

\usepackage[all]{xy}
\usepackage{tikz-cd}

\usepackage{thmtools}

\usepackage[
    colorlinks=true, 
    linkcolor=blue,  
    citecolor=red,  
    urlcolor=blue,  
    hypertexnames=false
]{hyperref} 
\usepackage{cleveref}
\newtheorem{theorem}{Theorem}[section]
\newtheorem{lemma}[theorem]{Lemma}
\newtheorem{proposition}[theorem]{Proposition}
\newtheorem{corollary}[theorem]{Corollary}

\theoremstyle{definition}
\newtheorem{definition}[theorem]{Definition}
\newtheorem{example}[theorem]{Example}
\newtheorem{remark}[theorem]{Remark}
\newcommand{\CatPL}{\mathbf{PLGrp}_{\text{reg}}^{\text{sc}}}
\newcommand{\CatLB}{\mathbf{LieBiAlg}_{\text{reg}}}

\numberwithin{equation}{section}

\title[Drinfeld Correspondence in Infinite-Dimension]{Drinfeld Correspondence in Infinite Dimensions}

\author{Praful Rahangdale}
\address{
Praful Rahangdale:
 Institut f\"ur Mathematik, Universit\"at Paderborn,
Warburger Str. 100, 33098 Paderborn, Germany}
\email{praful@math.uni-paderborn.de}


\begin{document}

 \begin{abstract}
 In this article, we establish the Drinfeld correspondence between Poisson Lie groups and their infinitesimal counterparts, Lie bialgebras, in the infinite-dimensional setting. Specifically, we extend this correspondence to regular Lie groups modeled on convenient vector spaces, with a particular focus on those modeled on nuclear Fr\'{e}chet and nuclear Silva spaces. Important examples of interest include the smooth loop group $C^{\infty}(\mathbb{S}^{1}, G)$ and the analytic loop group $C^{\omega}(\mathbb{S}^{1}, G)$ of a 1-connected real Lie group $G$, as well as $\widetilde{\mathrm{Diff}^{\infty}(M)_0}$ and $\widetilde{\mathrm{Diff}^{\omega}(M)_0}$---the universal covering groups of the identity components of the groups of smooth and real-analytic diffeomorphisms of a compact manifold $M$.
 \end{abstract}
\maketitle

\tableofcontents

\section{Introduction}

Poisson structures provide a mathematical framework to describe the dynamics of a classical system, where the phase space is a smooth manifold or a Lie group. A fundamental task in this context is to restrict the dynamics to certain submanifolds; a system possessing sufficiently many such submanifolds is called an integrable system. While Poisson structures and methods for identifying integrable systems are well-understood in finite dimensions, many physically significant systems—such as the Korteweg--de Vries (KdV) hierarchy and the non-linear Schrödinger equation—evolve in infinite-dimensional phase spaces \cite{9, 10, 16, 19}. Our aim is to develop a mathematical framework for describing such dynamics via Poisson structures on manifolds modeled on infinite-dimensional locally convex topological vector spaces.

(For the convenience of the reader, the relevant basic concepts and notation of infinite-dimensional calculus used throughout this work are collected in Appendix \ref{appendix:calculus}.)

The notion of a Poisson Lie group and its infinitesimal counterpart, the Lie bialgebra, were introduced by Drinfeld in \cite{7}, where the one-to-one correspondence between Poisson Lie groups and Lie bialgebras in the finite-dimensional case is established. We explore a framework in which this correspondence holds in an infinite-dimensional setting. Various difficulties arise when extending results concerning Poisson structures, leading to different approaches in the literature (e.g., \cite{5, 26, 32}), as summarized in \cite{15}. Primarily, the following obstacles prevent the direct application of finite-dimensional techniques:
\begin{enumerate}
    \item The set of differentials $df(x)$ for $f \in C^{\infty}(M)$ may not coincide with the cotangent space $T_{x}^*M$:
    \[
    \{df(x) \mid f \in C^{\infty}(M), x \in M\} \neq T_{x}^*M.
    \]
    \item Hamiltonian vector fields may not exist, as not all derivations on $C^{\infty}(M)$ are represented by smooth vector fields on the manifold.
    \item There is a failure of the isomorphism between multilinear forms and the exterior power of the tangent space: $L_{\text{alt}}^{k}(T_{x}^*M) \ncong \hat{\wedge}^{k}T_{x}M$.
\end{enumerate}

The correspondence between Banach Lie bialgebras and Manin triples was studied in \cite{32}. It was proven that the Lie algebra $\mathfrak{g}$ of a Poisson Lie group $(G, \mathbb{F}, \pi)$ carries both a Lie bialgebra structure and a Lie-Poisson space structure with respect to $\mathfrak{b} := \mathbb{F}_{e}$. This article provides a positive response to the open problem concerning the integration of a Banach Lie bialgebra $(\mathfrak{g}, \mathfrak{b}, \delta)$ to a Poisson Lie group structure $(G, \mathbb{F}_{\mathfrak{b}}, \pi)$, provided $\mathfrak{g}$ is the Lie algebra of a 1-connected Banach Lie group $G$ (see \cite[Remark 5.14]{32}).

In the algebraic setting, standard Lie bialgebra structures and their classical twists on loop algebras are explored in \cite{1, 2, 3}. A construction of a Lie bialgebra structure for smooth loop algebras is provided in \cite{23}. For manifolds that are regular as topological spaces and are modeled on nuclear Fréchet or nuclear Silva spaces, the following special facts are available:
\begin{enumerate}
    \item The existence of smooth bump functions \cite[16.10]{20} implies that $\{df(x) \mid f \in C^{\infty}(M)\} = T_{x}^*M$.
    \item The space of continuous derivations on $C^{\infty}(M)$ is topologically isomorphic to the space of smooth vector fields, ensuring the existence of Hamiltonian vector fields $X_{f} := \{-, f\} \in \Gamma(TM)$ \cite{30}.
    \item Grothendieck's solutions to the ``problem of topologies'' in the nuclear case imply that $L_{\text{alt}}^{k}(T_{x}^*M) \cong \hat{\wedge}^{k}T_{x}M$ is a topological isomorphism \cite[Corollary 2, 21.2]{17}, where $\hat{\wedge}$ denotes the skew-symmetric tensor product completed in the projective tensor product topology.
\end{enumerate}

We demonstrate that in this special case, the Drinfeld correspondence (Theorem \ref{D-4}) closely resembles the classical finite-dimensional setting. We aim to generalize this correspondence. One direction holds in significant generality: we show that the differentiation of the Poisson bivector $\pi$ of a convenient Poisson Lie group $(G, \mathbb{F}, \pi)$ at the identity defines a Lie bialgebra structure $(\mathfrak{g}, \mathbb{F}_e, d\pi(e))$ on the Lie algebra $\mathfrak{g}$ (Theorem \ref{D-1}).

The converse holds if $G$ is a 1-connected regular Lie group. We show that there exists a left adjoint integration functor (denoted ``Int'') between the category of Poisson structures on 1-connected regular Lie groups, $\CatPL$, and the category of Lie bialgebras, $\CatLB$ (Theorem \ref{equivalence}).

\subsection{Structure of the article}
In this article, we aim to prove two versions of the Drinfeld theorem in an infinite-dimensional setting based on the availability of Hamiltonian vector fields. For an infinite-dimensional convenient smooth manifold $M$ modeled on a convenient space $E$, we begin our article by extending the notion of a Poisson structure $(M, \mathbb{F}, \pi)$ in Definition \ref{3.2}. 

(\textit{For simplicity, we will call a convenient smooth manifold a smooth manifold and convenient smooth morphisms smooth morphisms, respectively.})

We then discuss Poisson Lie group structures in the convenient setting. In Section \ref{4}, we prove that the linearization of the Poisson Lie group structure $(G, \mathbb{F}_{\mathfrak{b}}, \pi)$ at the identity induces a Lie bialgebra structure $(\mathfrak{g}, \mathfrak{b}, \delta)$ on $\mathfrak{g}$ (Theorem \ref{D-1}). Moreover, if the existence of Hamiltonian vector fields is assumed ($\pi^{\sharp}(\mathbb{F}_{\mathfrak{b}}) \subseteq TG$), we show that $(\mathfrak{g}, \mathfrak{b}, \delta)$ is a Lie bialgebra and $\mathfrak{g}$ is a Lie-Poisson space with respect to $\mathfrak{b}$. The one-to-one correspondence between Manin triples $(\mathfrak{d}, \mathfrak{g}, \mathfrak{b})$ and Lie bialgebra structures $(\mathfrak{g}, \mathfrak{b}, \delta)$ is also discussed in Section \ref{4}. 

In Section \ref{5}, we provide a method to construct examples of Manin triples from a $C^{*}$-algebra equipped with a torus-invariant, faithful, tracial state. We provide explicit examples of Manin triples, including:
\begin{enumerate}
    \item $C^{\infty}(\mathbb{S}^{1}, \mathfrak{g})$; the smooth loop algebra of a real semisimple Lie algebra $\mathfrak{g}$.
    \item $C^{\omega}(\mathbb{S}^{1}, \mathfrak{g})$; the analytic loop algebra of a real semisimple Lie algebra $\mathfrak{g}$.
    \item $\Gamma^{\infty}(T\mathbb{S}^{1})$; the space of smooth vector fields on $\mathbb{S}^{1}$.
    \item $\Gamma^{\omega}(T\mathbb{S}^{1})$; the space of real-analytic vector fields on $\mathbb{S}^{1}$.
\end{enumerate}

In Section \ref{6}, we prove the Drinfeld theorem in both cases: when Hamiltonian vector fields exist and when they do not.

\begin{theorem}[see Theorem \ref{D-1}]
    Let $(G, \mathbb{F}_{\mathfrak{b}}, \pi)$ be a convenient Poisson Lie group with Lie algebra $\mathfrak{g}$. Let $\mathfrak{b} \subseteq \mathfrak{g}'$ be a convenient space such that the canonical map $\operatorname{incl}: \mathfrak{b} \to \mathfrak{g}', \alpha \mapsto \alpha$ is smooth, and the coadjoint action $\operatorname{Ad}^*: G \times \mathfrak{g}' \to \mathfrak{g}'$ restricts to a smooth action on $\mathfrak{b}$. Then the differential $\delta := d\pi(e)$ of $\pi$ at the identity induces a unique (up to isomorphism) convenient Lie bialgebra structure $(\mathfrak{g}, \mathfrak{b}, \delta)$ on $\mathfrak{g}$, where $\mathbb{F}_{\mathfrak{b}}(e) = \mathfrak{b}$. Moreover, if $\mathfrak{g}$ and $\mathfrak{b}$ are Fréchet or Silva spaces, then it defines a topological (continuous) Lie bialgebra structure.
\end{theorem}

\begin{theorem}[see Theorem \ref{D-2}]
    Let $G$ be a 1-connected convenient regular Lie group. Let $\mathfrak{b} \subseteq \mathfrak{g}'$ be a convenient space such that the canonical map $\operatorname{incl}: \mathfrak{b} \to \mathfrak{g}', \alpha \mapsto \alpha$ is smooth, and the coadjoint action $\operatorname{Ad}^*: G \times \mathfrak{g}' \to \mathfrak{g}'$ restricts to a smooth action on $\mathfrak{b}$. Then a convenient Lie bialgebra structure $(\mathfrak{g}, \mathfrak{b}, \delta)$ integrates to a convenient Poisson Lie group structure $(G, \mathbb{F}_{\mathfrak{b}}, \pi)$, where $\mathbb{F}_{\mathfrak{b}}$ is generated by the right translation of $\mathfrak{b}$ by the action of $G$ (i.e., $\mathbb{F}_{\mathfrak{b}}(g) := R_{g}^*\mathfrak{b}$), and $\pi$ is given by the right translation of the group 1-cocycle $\Theta$ of $G$, where $\Theta$ is the integration of the Lie algebra 1-cocycle $\delta$ given by formula \eqref{1-cocycle}.
\end{theorem}

\begin{theorem}[see Theorem \ref{D-3}]
    Let $G$ be a 1-connected regular Lie group with convenient Lie algebra $\mathfrak{g}$. Let $\mathfrak{b} \subseteq \mathfrak{g}'$ be a convenient space such that the canonical map $\operatorname{incl}: \mathfrak{b} \to \mathfrak{g}', \alpha \mapsto \alpha$ is smooth and the coadjoint action $\operatorname{Ad}^*: G \times \mathfrak{g}' \to \mathfrak{g}'$ restricts to a smooth action on $\mathfrak{b}$. Then the following assertions are equivalent up to isomorphism:
    \begin{enumerate}
        \item $(G, \mathbb{F}_{\mathfrak{b}}, \pi)$ is a Poisson Lie group structure such that $\pi^{\sharp}(\mathbb{F}_x) \subset T_xG$ for each $x \in G$.
        \item $(\mathfrak{g}, \mathfrak{b}, \delta)$ is a Lie bialgebra structure, and $\mathfrak{g}$ is a Lie-Poisson space with respect to $\mathfrak{b}$.
        \item $(\mathfrak{d} = \mathfrak{g} \oplus \mathfrak{b}, \mathfrak{g}, \mathfrak{b})$ is a Manin triple.
    \end{enumerate}
\end{theorem}

Later, we show that the correspondence between convenient Poisson Lie groups and their Lie bialgebras extends to an equivalence between the category of Poisson structures on 1-connected regular Lie groups $\CatPL$ and the category of Lie bialgebras $\CatLB$.

\begin{theorem}[Equivalence between $\CatPL$ and $\CatLB$, see Theorem \ref{equivalence}]
    The differentiation functor $\mathcal{L}ie$, defined by taking the tangent map at the identity, establishes an equivalence of categories between $\CatPL$ and $\CatLB$. Its inverse is a unique left adjoint functor (Integration, denoted ``Int'') up to natural isomorphism. This means every Lie bialgebra structure $(\mathfrak{g}, \mathfrak{b}, \delta)$ in $\CatLB$ integrates to a unique Poisson Lie group structure $(G, \mathbb{F}, \pi)$ on the corresponding Lie group $G$ such that $\mathcal{L}ie(G) = \mathfrak{g}$.
\end{theorem}

In the final section, we discuss a particular case where the modeling space of a Lie group is a nuclear Fréchet or a nuclear Silva space. We discuss examples fitting this setting, such as the smooth loop group $C^{\infty}(\mathbb{S}^{1}, G)$ and the real-analytic loop group $C^{\omega}(\mathbb{S}^{1}, G)$ of a finite-dimensional real Lie group $G$, as well as $\widetilde{\mathrm{Diff}^{\infty}(M)_0}$ and $\widetilde{\mathrm{Diff}^{\omega}(M)_0}$---the universal covering groups of the identity components of the groups of smooth and real-analytic diffeomorphisms of a compact manifold $M$.
\section{Poisson structures on infinite-dimensional manifolds}

Let $\pi_{\mathbb{F}}: \mathbb{F} \to M$ be a smooth vector bundle over a smooth manifold $M$ with a convenient typical fiber $F$. Let $L^2_{skew}(F)$ denote the space of continuous skew-symmetric bilinear forms on $F$, equipped with the topology of uniform convergence on the bounded subsets of $F\times F$.
We define the vector bundle $L^2_{skew}(\mathbb F)$ over $M$ as
\[ L^2_{skew}(\mathbb F):=\bigcup_{x\in M}L^2_{skew}(\mathbb F_x)\]
 where $L^2_{skew}(\mathbb F_x)$ denotes the space of all skew-symmetric bounded bilinear forms on $\mathbb F_x$, which is equipped with a convenient structure by \cite[Proposition 5.6]{20}.
The fibers of our new bundle are given by $L^2_{skew}(\mathbb{F})_x := L^2_{skew}(\mathbb{F}_x)$. The bundle projection $\pi_{skew}: L^2_{skew}(\mathbb{F}) \to M$ maps $\omega \in L^2_{skew}(\mathbb{F}_x)$ to $x \in M$.

For any local trivialization of the bundle $\mathbb{F}$ over an open set $U \subseteq M$, denoted by $\tau: \mathbb{F}|_U \to U \times F$, we have a fiber-wise topological isomorphism $\tau_x: \mathbb{F}_x \to F$ for each $x \in U$.

We define the induced local trivialization for $L^2_{skew}(\mathbb{F})$ over $U$ as
\[
\Theta_\tau : L^2_{skew}(\mathbb{F})|_U \to U \times L^2_{skew}(F)
\]
\[
\Theta_\tau(\omega) := (x, (\tau_x^{-1})^*\omega)
\]
where $x = \pi_{skew}(\omega)$, and the pullback $(\tau_x^{-1})^*\omega \in L^2_{skew}(F)$ is defined for all $u, v \in F$ by
\[
((\tau_x^{-1})^*\omega)(u, v) := \omega(\tau_x^{-1}(u), \tau_x^{-1}(v)).
\]

The inverse of this trivialization map is given by
\[
\Theta_\tau^{-1}: U \times L^2_{skew}(F) \to L^2_{skew}(\mathbb{F})|_U
\]
\[
(x, \eta) \mapsto \tau_x^*\eta
\]
where $\tau_x^*\eta \in L^2_{skew}(\mathbb{F}_x)$ is defined by $(\tau_x^*\eta)(w_1, w_2) = \eta(\tau_x(w_1), \tau_x(w_2))$ for $w_1, w_2 \in \mathbb{F}_x$.

To verify that this defines a smooth vector bundle structure, let $\sigma: \mathbb{F}|_V \to V \times F$ be another local trivialization such that $U \cap V \neq \emptyset$. The transition map for $\mathbb{F}$ is given by a smooth map $x \mapsto \sigma_x \circ \tau_x^{-1} \in GL(F)$.

The corresponding transition map for $L^2_{skew}(\mathbb{F})$ over $U \cap V$ is:
\[
\Theta_\sigma \circ \Theta_\tau^{-1}: (U \cap V) \times L^2_{skew}(F) \to (U \cap V) \times L^2_{skew}(F)
\]
\[
(x, \eta) \mapsto (x, (\sigma_x^{-1})^*(\tau_x^*\eta)) = (x, (\tau_x \circ \sigma_x^{-1})^*\eta).
\]

Since the map $x \mapsto \tau_x \circ \sigma_x^{-1} \in GL(F)$ is smooth, and the pullback operation on bounded bilinear forms depends smoothly on the operator $\tau_x \circ \sigma_x^{-1}$, the induced transition map into $GL(L^2_{skew}(F))$ is also smooth. Therefore, $L^2_{skew}(\mathbb{F})$ is a smooth vector bundle over $M$.
\begin{definition}[Weak subbundle]
    Let $p:E\to M$ be a convenient vector bundle over a smooth manifold $M$. A vector subbundle $p':F\to M$ of $p:E\to M$ endowed with a convenient structure is called a weak subbundle of $E$ if the canonical mapping $i_{F}: F\to E, v\mapsto v$ is a smooth vector bundle morphism with respect to the convenient structure of the vector bundle $p':F\to M$. 
\end{definition}

\begin{definition}[Poisson structure]\label{3.2}
Let $M$ be a convenient manifold modeled on a convenient space $E$. Let $\mathbb{F}$ be a weak  subbundle of $T'M$ with the canonical smooth vector bundle morphism $i_{\mathbb F}: \mathbb F\to T'M$ such that $\mathbb F_x$ separates points in $T_xM$ for each $x\in M$.

Let $\pi \in \Gamma(M, L_{\text{skew}}^{2}(\mathbb{F}))$ be a smooth bivector  such that for all local sections $\alpha,\beta,\gamma\in\Gamma(\mathbb F)$
satisfying $d\alpha=d\beta=d\gamma=0$, the following hold:
\begin{enumerate}
 \item The differential
$d(\pi(\alpha,\beta))$ factors uniquely through $\mathbb F$, i.e.
there exists $\widetilde{d(\pi(\alpha,\beta))}\in\Gamma(\mathbb F)$
such that
\[
d(\pi(\alpha,\beta))
=
i_{\mathbb F}\circ
\widetilde{d(\pi(\alpha,\beta))}.
\]
  \item The Jacobi identity
\[
\pi\big(\alpha,\widetilde{d(\pi(\beta,\gamma))}\big)
+
\pi\big(\beta,\widetilde{d(\pi(\gamma,\alpha))}\big)
+
\pi\big(\gamma,\widetilde{d(\pi(\alpha,\beta))}\big)
= 0
\]
holds.
\end{enumerate}
Then, the pair $(M, \mathbb{F}, \pi)$ is called a Poisson structure. 
\end{definition}

 \begin{corollary}\label{Poisson algebra}
 Let $(M,\mathbb F, \pi)$ be a Poisson structure. Then, the subalgebra 
 $$\mathcal A_{\mathbb F}:=\{f\in C^{\infty}(M): df \text{ factors as } i_{\mathbb F}\circ\widetilde{df} \text{ for  } \widetilde{df}\in \Gamma(\mathbb F)\}$$
with the convenient vector space structure induced by the linear mappings 
\begin{equation*}    
 \mathcal A_\mathbb F \xrightarrow{\operatorname{incl}} C^{\infty}(M,\mathbb R), f\mapsto f,\end{equation*}
$$\quad  
 \mathcal A_\mathbb F \xrightarrow {\widetilde{d}} \Gamma(\mathbb F), f\mapsto \widetilde{df}$$  is a Poisson algebra with respect to the bounded Lie bracket  $$\{-,-\}:\mathcal A_{\mathbb F}\times \mathcal A_{\mathbb F}\to \mathcal A_{\mathbb F}\\ , \{f, g\} := \pi(\widetilde{df}, \widetilde{dg}).$$
 \end{corollary}
 \begin{proof}
 The Poisson bracket \[\{-, -\}: \mathcal{A}_{\mathbb F} \times \mathcal{A}_{\mathbb F} \rightarrow \mathcal{A}_{\mathbb F}, (f,g)\mapsto \pi(\widetilde{df},\widetilde{dg})\] is bounded as $\pi \in L^2_{skew}(\mathbb F)$ is smooth, the differential operator $\widetilde d :\mathcal A_{\mathbb F}\to \Gamma(\mathbb F)$ is a bounded operator as the topology on $\mathcal A_{\mathbb F}$ is initial with respect to the mappings $\text{incl}$ and $\widetilde{d}$, and the evaluation $ (\pi,   \widetilde{df}, \widetilde{dg})\mapsto \pi(\widetilde{df}, \widetilde{dg})$ is bounded (evaluation map is always smooth in the convenient setting of \cite[Corollary 3.13]{20}). Now, consider 
\[ \{\{f,g\},h\}=\pi(d(\pi(\widetilde{df}, \widetilde{dg})),\widetilde{dh}).\] Thus, the Jacobi identity of $\{-,-\}$ follows from the Jacobi identity of $\pi$ (condition (2) of Definition \ref{3.2}).
To show that $\{-,-\}$ satisfies the Leibniz identity, we note that \[
\{f,gh\}=\pi(\widetilde{df}, \widetilde{d(gh))}=\pi(\widetilde{df}, g\widetilde{d(h)}+\widetilde{d(g)}h)=g\{f,h\}+\{f,g\}h.
\]
 \end{proof}
 \begin{remark}  
Note that the converse of Corollary 3.3  might not always hold. Let $\mathcal A\subseteq C^\infty( M)$ be a Poisson algebra equipped with a bounded Poisson bracket. Then, the bundle $\mathbb F$ formed by \[ \mathbb{F}_x := \{df(x) \mid f \in \mathcal{A} \}\] may not have a smooth vector bundle structure. Even if $\mathbb F$ has a smooth vector bundle structure, the existence of a smooth bivector $\pi \in \Gamma(L^2_{skew}(\mathbb F))$ is not guaranteed, as for each infinite-dimensional Hilbert manifold $H$, Poisson bracket exists on $C^\infty(M)$ such that the function $\{f,g\}$ for $f,g \in \mathcal A$ does not always depend on the 1-jets of $f$ and $g$ (see \cite{38}). This is because for each infinite-dimensional Hilbert manifold, there exist operational tangent vectors of the second order (see \cite[Remark 28.8]{20}). But in a special situation, when a manifold is modeled on a smoothly paracompact convenient manifold with bornological approximation property (see \cite[6.6]{20}), it is shown in \cite{39} that the dual map of the Poisson bracket, $\widehat{\{-,-\}}':C^\infty(M)'\rightarrow{} L^2_{skew}(C^\infty(M),\mathbb R)$ factors through a smooth section of a vector bundle $L^2_{skew}(T'M)$.
\end{remark}
\begin{remark}\label{smoooth anchor}
We will denote a Poisson structure $(M, \mathbb{F}, \pi)$ by $(M, \mathbb{F}, \rho)$ if the anchor map 
\[
\rho := \pi^\sharp: \mathbb{F} \to T''M, \quad \alpha \mapsto \pi(\alpha, -) 
\] 
takes values in $TM \subseteq T''M$. Note that $\rho$ is a smooth bundle morphism because $\pi: M \to L^2_{\text{skew}}(\mathbb{F})$ is smooth. By passing to a local trivialization and using the exponential law for multilinear mappings (\cite[Proposition 5.2]{20}), $\pi$ is smooth if and only if the induced bundle map $\widehat{\pi}: \mathbb{F} \to \mathbb{F}'$, defined fiberwise by $\alpha_x \mapsto \pi(x)(\alpha_x, -) = \rho(\alpha_x)$, is smooth. 

To see that the corestriction $\rho: \mathbb{F} \to TM$ is also smooth, we will show that the canonical inclusion map $\operatorname{incl}: TM \hookrightarrow T''M$ assigning $X \mapsto X$ is a smooth embedding. Let $E$ be the convenient vector space modeling the fibers of $TM$. The restriction of $\operatorname{incl}$ to a fiber is the canonical evaluation map $\iota_E: E \rightarrow E''$, defined by $\iota_E(v)(l) := l(v)$ for all $v \in E$ and $l \in E'$.

By definition in the convenient setting, a bounded linear map like $\iota_E$ is a smooth embedding if a curve $c: \mathbb{R} \rightarrow E$ is smooth if and only if when the composition $\iota_E \circ c: \mathbb{R} \rightarrow E''$ is a smooth curve in $E''$. 

If $c$ is a smooth curve in $E$, then because $\iota_E$ is a bounded linear operator, the composition $\iota_E \circ c$ is smooth. Conversely, suppose that the composition $\iota_E \circ c: \mathbb{R} \rightarrow E''$ is a smooth curve. For any continuous linear functional $l \in E'$, the point-evaluation map $\operatorname{ev}_l: E'' \rightarrow \mathbb{R}$ defined by $\Phi \mapsto \Phi(l)$ is a bounded linear operator and is therefore smooth. Thus, the composition $\operatorname{ev}_l \circ (\iota_E \circ c): \mathbb{R} \rightarrow \mathbb{R}$ is a smooth curve in $\mathbb{R}$. We observe that for any $t \in \mathbb{R}$ we have
\[ 
(\operatorname{ev}_l \circ \iota_E \circ c)(t) = \iota_E(c(t))(l) = l(c(t)) = (l \circ c)(t).
\]
This shows that $l \circ c: \mathbb{R} \rightarrow \mathbb{R}$ is smooth for every continuous linear functional $l \in E'$. By \cite[Corollary 2.3]{20}, a curve is smooth if and only if it is weakly smooth (i.e., smooth when composed with any continuous linear functional). Therefore, $c: \mathbb{R} \rightarrow E$ must be a smooth curve.

 Finally, passing to the vector bundles, we can evaluate $\operatorname{incl}$ over a local trivialization domain $U \subset M$. The map reads locally as $\operatorname{id}_U \times \iota_E: U \times E \rightarrow U \times E''$. Because the identity $\operatorname{id}_U$ is trivially a diffeomorphism and $\iota_E$ is a convenient embedding, their cartesian product is a smooth embedding of convenient manifolds. Gluing these local embeddings together establishes that the global bundle morphism $\operatorname{incl}: TM \hookrightarrow T''M$ is a smooth embedding, proving that $\rho$ is smooth.
\end{remark}

\begin{lemma}\cite[Proposition 7.7]{5}
Let $(M, \mathbb F, \rho)$ be a Poisson structure. Then, there exists $[\pi, \pi]_s \in \Gamma(L^3(T'M))$ such that
\[
[\pi, \pi]_s(\alpha, \gamma, \beta) = 2(\pi(d(\pi(\alpha, \beta)), \gamma) + \pi(\beta, d\pi(\gamma, \alpha)) +\pi(\gamma, d\pi(\alpha, \beta))).
\]
\end{lemma}

\begin{remark}[Lie algebroid structure on $\Gamma(\mathbb F)$ ]
 
We note that $\Gamma(\mathbb{F})$ is a module of the following algebra:
\[
\mathcal{A} := \{f \in C^{\infty}(M) \mid df(x) \in \mathbb{F}_{x}\}.
\]
It can be given a Lie algebroid structure using the following Lie bracket (see \cite[Proposition 7.3]{5}):
\[
[\alpha, \beta]_{\rho} := \mathcal{L}_{\rho(\alpha)}\beta - \mathcal{L}_{\rho(\beta)}\alpha - d(\pi(\alpha, \beta)).
\]
where $\mathcal L_{\rho (\alpha)}$, $\mathcal L_{\rho (\beta)}$ denotes the Lie derivative with respect to the vector fields $\rho(\alpha),  \rho(\beta)$ respectively.
\end{remark}

\begin{remark}
Note that the Jacobi identity of $[-,-]_\rho$ implies that the bundle morphism $\rho: \mathbb{F} \rightarrow TM$ satisfies $\rho[\alpha, \beta] = [\rho(\alpha), \rho(\beta)]$. Hence, it induces a Lie algebra morphism $\tilde{\sigma}: \Gamma(M, \mathbb{F}) \rightarrow \Gamma(M, TM) = \mathfrak{X}(M)$. 
\end{remark}

\begin{proposition}[Hamiltonian vector fields]
Let $(M, \mathbb{F}, \rho)$ be a Poisson manifold, and let $\mathcal{A}_{\mathbb{F}}$ be the convenient Poisson subalgebra defined in Corollary \ref{Poisson algebra}. Then, for each $f \in \mathcal{A}_{\mathbb{F}}$, the Hamiltonian vector field corresponding to $f$, defined by $X_{f} := \tilde{\rho}(\widetilde{df}) = \{f, -\}$, is a smooth vector field on $M$. Furthermore, the assignment mapping 
\[X_{-}: \mathcal{A}_{\mathbb {F}} \rightarrow \Gamma(TM), \quad f \mapsto X_{f}\] 
is a bounded, and therefore smooth, linear operator.
\end{proposition}

\begin{proof}
First, we show the section smoothness for a fixed element $f \in \mathcal{A}_{\mathfrak{b}}$. By the definition of the subalgebra $\mathcal{A}_{\mathbb{F}}$, the  differential $df$ factors through a unique smooth section $\widetilde{df} \in \Gamma(\mathbb{F})$ such that $df = \iota_{\mathbb{F}} \circ \widetilde{df}$. By Remark \ref{smoooth anchor} the anchor map $\rho: \mathbb{F} \rightarrow TM$ is a smooth vector bundle morphism. let $\rho: \mathbb{F} \rightarrow TM$ be a smooth bundle morphism. For any smooth section $\tilde{\alpha} \in \Gamma(\mathbb{F})$, the composition $\rho \circ \tilde{\alpha}$ is a smooth mapping from $M$ to $TM$. Because $\rho$ is a bundle morphism, it preserves the base point, meaning $\rho(\tilde{\alpha}(x)) \in T_{x}M$. Thus, $\tilde{\rho}(\tilde{\alpha}) \in \Gamma(TM)$.
 Therefore, the composition $X_{f} = \tilde{\rho}(\widetilde{df})$ lies in $\Gamma(TM)$. Thus, $X_{f}$ is a smooth vector field.

Second, we prove the operator smoothness of the assignment $X_{-}$. The mapping $f \mapsto X_{f}$ can be written as the composition of two linear operators:
\[X_{-} = \tilde{\rho} \circ \widetilde{d}.\]
Recall that $\mathcal{A}_{\mathbb F}$ is equipped with the initial convenient vector space structure induced by the mappings $\text{incl}: \mathcal{A}_{\mathbb{F}} \rightarrow C^{\infty}(M)$ and the differential operator $\widetilde{d}: \mathcal{A}_{\mathbb{F}} \rightarrow \Gamma(\mathbb{F})$. Thus, $\widetilde{d}$ is a bounded operator. Furthermore, the induced map $\tilde{\rho}:\Gamma(\mathbb F)\to \Gamma(TM)$ is a bounded linear map between convenient vector spaces. Indeed, consider the evaluation map $\text{ev}_{\mathbb{F}}: \Gamma(\mathbb{F}) \times M \to \mathbb{F}, (\tilde{\alpha}, x) \mapsto \tilde{\alpha}(x)$, which is smooth in the convenient setting. The map $\widehat{\tilde{\rho}}: \Gamma(\mathbb{F}) \times M \to TM, (\tilde \alpha, x)\mapsto \tilde \rho(\tilde\alpha)(x)$ is exactly the composition $\rho \circ \text{ev}_{\mathbb{F}}$. Since $\rho$ is smooth by assumption, the composition $\rho \circ \text{ev}_{\mathbb{F}}$ is smooth. By the exponential law (\cite[Corollary 3.13]{20}), the corresponding map $\tilde{\rho}: \Gamma(\mathbb{F}) \to \Gamma(TM)$ is bounded.  Because the composition of bounded linear maps in convenient calculus is bounded, the operator $X_{-}$ is bounded.
\end{proof}
\begin{remark}
Note that according to the Definition \ref{3.2} of a Poisson structure, the existence of the Hamiltonian vector field for every smooth function on a manifold is not automatic. See \cite[Table 1]{15} for the existence of the Hamiltonian vector field in different existing definitions of Poisson structures in infinite-dimension settings (see \cite[section 4]{15}). 
\end{remark}

\begin{definition}[Product Poisson structure]
Let $(M, \mathbb F_1, \pi_1)$ and $(N, \mathbb F_2, \pi_2)$ be two Poisson structures. Then $M\times N$ carries a natural Poisson structure ($M\times N, \mathbb F,  \pi)$ where
\begin{enumerate}
    \item $\mathbb F$ is a subbundle of $T'M\oplus T'M$ defined as $\mathbb F(p,q):=\mathbb F_1(p)\oplus \mathbb F_2(q).$
    \item $\pi \in \Gamma (L^2_{skew}(\mathbb F))$ is defined by \[
    \pi (\alpha_1+\alpha_2, \beta_1+\beta _2):=\pi_1(\alpha_1, \beta_1)+\pi_2(\alpha_2, \beta_2)
    \]  where $\alpha_1, \beta_1 \in \mathbb F_1$ and $\alpha_2, \beta_2 \in \mathbb F_2.$
    \end{enumerate}

\end{definition}

\begin{definition}[Poisson morphism]
Let $(M, \mathbb{F}_1, \pi_1)$ and $(N, \mathbb{F}_2, \pi_2)$ be two Poisson structures on smooth manifolds $M$ and $N$, respectively. A smooth map $\Phi: M \to N$ is called a Poisson morphism if it satisfies the following two conditions:
\begin{enumerate}
    \item For each $x \in M$, the dual of the differential $(d\Phi_x)^*: T_{\Phi(x)}'N \to T_x'M$ restricts to a well-defined bounded linear map between the respective subbundles, mapping $(\mathbb{F}_2)_{\Phi(x)}$ into $(\mathbb{F}_1)_x$. Globally, this means the dual map restricts to a smooth bundle morphism
    \[
    (d\Phi)^*|_{\mathbb{F}_2}: \Phi^*\mathbb{F}_2 \to \mathbb{F}_1.
    \]
    \item The bivector fields are $\Phi$-related. That is, the pushforward of $\pi_1$ along $\Phi$ coincides with $\pi_2$, which is written as
    \[
    \pi_2(\Phi(x)) = \Phi_*\pi_1(x),
    \]
    where $\Phi_* : L^2_{\text{skew}}((\mathbb{F}_1)_x) \to L^2_{\text{skew}}((\mathbb{F}_2)_{\Phi(x)})$ is the pushforward map on the bilinear forms induced by the dual differential. That means, for all $x \in M$ and covectors $\alpha, \beta \in (\mathbb{F}_2)_{\Phi(x)}$, this compatibility condition evaluates to
    \[
    \pi_2(\Phi(x))(\alpha, \beta) = \pi_1(x)((d\Phi_x)^*\alpha, (d\Phi_x)^*\beta).
    \]
\end{enumerate}
\end{definition}
\section{Poisson Lie groups and their linearization at the identity}\label{3}

This section concerns the linearization of a Poisson Lie group $G$ endowed with a Poisson structure $(G, \mathbb{F}, \pi)$. We show that differentiation of the Poisson tensor $\pi$ at the identity element $e \in G$ induces a Lie bialgebra structure on the Lie algebra $\mathfrak{g} = T_e G$. Furthermore, if the existence of Hamiltonian vector fields is guaranteed by the Poisson structure, we show that the Lie algebra $\mathfrak{g}$ is a Lie Poisson space with respect to the subspace $\mathfrak{b} := \mathbb{F}_e \subseteq \mathfrak{g}'$.

\begin{definition}[Poisson Lie group] \label{Lie-Poisson-group}
A Lie group $G$ modeled on a convenient Lie algebra $\mathfrak{g}$ is called a convenient Poisson Lie group with a convenient Poisson structure $(G, \mathbb{F}, \pi)$ if 
\begin{enumerate}
    \item The smooth left and right actions of $G$ on $T'G$ restrict to  smooth actions on $\mathbb F:$ \[ L^{*}: G\times \mathbb F \to \mathbb F, G\times \mathbb F(g_0)\ni(g, \alpha)\mapsto L_{g}^*\alpha\in \mathbb F(gg_0) ,\]
    \[ R^{*}: G\times \mathbb F \to \mathbb F, G\times \mathbb F(g_0)\ni(g, \alpha)\mapsto R_{g}^*\alpha\in \mathbb F(g_0g).\]

\item The product map $m:G \times G\to G, (g,h)\to gh$ is a Poisson map, equivalently,
\begin{equation} \label{multiplicative}  
\pi(xy) = L^{**}_x\pi(y) + R^{**}_y\pi(x)
\end{equation}

where $L^{**}_x$ and $R^{**}_y$ are the induced left and right actions of $G$ on $L^2_{skew}(\mathbb F)$ from the restricted smooth left and right actions $L^*$ and $R^*$ of $G$ on $\mathbb F$, respectively.
\end{enumerate}
\end{definition}

\begin{definition}
A bivector  $\pi \in \Gamma(L^2_{skew}(\mathbb F))$ is called left invariant if $\pi(gh) = L_{g}^{**}\pi(h)$ and right invariant if $\pi(gh) = R_{h}^{**}\pi(g)$ for all $g, h \in G$.
\end{definition}
\begin{proposition}\label{linearization}
Let $(G, \mathbb{F}, \pi)$ be a convenient Poisson Lie group, let $\mathfrak{g}$ be the Lie algebra of $G$, and let $\mathfrak{b} = \mathbb{F}_e$. Then, the following map defines a bounded Lie algebra structure on $\mathfrak{b}$:
\begin{equation} \label{linearzitaion_of_pi}
[\alpha(e), \beta(e)] = \widetilde{d(\pi(\alpha, \beta))}(e),
\end{equation}
where $\alpha$ and $\beta$ are left-invariant (or right-invariant) sections of $\mathbb{F}$.
\end{proposition}

\begin{proof}
For any $\xi, \eta \in \mathfrak{b}$, let $\alpha = \xi^L$ and $\beta = \eta^L$ denote the unique left-invariant sections of $\mathbb{F}$ extending $\xi$ and $\eta$, respectively (so that $\alpha(e) = \xi$ and $\beta(e) = \eta$). Since the left translations of $G$ on $T'G$ restrict smoothly to $\mathbb{F}$, the extension map 
\[
E: \mathfrak{b} \to \Gamma(\mathbb{F}), \quad \xi \mapsto \xi^L
\]
is a bounded linear map. 

Because $\pi \in \Gamma(M, L_{\text{skew}}^{2}(\mathbb{F}))$ is a smooth bivector, the evaluation map 
\[
P: \Gamma(\mathbb{F}) \times \Gamma(\mathbb{F}) \to C^\infty(G), \quad (\alpha, \beta) \mapsto \pi(\alpha, \beta)
\]
is a bounded bilinear map. Furthermore, the differential operator $d: C^\infty(G) \to \Gamma(T'G)$ is a bounded linear map (see \cite[Theorem 3.18]{20}). By the definition of the Poisson structure (Definition \ref{3.2}), the differential $d(\pi(\alpha, \beta))$ factors uniquely through $\mathbb{F}$, meaning the operation
\[
\tilde d:f \mapsto \widetilde{df} \in \Gamma(\mathbb{F})
\]
is also bounded and linear on its domain. Finally, the evaluation at the identity 
\[
\operatorname{ev}_e: \Gamma(\mathbb{F}) \to \mathfrak{b}, \quad \gamma \mapsto \gamma(e)
\]
is a bounded linear map. 

The Lie bracket on $\mathfrak{b}$ is given by the composition of these bounded linear and bilinear mappings
\[
[\xi, \eta] = \operatorname{ev}_e \left( \widetilde{d} \left( P(E(\xi), E(\eta)) \right) \right).
\]
Because it is a composition of bounded maps, $[-,-]: \mathfrak{b} \times \mathfrak{b} \to \mathfrak{b}$ is a bounded bilinear map. The fact that this bracket satisfies the Jacobi identity on $\mathfrak{b}$ is a direct consequence of the Jacobi identity of the Poisson bivector $\pi$ ( condition (2) of Definition \ref{3.2}). Thus, $\mathfrak{b}$ is equipped with a bounded Lie algebra structure.
\end{proof}

Now, if $\mathfrak{b}$ is Fréchet or Silva, $[-,-]$ is a continuous bilinear map by the following two lemmas.

\begin{lemma}
If $E$ is a Fréchet space or a Silva space, then a bounded linear map from $E$ to any locally convex space is continuous.
\end{lemma}

\begin{proof}
Note that Fréchet and Silva space is bornological (\cite[8.2, Corollary 2]{28}). By (\cite[8.3]{28}), a linear map from a bornological space to a locally convex space is continuous if and only if it is bounded.

\end{proof}

\begin{lemma}
If $E$ is a Fréchet space or a Silva space, then a separately continuous bilinear map from $E \times E$ to any locally convex space is continuous.
\end{lemma}
\begin{proof}
Let $\beta: E \times E \to F$ be a separately continuous bilinear map into a locally convex space $F$. 

If $E$ is a Fréchet space, then by \cite[5.1]{28}, every separately continuous bilinear map on $E \times E$ is  continuous. 

Now, suppose $E$ is a Silva space. By definition, $E = \varinjlim E_n$ is a regular direct limit of Banach spaces $E_n$. The product $E \times E$ is also a Silva space, and its topology coincides with the regular direct limit of the products, so $E \times E = \varinjlim (E_n \times E_n)$. By the universal property of the direct limit topology, $\beta$ is continuous if and only if all of its restrictions $\beta|_{E_n \times E_n}$ are continuous. 

Since each $E_n$ is a Banach space, it is inherently a Fréchet space. Because $\beta$ is separately continuous on $E \times E$, its restriction $\beta|_{E_n \times E_n}$ is a separately continuous bilinear map on a Fréchet space. By the Fréchet case established above, each restriction $\beta|_{E_n \times E_n}$ is continuous. Therefore, $\beta$ is continuous on $E \times E$.
\end{proof}

\begin{definition}[Lie group $1$-cocycle]\label{Lie group 1-cocycle}
Let $G$ be a Smooth Lie group with a smooth action of $G$ on a convenient space $V$ which is  a Lie group homomorphism $\alpha:G\to \mathrm{GL}(V) $. Then a map $\theta: G\to V$ is called a Lie group $1$-cocycle relative to $\alpha$ if \[ \theta(gh)=\theta(g)+\alpha (g )(\theta(h)), \] for any $g,h\in G$.
\end{definition}
\begin{definition}[Lie algebra $1$-cocycle]\label{Lie algebra 1-cocycle}
The derivative $d\theta (e): \mathfrak g\to \mathfrak{gl}(V)$ satisfies\[
d\theta(e)([X,Y])=d\alpha(e)(X)(d\theta(e)(Y))-d\alpha(e)(Y)(d\theta(e)(X))
\] for any $X, Y\in \mathfrak g$. We then call $d\theta(e)$ a Lie algebra $1$-cocycle of $\mathfrak g$ relative to the action $d\alpha(e)$.
\end{definition}

\begin{proposition}[$1$-cocycle of a Poisson Lie group]\label{Poisson 1-cocycle}
Let $(G, \mathbb{F}, \pi)$ be a Poisson Lie group structure. The map $\theta: G \rightarrow L^2_{\text{skew}}(\mathfrak{b})$ defined by $\theta(x) := R_{x^{-1}}^{**}\pi(x)$ is a smooth $1$-cocycle of $G$ relative to the restricted smooth adjoint action $\operatorname{Ad}^{(2)}$ of $G$ on $L^2_{\text{skew}}(\mathfrak{b})$ given by
\[ 
\operatorname{Ad}^{(2)}: G \times L^2_{\text{skew}}(\mathfrak{b}) \rightarrow L^2_{\text{skew}}(\mathfrak{b}), \quad (g, \eta) \mapsto \left( R_{g^{-1}}^{**} \circ L_{g}^{**} \right)(\eta), 
\]
where $L_{g}^{**}$ and $R_{h}^{**}$ denote the pushforwards by left and right translations. For $\eta \in L^2_{\text{skew}}(\mathfrak{b})$ and $\alpha, \beta \in \mathfrak{b}$, these are defined explicitly as
\[ 
L_{g}^{**}\eta(\alpha, \beta) := \eta(L_{g}^*\alpha, L_{g}^*\beta), \quad R_{h}^{**}\eta(\alpha, \beta) := \eta(R_{h}^{*}\alpha, R_{h}^*\beta). 
\]
Furthermore, if the Lie algebra $\mathfrak{g}$ of $G$ is a Fréchet or a Silva space, then $\theta$ is a Bastiani smooth $1$-cocycle, and the adjoint action is Bastiani smooth.
\end{proposition}

\begin{proof}
First, we verify the algebraic $1$-cocycle condition. The left and right actions of $G$ on $\mathbb{F}$ imply that the adjoint action $\operatorname{Ad}^{(2)}$ is well-defined. By the definition of a Poisson Lie group, the bivector field $\pi$ is multiplicative, meaning $\pi(xy) = L_{x}^{**}\pi(y) + R_{y}^{**}\pi(x)$ for all $x,y \in G$. Applying $R_{(xy)^{-1}}^{**}$ to both sides yields
\begin{align*}
\theta(xy) &= R_{(xy)^{-1}}^{**} \left( L_{x}^{**}\pi(y) + R_{y}^{**}\pi(x) \right) \\
&= \left( R_{y^{-1}}^{**} \circ R_{x^{-1}}^{**} \circ L_{x}^{**} \right) \pi(y) + \left( R_{y^{-1}}^{**} \circ R_{x^{-1}}^{**} \circ R_{y}^{**} \right) \pi(x) \\
&= \left( R_{y^{-1}}^{**} \circ L_{x}^{**} \right) \left( R_{x^{-1}}^{**}\pi(y) \right) + R_{x^{-1}}^{**}\pi(x) \\
&= \operatorname{Ad}^{(2)}_{x}(\theta(y)) + \theta(x),
\end{align*}
where we used the fact that left and right translations commute. This shows that $\theta$ satisfies the $1$-cocycle condition.

Next, we show that $\theta$ is smooth in the convenient sense. Since $\mathfrak{b}$ is a convenient vector space, the space of bounded skew-symmetric bilinear forms $L^2_{\text{skew}}(\mathfrak{b})$ is also convenient. By the convenient uniform boundedness principle (\cite[Theorem 5.18]{20}), a curve into a space of bounded multilinear maps is smooth if and only if its composition with point evaluations is smooth. For fixed $\alpha, \beta \in \mathfrak{b}$, we evaluate $\theta(x)$ to obtain
\[ 
\theta(x)(\alpha, \beta) = \pi(x)(R_{x^{-1}}^*\alpha, R_{x^{-1}}^*\beta) = \pi(x)(\tilde{\alpha}(x), \tilde{\beta}(x)), 
\]
where $\tilde{\alpha}, \tilde{\beta} \in \Gamma(\mathbb{F})$ are the right-invariant sections defined by $\tilde{\alpha}(x) := R_{x^{-1}}^*\alpha$ and $\tilde{\beta}(x) := R_{x^{-1}}^*\beta$. The mappings $x \mapsto \tilde{\alpha}(x)$ are smooth because the right action $R^{*}: G \times \mathfrak{b} \rightarrow \mathbb{F}$ is assumed to be smooth. Because $\pi$ is a smooth section, its evaluation on smooth sections $\tilde{\alpha}$ and $\tilde{\beta}$ yields a smooth function from $G$ to $\mathbb{R}$. Therefore, $\theta: G \rightarrow L^2_{\text{skew}}(\mathfrak{b})$ is conveniently smooth.

It remains to show that the adjoint action $\operatorname{Ad}^{(2)}: G \times L^2_{\text{skew}}(\mathfrak{b}) \rightarrow L^2_{\text{skew}}(\mathfrak{b})$ is conveniently smooth. By the exponential law for convenient vector spaces, the joint smoothness of this map is equivalent to the smoothness of the curried map
\[ 
\check{\operatorname{Ad}}^{(2)}: G \rightarrow C^{\infty}\left(L^2_{\text{skew}}(\mathfrak{b}), L^2_{\text{skew}}(\mathfrak{b})\right). 
\]
Because the action is linear in the second variable, $\check{\operatorname{Ad}}^{(2)}$ actually takes values in the space of bounded linear operators $L\left(L^2_{\text{skew}}(\mathfrak{b}), L^2_{\text{skew}}(\mathfrak{b})\right)$. Since $L\left(L^2_{\text{skew}}(\mathfrak{b}), L^2_{\text{skew}}(\mathfrak{b})\right)$ is a $c^\infty$-closed (closed with respect to the convenient topology) linear subspace of the smooth mapping space $C^{\infty}\left(L^2_{\text{skew}}(\mathfrak{b}), L^2_{\text{skew}}(\mathfrak{b})\right)$ by \cite[Proposition 5.6]{20}, it suffices to check smoothness into this subspace. Applying the uniform boundedness principle, this map is smooth if and only if its evaluation on any fixed form $\omega \in L^2_{\text{skew}}(\mathfrak{b})$ is a smooth map from $G$ to $L^2_{\text{skew}}(\mathfrak{b})$. Applying the uniform boundedness principle a second time to the space $L^2_{\text{skew}}(\mathfrak{b})$, it suffices to show that for fixed vectors $b_1, b_2 \in \mathfrak{b}$, the scalar mapping
\[ 
g \mapsto \operatorname{Ad}^{(2)}_g(\omega)(b_1, b_2) = \omega(\operatorname{Ad}_{g^{-1}}^*b_1, \operatorname{Ad}_{g^{-1}}^*b_2) 
\]
is a smooth function on $G$. Because $G$ acts smoothly on $\mathfrak{b}$ via the restricted standard adjoint action (due to assumption (2) of Definition \ref{Lie-Poisson-group} ), the map $g \mapsto \operatorname{Ad}_{g^{-1}}^*b_i$ is smooth. The composition with the bounded, and hence smooth, bilinear form $\omega$ guarantees that this evaluation is a smooth function (\cite[Corollary 3.13]{20}). Thus, the joint action $\operatorname{Ad}^{(2)}$ is conveniently smooth.

Finally, assume $G$ is a smooth Lie group with a Fréchet or Silva Lie algebra $\mathfrak{g}$. In this setting, we note that smoothness in the convenient sense is equivalent to smoothness in the Bastiani sense. Since $\mathfrak{b}$ is a Montel space, the convenient topology of uniform convergence on bounded sets on $L^2_{\text{skew}}(\mathfrak{b})$ (for which it is Mackey complete) coincides with the topology of uniform convergence on compact sets. Therefore, to demonstrate the continuity of $\theta$, it is sufficient to show that any convergent sequence (or net) in a Fréchet (respectively, Silva) space is Mackey convergent, which holds true by \cite[4.11]{20}. This property extends to each derivative of $\theta$. Thus, all derivatives of $\theta$ are continuous, proving $\theta$ smooth in the Bastiani sense. By identical arguments, the adjoint action $\operatorname{Ad}^{(2)}$ is also Bastiani smooth.
\end{proof}

\begin{definition}[Poisson Lie group homomorphism]\label{poissonhom}\label{Cocycle compatibility}
Let $(G, \mathbb{F}_1, \pi_1)$ and $(H, \mathbb{F}_2, \pi_2)$ be two Poisson Lie groups, and let $\Theta_1$ and $\Theta_2$ be their associated $1$-cocycles. A Lie group homomorphism $\Phi: G \to H$ is called a Poisson Lie group homomorphism if it is also a Poisson morphism between $G$ and $H$. 

Let $\phi := d\Phi(e) : \mathfrak{g} \to \mathfrak{h}$ denote the induced Lie algebra homomorphism, and let $\phi^*: \mathfrak{h}' \to \mathfrak{g}'$ denote its dual. Following the conditions for a Poisson morphism, $\phi^*$ must restrict to a well-defined bounded linear map between the respective subbundles at the identity
\[
\phi^*|_{\mathfrak{b}_2}: \mathfrak{b}_2 \to \mathfrak{b}_1.
\]
The compatibility condition of the Lie group homomorphism $\Phi$ with respect to the group $1$-cocycles $\Theta_1$ and $\Theta_2$ is written as
\begin{equation}  
\Theta_2 \circ \Phi = (\phi^*|_{\mathfrak{b}_2})_* \circ \Theta_1,
\end{equation} 
where $(\phi^*|_{\mathfrak{b}_2})_*: L^2_{\text{skew}}(\mathfrak{b}_1) \to L^2_{\text{skew}}(\mathfrak{b}_2)$ is the induced pushforward map on the skew-symmetric bilinear forms. Equivalently, for all $g \in G$ and dual elements $\alpha, \beta \in \mathfrak{b}_2$, we have
\[
\Theta_2(\Phi(g))(\alpha, \beta) = \Theta_1(g)(\phi^*(\alpha), \phi^*(\beta)).
\]
\end{definition}
\subsubsection{Regular Lie group structure on a semidirect product }\cite[section 5.8]{20} \label{regularsemidirect}\\
Let $G, H$ be regular smooth regular Lie groups and $\alpha : G\times H\to H$ be a smooth action of $G$ on $H$ by automorphisms on $H$ such that $\hat{\alpha}: G\to \operatorname{Aut}(H)$ is a group homomorphism. Then the semidirect product $G \rtimes H$ carries a smooth regular Lie group structure, where the 

product map \[(g_1, h_1) \times (g_2, h_2) \mapsto (g_1\alpha(h_1)(g_2), h_1h_2),\] and the inverse map \[(g_1, h_1) \mapsto (\alpha(h_1^{-1})(g_1), h_1^{-1})\] are smooth in the convenient sense. \\

 The following lemma arises naturally from the semidirect product Lie group structure on $G \rtimes L$.

\begin{lemma}\label{4.4}
Let $G$ be a smooth Lie group with a smooth action $\alpha$ on a convenient vector space $L$. A smooth map $\theta: G \rightarrow L$ is a $1$-cocycle relative to $\alpha$ if and only if the map $\tilde{\theta}: G \rightarrow G \rtimes L$ defined by $\tilde{\theta}(g) := (g, \theta(g))$ is a smooth Lie group homomorphism.
\end{lemma}

We note that the action $\alpha$ of $G$ on $L$ induces a natural Lie algebra action $d\alpha(e)$ of the Lie algebra $\mathfrak{g}$ on $L$. The space $\mathfrak{g} \times L$ carries a semidirect product Lie algebra structure, denoted $\mathfrak{g} \rtimes L$, equipped with the following Lie bracket
\[
[(X, l_1), (Y, l_2)] := ([X, Y]_{\mathfrak{g}}, X \cdot l_2 - Y \cdot l_1)
\]
for $X, Y \in \mathfrak{g}$ and $l_1, l_2 \in L$, where $X \cdot l = d\alpha(e)(X)(l)$. The following lemma is the infinitesimal counterpart of Lemma \ref{4.4}, following directly from this Lie algebra structure.

\begin{lemma}\label{4.5}
Let $\mathfrak{g}$ be a convenient Lie algebra with a smooth action $\tilde{\alpha}: \mathfrak{g} \to \mathfrak{gl}(l)$ on a convenient vector space $l$. A bounded linear map $\delta: \mathfrak{g} \rightarrow l$ is a bounded $1$-cocycle of $\mathfrak{g}$ relative to $\tilde{\alpha}$ if and only if the map $\tilde{\delta}: \mathfrak{g} \rightarrow \mathfrak{g} \rtimes l$ defined by $\tilde{\delta}(X) := (X, \delta(X))$ is a bounded Lie algebra homomorphism.
\end{lemma}
\begin{proposition}\cite[Theorem 7.3]{22}\label{proposition 4.10}
Let $G$ and $H$ be smooth Lie groups with the Lie algebras $\mathfrak{g}$ and $\mathfrak{h}$. Such that $G$ is 1-connected, and $H$ is regular. Then if $\delta: \mathfrak{g} \rightarrow \mathfrak{h}$ is a bounded Lie algebra homomorphism, then there exists $\theta: G \rightarrow H$ a unique smooth Lie group homomorphism such that $d(\theta)(e) = \delta$.
\end{proposition}

The following lemma shows that $d(\theta)(e)$, the differentiation of $\theta$ induces a $1$-cocycle of $\mathfrak{g}$ with values in $L^2_{skew}(\mathfrak{b})$, and the dual of it induces a Lie bracket on $\mathfrak{b}$.

\begin{lemma}\label{4.7}
Let $(G, \mathbb{F}, \pi)$ be a convenient Poisson Lie group, and let $\theta$ be the corresponding smooth $1$-cocycle. Then, $\delta := d\theta(e): \mathfrak{g} \rightarrow L^2_{\text{skew}}(\mathfrak{b})$ is a smooth $1$-cocycle of $\mathfrak{g}$ with values in $L^2_{\text{skew}}(\mathfrak{b})$ relative to the adjoint action $\operatorname{ad}^{(2)}$ of $\mathfrak{g}$ on $L^2_{\text{skew}}(\mathfrak{b})$ given by
\begin{align*}
\operatorname{ad}^{(2)}&: \mathfrak{g} \to \mathfrak{gl}(L^2_{\text{skew}}(\mathfrak{b})), \\
\operatorname{ad}^{(2)}_{X} \eta (\alpha, \beta) &= \eta(\operatorname{ad}^*_{X} \alpha, \beta) + \eta(\alpha, \operatorname{ad}^*_X \beta)
\end{align*}
for $X \in \mathfrak{g}, \eta \in L^2_{\text{skew}}(\mathfrak{b})$, and $\alpha, \beta \in \mathfrak{b}$. Furthermore, the map $[-,-]_{\mathfrak{b}} := \delta'_{\mid \mathfrak{b} \times \mathfrak{b}}: \mathfrak{b} \times \mathfrak{b} \rightarrow \mathfrak{b}$ defines a Lie algebra structure on $\mathfrak{b}$ with a bounded Lie bracket morphism given as in Proposition 4.3 by
\[
[\alpha, \beta]_{\mathfrak{b}} = d(\theta(\alpha, \beta))(e)
\]
for $\alpha, \beta \in \mathfrak{b}$.
\end{lemma}

\begin{proof}
By Proposition \ref{Poisson 1-cocycle}, $\theta: G \rightarrow L^2_{\text{skew}}(\mathfrak{b})$ is a smooth $1$-cocycle relative to the action $\operatorname{Ad}^{(2)}$. According to Lemma \ref{4.4}, this holds if and only if the map $\tilde{\theta}: G \rightarrow G \rtimes L^2_{\text{skew}}(\mathfrak{b})$ defined by $g \mapsto (g, \theta(g))$ is a smooth Lie group homomorphism.

By Proposition \ref{proposition 4.10}, differentiating $\tilde{\theta}$ at the identity $e \in G$ yields a bounded Lie algebra homomorphism between the corresponding Lie algebras
\[
d\tilde{\theta}(e): \mathfrak{g} \rightarrow \mathfrak{g} \rtimes L^2_{\text{skew}}(\mathfrak{b}).
\]
Evaluating this differential at $X \in \mathfrak{g}$, we obtain $d\tilde{\theta}(e)(X) = (X, d\theta(e)(X)) = (X, \delta(X))$. By Lemma \ref{4.5}, the fact that $X \mapsto (X, \delta(X))$ is a Lie algebra homomorphism implies that $\delta = d\theta(e)$ is a smooth $1$-cocycle of $\mathfrak{g}$ relative to the infinitesimal derived action $\operatorname{ad}^{(2)} = d(\operatorname{Ad}^{(2)})(e)$. Explicit differentiation of the action $\operatorname{Ad}^{(2)}$ defined in Proposition \ref{Poisson 1-cocycle} yields the stated formula for $\operatorname{ad}^{(2)}_X$.

For the second part of the statement, we must recover the Lie bracket on $\mathfrak{b}$ established in Proposition 4.3. Let $\xi, \eta \in \mathfrak{b}$ and extend them to right-invariant sections $\alpha, \beta \in \Gamma(\mathbb{F})$ such that $\alpha(g) = R_{g^{-1}}^*\xi$ and $\beta(g) = R_{g^{-1}}^*\eta$. Consider the smooth function $f \in C^\infty(G)$ defined by the pairing
\[
f(g) := \pi(g)(\alpha(g), \beta(g)).
\]
By the definition of the $1$-cocycle $\theta(g) = R_{g^{-1}}^{**}\pi(g)$, evaluating $\theta(g)$ on the vectors $\xi, \eta$ yields exactly this function:
\[
\theta(g)(\xi, \eta) = \pi(g)(R_{g^{-1}}^*\xi, R_{g^{-1}}^*\eta) = \pi(g)(\alpha(g), \beta(g)) = f(g).
\]
Differentiating this expression at the identity $e$ in the direction of $X \in \mathfrak{g}$, we extend $X$ to a left-invariant vector field $\tilde{X}$ and apply the Lie derivative $\mathcal{L}_{\tilde{X}}$ to the scalar function $f = \pi(\alpha, \beta)$. Using the Leibniz rule for Lie derivatives, we evaluate at $e$:
\[
df(e)(X) = (\mathcal{L}_{\tilde{X}}\pi)(e)(\alpha(e), \beta(e)) + \pi(e)((\mathcal{L}_{\tilde{X}}\alpha)(e), \beta(e)) + \pi(e)(\alpha(e), (\mathcal{L}_{\tilde{X}}\beta)(e)).
\]
Because the Poisson bivector vanishes at the identity ($\pi(e) = 0$), the two terms involving the Lie derivatives of the sections vanish identically. Furthermore, the Lie derivative of a tensor at a point where the tensor vanishes is simply its ordinary differential, so $(\mathcal{L}_{\tilde{X}}\pi)(e) = d\pi(e)(X)$. Thus, we obtain
\[
df(e)(X) = d\pi(e)(X)(\xi, \eta) = d\theta(e)(X)(\xi, \eta) = \delta(X)(\xi, \eta).
\]
By the definition of the dual map $\delta': (L^2_{\text{skew}}(\mathfrak{b}))' \to \mathfrak{g}'$, this evaluation defines the bracket $[\xi, \eta]_{\mathfrak{b}}$ via the duality pairing with $X$
\[
\delta(X)(\xi, \eta) = \langle X, \delta'(\xi \wedge \eta) \rangle = \langle X, [\xi, \eta]_{\mathfrak{b}} \rangle.
\]
On the other hand, applying the factorization property of the differential of the Poisson bivector (Definition 3.2) to the function $f = \pi(\alpha, \beta)$, the differential $df$ factors through the subbundle $\mathbb{F}$ as $df = i_{\mathbb{F}} \circ \widetilde{df}$. Evaluating $df$ at $e$ against $X \in \mathfrak{g}$ gives
\[
df(e)(X) = \langle X, df(e) \rangle = \langle X, \widetilde{d(\pi(\alpha, \beta))}(e) \rangle.
\]
Equating the two expressions for $df(e)(X)$, we obtain
\[
\langle X, [\xi, \eta]_{\mathfrak{b}} \rangle = \langle X, \widetilde{d(\pi(\alpha, \beta))}(e) \rangle.
\]
Since this identity holds for all $X \in \mathfrak{g}$, the non-degeneracy of the pairing implies $[\xi, \eta]_{\mathfrak{b}} = \widetilde{d(\pi(\alpha, \beta))}(e)$.
\end{proof}
\section{Lie bialgebra, Lie-Poisson space and Manin triple}\label{4}
This section deals with the correspondence between Lie bialgebras and Manin triples. We show that if  $(\mathfrak g\oplus\mathfrak b, \mathfrak{g}, \mathfrak{b})$ is  a Manin triple, then $(\mathfrak{g}, \mathfrak{b}, \delta)$ inherits a Lie bialgebra structure. Furthermore, we show that $\mathfrak{g}$ carries a Lie-Poisson space structure with respect to $\mathfrak{b}$. Conversely, if $(\mathfrak{g}, \mathfrak{b}, \delta)$ is a Lie bialgebra structure and $\mathfrak{g}$ is a Lie-Poisson space with respect to $\mathfrak{b}$, then $(\mathfrak{g}\oplus\mathfrak b, \mathfrak{g}, \mathfrak{b})$ constitutes a Manin triple.
\begin{definition}[Convenient Lie bialgebra]
Let $\mathfrak{g}$ be a convenient Lie algebra with a smooth Lie bracket $[-,-]_{\mathfrak{g}}$. Then a convenient Lie bialgebra structure on $\mathfrak{g}$ is a pair $(\mathfrak{g}, \mathfrak{b}, \delta)$, where $\mathfrak{b} \subset \mathfrak{g}'$ separates points in  $\mathfrak{g}$ such that the smooth co-adjoint action of $\mathfrak{g}$ on $\mathfrak{g}'$ restricts to a smooth action of $\mathfrak{g}$ on $\mathfrak{b}$ and there exists a map
\[
\delta : \mathfrak{g} \rightarrow L^2_{skew}(\mathfrak{b})
\]
such that $\delta$ is a smooth $1$-cocycle of $\mathfrak{g}$ on $L^2_{skew}(\mathfrak{b})$ relative to the adjoint action $\operatorname{ad}^{(2)}$ of $\mathfrak{g}$ on $L^2_{skew}(\mathfrak{b})$ (\ref{4.7}), and
\[
[-, -]_{\mathfrak{b}} := \delta'|_{\mathfrak b \wedge\mathfrak b} : \wedge^2\mathfrak{b} \rightarrow \mathfrak{b}
\]
defines a bounded Lie algebra structure on $\mathfrak{b}$, where $\wedge$ denotes the skew-symmetric tensor product completed in the bornological tensor product topology (see \cite[5.6]{20} for the definition of bornological tensor product). 
\end{definition}

\begin{remark}[Topological Lie bialgebra structure]
If $\mathfrak{g}$ is a locally convex topological space, we say $(\mathfrak{g}, \mathfrak{b}, \delta)$ is a topological Lie bialgebra structure on $\mathfrak{g}$ if $\delta $ is a continuous $1$-cocycle of $\mathfrak{g}$ on $L^2_{skew}(\mathfrak{b})$ relative to the adjoint action $\operatorname{ad}^{(2)}$ of $\mathfrak{g}$ on $L^2(\mathfrak{b})$, and
\[
[-, -]_{\mathfrak{b}} := \delta'|_{\mathfrak b \wedge\mathfrak b} : \wedge^2\mathfrak{b} \rightarrow \mathfrak{b}
\]
defines a continuous Lie algebra bracket $\mathfrak{b}$, where $\wedge$ denotes the skew-symmetric tensor product completed in the projective tensor product topology.
\end{remark}
\begin{definition}[Linear-Poisson structure](see also \cite[Corollary 2.11]{25}
Let $V$ be a convenient space such that the subspace $\mathfrak b\subseteq V'$ separates points of $V$ and $\mathfrak b$ is a Lie algebra with a bounded Lie bracket $[-,-]$. Then, the subalgebra

$\mathcal{A}_{\mathfrak{b}}:=\{f\in C^{\infty}(\mathfrak g_{}): df \text{ factors as } i_{\mathfrak b_{}}\circ\widetilde{df} \text{ for  } \widetilde{df}\in C^\infty(V,\mathfrak b_{})\}$ has convenient vector space structure induced by the linear mappings 
$$ \mathcal A_{\mathfrak b_{}} \xrightarrow{\text{incl}} C^{\infty}(V,\mathbb R), f\mapsto f,$$
$$\quad  
 \mathcal  A_{\mathfrak b_{}} \xrightarrow {\widetilde{d}} C^\infty(V, \mathfrak b_{}), f\mapsto \widetilde{df}$$ is a Poisson algebra with the bounded Poisson bracket $\{-,-\}$ defined as
\[
\{f,g\}(v):=\langle v, [df(v), dg(v)]\rangle
\]
for $v\in V$ and $f, g\in \mathcal A_{\mathfrak b}$, where $\langle \cdot, \cdot \rangle$ denotes the non-degenerate pairing between $V$ and $\mathfrak{b}$ (obtained by restricting the natural evaluation pairing between $V$ and $V'$, since $\mathfrak{b} \subseteq V'$ separates points of $V$). 
Then $(V, \mathcal{A}_{\mathfrak{b}}, \{-, -\})$ is called a Linear Poisson structure on $V$. 
\end{definition}
\begin{remark}[Another take on Lie bialgebra structures]
Let $(\mathfrak{g}, \mathfrak{b}, \delta)$ be a Lie bialgebra structure. This is equivalent to the Linear-Poisson structure $(\mathfrak{g}, \mathcal{A}_{\mathfrak{b}}, \{-, -\})$ on $\mathfrak{g}$, where 
\[
\mathcal{A}_{\mathfrak{b}} := \{f \in C^{\infty}(\mathfrak{g}) \mid df \text{ factors as } i_{\mathfrak{b}} \circ \widetilde{df} \text{ for } \widetilde{df} \in C^\infty(\mathfrak{g}, \mathfrak{b})\}
\]
and the Poisson bracket on $\mathcal{A}_{\mathfrak{b}}$ is given by 
\[ 
\{f,g\}(X) :=\langle X, [\widetilde{df}(X), \widetilde{dg}]_{\mathfrak b}\rangle=\delta(X)(\widetilde{df}(X), \widetilde{dg}(X)) 
\] 
for $X \in \mathfrak{g}$ and $f, g \in \mathcal{A}_{\mathfrak{b}}$. Furthermore, the Poisson structure is compatible with the Lie bracket on $\mathfrak{g}$, which means
\[
\{f, g\}([X, Y]) = \{\operatorname{ad}^*_X f, \operatorname{ad}^*_X g\}(Y) - \{\operatorname{ad}^*_Y f, \operatorname{ad}^*_Y g\}(X)
\]
for $f, g \in \mathcal{A}_\mathfrak{b}$ and $X, Y \in \mathfrak{g}$, where $\operatorname{ad}^*$ is the coadjoint action of $\mathfrak{g}$ on $\mathcal{A}_{\mathfrak{b}}$ given by 
\[
\operatorname{ad}^*_X f(Y) := f([X, Y]_{\mathfrak g})
\] 
for $X, Y \in \mathfrak{g}$, $[-,-]_\mathfrak g$ is the Lie bracket on $\mathfrak g$, $[-,-]_{\mathfrak b}$ is the Lie bracket on $\mathfrak b$, and the paring $\langle\cdot,\cdot\rangle$ is as defined in Definition 5.3.
\end{remark}

\begin{definition}[Lie bialgebra morphism]
Let $(\mathfrak{g}, \mathfrak{b}_1, \delta_1)$ and $(\mathfrak{h}, \mathfrak{b}_2, \delta_2)$ be two Lie bialgebras. A Lie algebra homomorphism $\phi: \mathfrak{g} \to \mathfrak{h}$ is called a Lie bialgebra morphism if it satisfies the following two conditions:
\begin{enumerate}
    \item The dual map $\phi^*: \mathfrak{h}' \to \mathfrak{g}'$ restricts to a well-defined bounded linear map between the respective subalgebras:
    \[
    \phi^*|_{\mathfrak{b}_2}: \mathfrak{b}_2 \to \mathfrak{b}_1.
    \]
    \item The 1-cocycles $\delta_1$ and $\delta_2$ satisfy the compatibility condition
    \begin{equation}\label{cocycle compatibility}
    \delta_2 \circ \phi = (\phi^*|_{\mathfrak{b}_2})_* \circ \delta_1,
    \end{equation} 
    where $(\phi^*|_{\mathfrak{b}_2})_*: L^2_{\text{skew}}(\mathfrak{b}_1) \to L^2_{\text{skew}}(\mathfrak{b}_2)$ is the induced pushforward map on the skew-symmetric bilinear forms. Explicitly, for all $X \in \mathfrak{g}$ and dual elements $\alpha, \beta \in \mathfrak{b}_2$, we have
    \[
    \delta_2(\phi(X))(\alpha, \beta) = \delta_1(X)(\phi^*(\alpha), \phi^*(\beta)).
    \]
\end{enumerate}
\end{definition}
Banach Lie-Poisson spaces were introduced in \cite{27} and extended to an arbitrary duality pairing in \cite{32}. In this article, we extend the notion to the convenient setting.

\begin{definition}[Lie-Poisson space]
Let $\mathfrak{g}_+$ be a convenient space and $\mathfrak{g}_- \subset \mathfrak{g}'_+$ be a convenient subspace. Let $\langle \cdot, \cdot \rangle: \mathfrak{g}_+ \times \mathfrak{g}_- \rightarrow \mathbb{K}$ ($\mathbb K\in \{\mathbb R, \mathbb C\}$) be a weak nondegenerate pairing between two convenient spaces $\mathfrak{g}_+$ and $\mathfrak{g}_-$. Then $\mathfrak{g}_+$ is called a Lie-Poisson space with respect to $\mathfrak{g}_-$ if $\mathfrak{g}_-$ is a Lie algebra such that the continuous co-adjoint action of $\mathfrak{g}_-$ on $\mathfrak{g}^*_-$ restricts to $\mathfrak{g}_+$, that is, $\text{ad}^*_{\alpha}x \in \mathfrak{g}_+$, for all $x \in \mathfrak{g}_+$ and $\alpha \in \mathfrak{g}_-$, and $\text{ad}^*: \mathfrak{g}_+ \times \mathfrak{g}_- \rightarrow \mathfrak{g}_+$ is smooth.
\end{definition}

\begin{remark}
We note that a continuous linear map on a Banach space is bounded and thus smooth in the convenient sense, so we can adapt the definition of a Lie-Poisson space in  \cite[Definition 3.12]{32} in the convenient setting safely. The following results follows thereby.
\end{remark}

\begin{corollary}\cite[Theorem 3.14]{32}
If $(\mathfrak{g}_+, \mathfrak{g}_-, \delta)$ is a Lie bialgebra structure, then $\delta^*|_{\mathfrak{g}_- \wedge \mathfrak{g}_-}$ defines a Lie algebra structure on $\mathfrak{g}_-$, and $(\mathfrak{g}_+, \mathcal{A}_{\mathfrak{g}_-}, \{-,-\})$ is a Linear-Poisson structure. Here, the subalgebra
\[
\mathcal{A}_{\mathfrak{g}_-} := \{f \in C^{\infty}(\mathfrak{g}_+) \mid df \text{ factors as } i_{\mathfrak{g}_-} \circ \widetilde{df} \text{ for } \widetilde{df} \in C^\infty(\mathfrak{g}_+, \mathfrak{g}_-)\}
\]
equipped with the convenient vector space structure induced by the linear mappings 
\[
\operatorname{incl}: \mathcal{A}_{\mathfrak{g}_-} \to C^{\infty}(\mathfrak{g}_+, \mathbb{R}), \quad f \mapsto f,
\]
\[
\widetilde{d}: \mathcal{A}_{\mathfrak{g}_-} \to C^\infty(\mathfrak{g}_+, \mathfrak{g}_-), \quad f \mapsto \widetilde{df},
\]
is a Poisson algebra with respect to the bounded Poisson bracket 
\[
\{-,-\}: \mathcal{A}_{\mathfrak{g}_-} \times \mathcal{A}_{\mathfrak{g}_-} \to \mathcal{A}_{\mathfrak{g}_-}, \quad \{f, g\}(x) := \langle x, [\widetilde{df}_x, \widetilde{dg}_x]_{\mathfrak{g}_-} \rangle
\]
for $f, g \in \mathcal{A}_{\mathfrak{g}_-}$ and $x \in \mathfrak{g}_+$. In this expression, $[-,-]_{\mathfrak{g}_-}$ is the Lie bracket on $\mathfrak{g}_-$, and $\langle \cdot, \cdot \rangle$ denotes the non-degenerate pairing between $\mathfrak{g}_+$ and $\mathfrak{g}_-$ (obtained by restricting the natural evaluation pairing between $\mathfrak{g}_+$ and $\mathfrak{g}_+'$, since $\mathfrak{g}_- \subseteq \mathfrak{g}_+'$ separates points of $\mathfrak{g}_+$). 

The fact that $\mathfrak{g}_+$ is a Lie-Poisson space with respect to $\mathfrak{g}_-$ is equivalent to the property that for each $f \in \mathcal{A}_{\mathfrak{g}_-}$, $\{f, \cdot\} \in \Gamma(T\mathfrak{g}_+)$ defines a Hamiltonian vector field.
\end{corollary}

\begin{corollary}{\cite[Theorem 5.13]{32}}
Let $(G, \mathbb{F}_{\mathfrak{b}}, \pi)$ be a Poisson Lie group structure. Then $\mathfrak{g}$ is a Lie Poisson space with respect to $\mathfrak{b} := \mathbb{F}_e$.
\end{corollary}

\begin{definition}[Manin triple]
A triple $(\mathfrak{g}, \mathfrak{g}_+, \mathfrak{g}_-)$ of convenient Lie algebras is called a Manin triple if $\mathfrak{g} = \mathfrak{g}_+ \oplus \mathfrak{g}_-$ and there exists a weakly non-degenerate, symmetric, adjoint-invariant bilinear form $\langle -, - \rangle: \mathfrak g \times \mathfrak g\to \mathbb K$ ($\mathbb K \in \{\mathbb R, \mathbb C\}$) on $\mathfrak{g}$ such that both $\mathfrak{g}_+$ and $\mathfrak{g}_-$ are isotropic, and the map
\[
\mathfrak{g}_+ \times \mathfrak{g}_- \rightarrow \mathfrak{g}_- \subseteq \mathfrak{g}_+', \quad (\alpha, \beta) \mapsto \operatorname{ad}^*_{\alpha}(\beta) = [\alpha, \beta]
\]
is bounded with respect to the topology of $\mathfrak{g}_-$.
\end{definition}

 Now, from Proposition \ref{4.7} and Lemma \ref{proposition 4.10}, we conclude that the Lie algebra of a convenient Poisson Lie group $(G, \mathbb{F}, \pi)$ carries a smooth (bounded) bialgebra structure $(\mathfrak{g}, \mathfrak{b}, \delta)$, where $\mathfrak b:=\mathbb F_{e}.$
\begin{theorem}\label{D-0}
Let $(G, \mathbb{F}, \pi)$ be a convenient Poisson Lie group with Lie algebra $\mathfrak{g}$, and let $\mathfrak{b} = \mathbb{F}_e$. Then, there exists a convenient Lie bialgebra structure $(\mathfrak{g}, \mathfrak{b}, \delta)$ on $\mathfrak{g}$, where the cocycle $\delta$ is defined by
\[
\delta(X)(\alpha, \beta) := d(\pi(\tilde{\alpha}, \tilde{\beta}))(e)(X)
\]
for $X \in \mathfrak{g}$ and $\alpha, \beta \in \mathfrak{b}$, and where $\tilde{\alpha}, \tilde{\beta} \in \Gamma(\mathbb{F})$ are the right-invariant sections extending $\alpha$ and $\beta$, given by $\tilde{\alpha}(g) := R_{g^{-1}*}\alpha$ and $\tilde{\beta}(g) := R_{g^{-1}*}\beta$. Moreover, if $\mathfrak{g}$ and $\mathfrak{b}$ are Fréchet or Silva spaces, then $(\mathfrak{g}, \mathfrak{b}, \delta)$ is a topological Lie bialgebra.
\end{theorem}

\begin{theorem}\cite[Theorem 4.9]{32}
Let $(\mathfrak{a}, \mathfrak{g}, \mathfrak{b})$ be a Manin triple, then $(\mathfrak{g}, \mathfrak{b}, \delta)$ carries a Lie bialgebra structure, where $\delta : \mathfrak{g} \rightarrow \wedge^2\mathfrak{b}$ is the restriction of $[-,-]'_{\mathfrak{b}}$ on $\mathfrak{g}$:
\[
\delta(X)(\alpha, \beta) = X([\alpha, \beta]_{\mathfrak{b}}).
\]
Furthermore, $\mathfrak{g}$ carries a Lie-Poisson space structure with respect to $\mathfrak{b}$. 

Conversely, if $(\mathfrak{g}, \mathfrak{b}, \delta)$ is a Lie bialgebra structure and $\mathfrak{g}$ is a Lie-Poisson space with respect to $\mathfrak{b}$, then $(\mathfrak{a}, \mathfrak{g}, \mathfrak{b})$ is a Manin triple for the natural non-degenerate pairing
\[
\langle \cdot, \cdot \rangle_{\mathfrak{a}}: \mathfrak{a} \times \mathfrak{a} \rightarrow \mathbb{K}, \qquad (\mathbb K\in \{\mathbb R, \mathbb C\})
\]
\[
(x, \alpha) \times (y, \beta) \mapsto \langle x, \beta \rangle + \langle y, \alpha \rangle,
\]
with a Lie bracket on $\mathfrak{a}$ given by
\[
[\cdot, \cdot] : \mathfrak{a} \times \mathfrak{a} \rightarrow \mathfrak{a} = \mathfrak{g} \oplus \mathfrak{b}
\]
\[
((x, \alpha), (y, \beta)) \mapsto ([x, y]_{\mathfrak{g}_+} + \operatorname {ad}^*_{\beta}(x) - \operatorname{ad}^*_{\alpha}(y), [\alpha, \beta]_{\mathfrak{g}_-} + \operatorname{ad}^*_{y}(\alpha) - \operatorname{ad}^*_{x}(\beta)).
\]
\end{theorem}

\begin{proof}
In \cite[Theorem 4.9]{32} (see also \cite[Theorem 2.3]{32} ), the theorem is proved in the Banach setting when $\mathfrak{a}, \mathfrak{g}, \mathfrak{b}$ are Banach spaces. But, since our definitions of a Lie bialgebra and a Manin triple in the convenient setting are an adaptation of the definitions in \cite{32}, and as a linear map is smooth in the convenient sense if and only if it is bounded, which is the same as continuous in the Banach setting; whence the same proof also works in our setting.
\end{proof}
\section{Manin triples from weight space decompositions}\label{5}

In this section, we construct examples of Manin triples starting from a Lie algebra that admits a weight space decomposition such that the zero-modes are abelian with respect to the Lie bracket. Furthermore, we provide the weight space decomposition for the spaces of smooth and analytic vectors of a $C^*$-algebra equipped with a strongly continuous circle action and a torus-invariant, faithful, tracial state.
\begin{definition}[Weight space decomposition]\label{Weight space decomposition}
Let $\mathfrak{g}$ be a Lie algebra with a nondegenerate symmetric invariant bilinear form $\langle\langle -, - \rangle\rangle : \mathfrak{g} \times \mathfrak{g} \rightarrow \mathbb{K}$ ($\mathbb K\in \{\mathbb R, \mathbb C\}$) is said to have a weight space decomposition if $\mathfrak{g}$ has the following topological direct sum:
\[
\mathfrak{g} = \mathfrak{h} \oplus \bigoplus_{\alpha \in \mathbb{Z} - \{0\}} \mathfrak{g}_{\alpha}
\]
with
\begin{align*}
[\mathfrak{h}, \mathfrak{h}] &= 0, \\
[\mathfrak{g}_{\alpha}, \mathfrak{g}_{\beta}] &\subseteq \mathfrak{g}_{\alpha+\beta}.
\end{align*}
\end{definition}

\begin{theorem}\label{Weight space decomposition-manin}
Let $\mathfrak{g}$ be a Lie algebra with a weight space decomposition
\[
\mathfrak{g} = \mathfrak{h} \oplus \bigoplus_{\alpha \in \mathbb{Z} - \{0\}} \mathfrak{g}_{\alpha}.
\]
Let us define
\begin{align*}
\mathfrak{p}_1 &:= \{(a, a) : a \in \mathfrak{g}\} \quad \text{(diagonal subalgebra)}, \\
\mathfrak{p}_2 &:= \{(x, y) : x \in \mathfrak{g}_{\ge 0}, y \in \mathfrak{g}_{\le 0}, \pi_0(x) + \pi_0(y) = 0\},
\end{align*}
where
\begin{align*}
\mathfrak{g}_{\ge 0} &:= \mathfrak{h} \oplus \bigoplus_{m > 0} \mathfrak{g}_m, \\
\mathfrak{g}_{\le 0} &:= \mathfrak{h} \oplus \bigoplus_{m < 0} \mathfrak{g}_m,
\end{align*}
and $\pi_0: \mathfrak{g} \rightarrow \mathfrak{h}$ is the continuous projection of $\mathfrak{g}$ on the subalgebra $\mathfrak{h}$. Then $(\mathfrak{p}, \mathfrak{p}_1, \mathfrak{p}_2)$ is a Manin triple.
\end{theorem}

\begin{proof}
We refer to Section \ref{Construction of Manin triples from $C^*$-algebras}, which is a special case of this theorem, but the proof follows along similar arguments.
\end{proof}

\subsection{Smooth vectors and their weight space decomposition}
Let $A$ be a unital $C^*$-algebra endowed with a strongly continuous circle action \[\alpha : \mathbb{S}^1 \rightarrow \operatorname{Aut}(A),  \quad t \mapsto \alpha_t,\] and let $\omega : A \rightarrow \mathbb{C}$ be a faithful tracial $\alpha$-invariant state  (see \cite[Section 2.3]{4} for the definition of a state). Recall that $A$  is called  $\alpha$-invariant if \[
\omega(\alpha_t(a)) = \omega(a) \quad \text{for all } t \in \mathbb{S}^1, a \in A.
\] 
and $\omega$ is called tracial if
\[
\omega(ab) = \omega(ba)  \quad \text{for all } a,b\in A.
\]
We further define the weight subspaces $A_m$ of $A$ given by
\[
A_m := \{a \in A : \alpha_t(a) = e^{itm}a \ \forall t \in \mathbb{S}^1\}, \quad m \in \mathbb{Z}.
\]
Since $\mathbb{S}^1$ is a compact group, then by \cite{18}, the sum
$
\bigoplus_{m \in \mathbb{Z}} A_m$ is dense in $A$.

\begin{definition}
The space of smooth vectors is defined by
\[
A_{\infty} := \{a \in A : \text{the map } t \mapsto \alpha_t(a) \text{ is } C^{\infty}(\mathbb{R}, A)\}.
\]
For each $k \in \mathbb{N}_0$, define the Fréchet spaces
\[
A^{(k)} := \{a \in A_{\infty} : \|a\|_{(k)} < \infty\}
\]
where the seminorm is given by
\[
\|a\|_{(k)} := \max_{0 \le j \le k} \sup_{t \in [0, 2\pi]} \left\| \frac{d^j}{dt^j}\alpha_t(a) \right\|_A.
\]
Note that $A_{\infty} = \bigcap_{k \in \mathbb{N}_0} A^{(k)}$, and the system of seminorms $(\| \cdot \|_{(k)})_{k \in \mathbb{N}_0}$ endow $A_{\infty}$ with a Fréchet space structure.
\end{definition}

\begin{lemma}
The space $A_{\infty}$ is a Fréchet space.
\end{lemma}

\subsubsection{Weight space decomposition}
We note that $A_m \subseteq A_{\infty}$ for each $n \in \mathbb{Z}$. By  \cite[Proposition 3.3]{18}, the Fourier projection onto $A_m$ is defined by
\[
P_m(a) := \frac{1}{2\pi} \int_0^{2\pi} e^{-imt}\alpha_t(a) \, dt.
\]
The smooth weight spaces decompose $A_{\infty}$ as a topological direct sum:
\[
A_{\infty} = \bigoplus_{m \in \mathbb{Z}} A_{\infty}^m.
\]

\begin{proof}
For $a \in A_{\infty}$, the Fourier integral is well-defined and exists since $t \mapsto \alpha_t(a)$ is smooth. The element $P_m(a)$ satisfies $\alpha_t(P_m(a)) = e^{imt} P_m(a)$:
\[
\alpha_t(P_m(a)) = \frac{1}{2\pi} \int_0^{2\pi} e^{-ims}\alpha_{t+s}(a) \, ds = e^{imt} \frac{1}{2\pi} \int_0^{2\pi} e^{-ims}\alpha_s(a) \, ds = e^{imt} P_m(a),
\]
which uses the periodicity of $\alpha$ on the circle. Thus $P_m(a) \in A_{\infty}^m$. The decomposition $a = \sum_{m \in \mathbb{Z}} P_m(a)$ (with convergence in $A_{\infty}$, proved in section \ref{Partial sum convergence for smooth vectors} below) follows by the Fourier decomposition of the periodic function $t \mapsto \alpha_t(a)$. Uniqueness is immediate from the orthogonality of the weight spaces: if $a = \sum_m x_m$ with $x_m \in A_{\infty}^m$, then pairing with $e^{-imt}$ recovers $x_m = P_m(a)$ uniquely. The direct sum is topological because the Fourier projections $P_m : A_{\infty} \rightarrow A_{\infty}^m$ are continuous with respect to the Fréchet topology (\cite[ Proposition 3.3]{18}).
\end{proof}

\subsubsection{Rapid decay of Fourier coefficients for smooth vectors}
\begin{lemma}[Rapid decay for smooth vectors]\label{Rapid decay for smooth vectors}
For each $a \in A_{\infty}$ and each $k \in \mathbb{N}_0$, there exists a constant $C(k)_a > 0$ such that
\[
\|P_m(a)\|_A \le \frac{C(k)_a}{|m|^k} \quad \text{for all } m \neq 0.
\]
Equivalently, with $\|P_0(a)\|_A$ determined by $P_0(a)$, the Fourier coefficients satisfy
\[
\|P_m(a)\|_A \le \|a\|_{(k)}|m|^{-k}
\]
for each $k \in \mathbb{N}_0$ and $m \neq 0$.
\end{lemma}

\begin{proof}
Let $a \in A_{\infty}$ be fixed. The Fourier coefficient is
\[
P_m(a) = \frac{1}{2\pi} \int_0^{2\pi} e^{-imt}\alpha_t(a) \, dt.
\]
For $m \neq 0$, perform integration by parts $k$ times. Set $u = e^{-imt}$ and $dv = \alpha_t(a) \, dt$. Then $du = -im \cdot e^{-imt}dt$. By the smoothness of $t \mapsto \alpha_t(a)$, the map $v(t) := \int_0^t \alpha_s(a) \, ds$ is well-defined and differentiable with $\frac{dv}{dt} = \alpha_t(a)$. After $k$ successive integrations by parts, by the periodicity of $\alpha_t(a)$, the boundary terms vanish. We obtain
\[
\int_0^{2\pi} e^{-imt}\alpha_t(a) \, dt = \frac{1}{(im)^k} \int_0^{2\pi} e^{-imt} \frac{d^k}{dt^k}\alpha_t(a) \, dt.
\]
Thus,
\[
P_m(a) = \frac{1}{2\pi(im)^k} \int_0^{2\pi} e^{-imt} \frac{d^k}{dt^k}\alpha_t(a) \, dt.
\]
Taking norms yields:
\[
\|P_m(a)\|_A \le \frac{1}{2\pi|m|^k} \int_0^{2\pi} \left\| \frac{d^k}{dt^k}\alpha_t(a) \right\|_A dt \le \frac{1}{|m|^k} \max_{t \in [0, 2\pi]} \left\| \frac{d^k}{dt^k}\alpha_t(a) \right\|_A.
\]
By definition of the seminorm $\| \cdot \|_{(k)}$:
\[
\|P_m(a)\|_A \le \frac{1}{|m|^k}\|a\|_{(k)}.
\]
\end{proof}

\subsubsection{Partial sum convergence for smooth vectors}\label{Partial sum convergence for smooth vectors}
For $a \in A_{\infty}$, the partial sums
\[
S_N(a) := \sum_{|m| \le N} P_m(a)
\]
converge to $a$ in $A_{\infty}$ with polynomial rate. that means, for each $k \in \mathbb{N}_0$,
\[
\|a - S_N(a)\|_{(k)} \le C(k)_a (k-1) N^{k-1} \quad \text{for } k \ge 2.
\]

\begin{proof}
The error in the partial sum is
\[
a - S_N(a) = \sum_{|m| > N} P_m(a).
\]
By the triangle inequality and Lemma \ref{Rapid decay for smooth vectors},
\[
\|a - S_N(a)\|_A \le \sum_{|m| > N} \|P_m(a)\|_A \le \sum_{|m| > N} \frac{\|a\|_{(k)}}{|m|^k} = \|a\|_{(k)} \sum_{|m| > N} \frac{1}{|m|^k}.
\]
For $k \ge 2$, the tail sum over integers can be bounded using integral comparison:
\[
\sum_{|m| > N} \frac{1}{|m|^k} = 2 \sum_{m = N+1}^{\infty} \frac{1}{m^k} \le 2 \int_N^{\infty} \frac{1}{x^k} dx = \frac{2}{(k-1)N^{k-1}}.
\]
Thus
\[
\|a - S_N(a)\|_A \le \|a\|_{(k)} \cdot \frac{2}{(k-1)N^{k-1}} \rightarrow 0 \quad \text{as } N \rightarrow \infty.
\]
For the Fréchet seminorm $\| \cdot \|_{(k)}$, we also need a bound on higher order derivatives. Note that
\[
\frac{d^j}{dt^j}\alpha_t(S_N(a)) = \sum_{|m| \le N} \frac{d^j}{dt^j}P_m(a)e^{imt}
\]
for $j = 0, 1, \dots, k$. The Fourier coefficients $P_m(a)$ satisfy rapid decay (Lemma \ref{Rapid decay for smooth vectors}), which is preserved under differentiation (multiplication by powers of $m$). Therefore,
\[
\left\| \frac{d^j}{dt^j}(\alpha_t(a) - \alpha_t(S_N(a))) \right\|_A \le \sum_{|m| > N} |m|^j \|P_m(a)\|_A \le \sum_{|m| > N} |m|^j \cdot \|a\|_{(k)}|m|^{-k} = \|a\|_{(k)} \sum_{|m| > N} |m|^{j-k}.
\]
For $j \le k$, we have $j - k \le 0$, so $|m|^{j-k} \le |m|^{-(k-j)}$. Thus,
\[
\sum_{|m| > N} |m|^{j-k} = \sum_{|m| > N} \frac{1}{|m|^{k-j}} \le \frac{2}{(k-j-1)N^{k-j-1}}.
\]
Therefore,
\[
\max_{0 \le j \le k} \sup_{t \in [0, 2\pi]} \left\| \frac{d^j}{dt^j}(\alpha_t(a) - \alpha_t(S_N(a))) \right\|_A \rightarrow 0 \quad \text{as } N \rightarrow \infty.
\]
Thus, $S_N(a) \rightarrow a$ in the Fréchet topology of $A_{\infty}$.
\end{proof}

\subsection{Analytic vectors and their weight decomposition}
Let $A$ be a unital $C^*$-algebra endowed with a strongly continuous circle action $\alpha : \mathbb{S}^1 \rightarrow \text{Aut}(A)$.

\begin{definition}[Space of analytic vectors]
The space of analytic vectors $A_{\omega}$ is defined as the set of elements $a \in A$ such that the map $t \mapsto \alpha_t(a)$ extends to a holomorphic map $f_a: S_r \rightarrow A$ for some $r > 0$, where $S_r := \{z \in \mathbb{C} : |\text{Im}(z)| < r\}$ is the open horizontal strip of width $2r$.
\end{definition}

For each $n \in \mathbb{N}$, we define the normed space $B_n$ as the set of elements extending to the strip $S_{1/n}$:
\[
B_n := \{a \in A_{\omega} : \|a\|_n < \infty\}
\]
where the norm is defined by
\[
\|a\|_n := \sum_{k=0}^n \sup_{z \in S_{1/n}} \left\| \frac{d^k}{dz^k}\alpha_z(a) \right\|_A.
\]

\begin{lemma}\label{5.7}
Each space $(B_n, \| \cdot \|_n)$ is a Banach space.
\end{lemma}

\begin{proof}
Let $\{a_j\}_{j \in \mathbb{N}}$ be a Cauchy sequence in $B_n$. From the definition of norm, for every $0 \le k \le n$, the sequence of functions $f^{(k)}_j(z) := \frac{d^k}{dz^k}\alpha_z(a_j)$ is Cauchy with respect to the uniform norm on $S_{1/n}$. In particular, the sequence of holomorphic functions $f_j(z) = \alpha_z(a_j)$ converges uniformly on $S_{1/n}$ to a function $f_{\infty} : S_{1/n} \rightarrow A$. Since uniform limits of holomorphic functions are holomorphic, $f_{\infty}$ is holomorphic. Furthermore, uniform convergence implies uniform convergence of derivatives, so $f^{(k)}_j$ converges to $f^{(k)}_{\infty}$ uniformly on $S_{1/n}$. Since $A$ is complete, the sequence $\{a_j\} = \{f_j(0)\}$ converges to a limit $a_{\infty} \in A$. By the strong continuity of $\alpha$, we have $\alpha_t(a_{\infty}) = \lim_{j \rightarrow \infty} \alpha_t(a_j) = f_{\infty}(t)$ for $t \in \mathbb{R}$. Thus $f_{\infty}$ is the unique analytic extension of the group action $\alpha_t(a_{\infty})$, so $a_{\infty} \in A_{\omega}$. Since the convergence is uniform on $S_{1/n}$, $a_{\infty}$ is in $B_n$ and $\|a_j - a_{\infty}\|_n \rightarrow 0$.
\end{proof}

\textit{Continuous compact injections $B_n \hookrightarrow B_{n+1}$:}
\begin{lemma}\label{5.8}
For each $n \in \mathbb{N}$, the canonical inclusion map $i_n: B_n \hookrightarrow B_{n+1}$ is a continuous linear injection.
\end{lemma}

\begin{proof}
Let $a \in B_n$. Since $S_{1/(n+1)} \subset S_{1/n}$, for any $k \le n$, the supremum of the $k$-th derivative over the smaller strip $S_{1/(n+1)}$ is bounded by the supremum over $S_{1/n}$. For the derivative of order $n+1$, we use Cauchy’s integral estimates. Let $z \in S_{1/(n+1)}$. We choose $\epsilon = \frac{1}{n} - \frac{1}{n+1} > 0$. The closed disk $D(z, \epsilon)$ is contained in $S_{1/n}$. Thus,
\[
\left\| \frac{d^{n+1}}{dz^{n+1}}\alpha_z(a) \right\|_A \le \frac{(n+1)!}{\epsilon^{n+1}} \sup_{\zeta \in S_{1/n}} \|\alpha_{\zeta}(a)\|_A.
\]
Since $\|\alpha_{\zeta}(a)\|_A \le \|a\|_n$, the supremum on $S_{1/(n+1)}$ is finite and bounded by $C_n\|a\|_n$. Summing the bounds for all derivatives $k = 0$ to $n+1$ shows that $\|a\|_{n+1} \le C_n\|a\|_n$ for some constant $C_n$.
\end{proof}

\begin{lemma}\label{5.9}
The inclusion map $i_n: B_n \rightarrow B_{n+1}$ is a compact operator.
\end{lemma}

\begin{proof}
Let $\{a_j\}_{j \in \mathbb{N}}$ be a bounded sequence in $B_n$, so $\|a_j\|_n \le M$. The family of functions $\mathcal{F} = \{f_j(z) := \alpha_z(a_j)\}$ is uniformly bounded on the strip $S_{1/n}$. We apply the Ascoli-Arzelà theorem. Since $\alpha$ is an action of the circle group $\mathbb{S}^1$, the functions $f_j(z)$ are $2\pi$-periodic in the real part of $z$. Thus, it suffices to consider the compact rectangle:
\[
K := \left\{z \in \mathbb{C} : 0 \le \text{Re}(z) \le 2\pi, |\text{Im}(z)| \le \frac{1}{n+1} \right\}.
\]
The uniform boundedness of $\mathcal{F}$ on $S_{1/n}$ implies the uniform boundedness of the family $\mathcal{F}$ and its derivatives $\mathcal{F}'$ on $K$ (by Cauchy’s estimates). The family $\mathcal{F}$ is therefore equicontinuous on $K$. By the Ascoli-Arzelà theorem, there exists a subsequence $f_{i_j}$ that converges uniformly on $K$ to a holomorphic function $f_{\infty}$. Due to the $2\pi$-periodicity, uniform convergence on $K$ implies uniform convergence on the infinite strip $S_{1/(n+1)}$. Uniform convergence on the open strip $S_{1/(n+1)}$ implies the uniform convergence of all derivatives $\frac{d^k}{dz^k}f_j$ on $S_{1/(n+1)}$ for $k = 0, \dots, n+1$. Thus, $\{a_{i_j}\}$ converges in the norm of $B_{n+1}$.
\end{proof}

\subsubsection{$A_{\omega}$ is a Silva space}

\begin{theorem}
The space of analytic vectors $A_{\omega}$ is a Silva space, given by the strict inductive limit:
\[
A_{\omega} = \varinjlim B_n.
\]
\end{theorem}

\begin{proof}
The space $A_{\omega}$ is clearly the union of the $B_n$. By Lemmas \ref{5.7}, \ref{5.8}, and \ref{5.9}, $A_{\omega}$ is a countable inductive limit of Banach spaces with compact structural maps, which by definition is a Silva space.
\end{proof}

\subsubsection{Weight space decomposition}
The analytic weight space $A_{\omega}^m$ is the eigenspace of the action $\alpha$ corresponding to the eigenvalue $e^{imt}$:
\[
A_{\omega}^m := \{a \in A_{\omega} : \alpha_t(a) = e^{imt}a\}.
\]
The Fourier projection $P_m : A_{\omega} \rightarrow A_{\omega}^m$ is defined by
\[
P_m(a) = \frac{1}{2\pi} \int_0^{2\pi} e^{-imt}\alpha_t(a) \, dt.
\]
The space $A_{\omega}$ decomposes as a topological direct sum:
\[
A_{\omega} = \bigoplus_{m \in \mathbb{Z}} A_{\omega}^m.
\]

\begin{lemma}[Rapid decay of Fourier coefficients]\label{5.11}
For each $a \in A_{\omega}$, there exist constants $C_a > 0$ and $\rho_a > 1$ such that the Fourier coefficients satisfy:
\[
\|P_m(a)\|_A \le C_a \rho_a^{-|m|}.
\]
\end{lemma}

\begin{proof}
Since $a \in A_{\omega}$, there exists $r_a > 0$ such that $f(z) = \alpha_z(a)$ is holomorphic on $S_{r_a}$. By Cauchy’s theorem and the $2\pi$-periodicity of $f(z)$, we can shift the integration contour from the real axis to any path $t \mapsto t + iy$ for $|y| < r_a$:
\[
P_m(a) = \frac{1}{2\pi} \int_0^{2\pi} e^{-im(t+iy)}f(t+iy) \, dt = \frac{e^{my}}{2\pi} \int_0^{2\pi} e^{-imt}f(t+iy) \, dt.
\]
To ensure decay, we choose $y$ such that $my$ is negative. Let $r'$ be a number such that $0 < r' < r_a$. Let $M = \sup_{z \in S_{r_a}}\|f(z)\|_A$.
\begin{enumerate}
    \item If $m \ge 0$: We choose $y = -r'$. Then $e^{my} = e^{-mr'}$.
    \item If $m < 0$: We choose $y = r'$. Then $e^{my} = e^{mr'} = e^{-|m|r'}$.
\end{enumerate}
In both cases, we take the modulus:
\[
\|P_m(a)\|_A \le \frac{e^{-|m|r'}}{2\pi} \int_0^{2\pi} \|f(t \pm ir')\|_A \, dt \le e^{-|m|r'}M.
\]
Setting $C_a = M$ and $\rho_a = e^{r'} > 1$ yields the desired exponential decay.
\end{proof}

\subsubsection{Convergence of partial sums}
The partial sums $S_N(a) := \sum_{|m| \le N} P_m(a)$ converge to $a$ in the Silva topology of $A_{\omega}$.

\begin{proof}
Since $a \in A_{\omega}$, there exists an index $k$ such that $a \in B_k$. We show that the sequence $S_N(a)$ converges to $a$ in the Banach space norm $\| \cdot \|_k$. The error is given by the tail of the series:
\[
\|a - S_N(a)\|_k = \left\| \sum_{|m| > N} P_m(a) \right\|_k.
\]
By the definition of the $\| \cdot \|_k$ norm, this is:
\[
\|a - S_N(a)\|_k = \sum_{j=0}^k \sup_{z \in S_{1/k}} \left\| \frac{d^j}{dz^j} \sum_{|m| > N} e^{imz}P_m(a) \right\|_A.
\]
Since the terms converge absolutely and uniformly on $S_{1/k}$, we interchange the derivative and summation:
\[
\|a - S_N(a)\|_k \le \sum_{j=0}^k \sum_{|m| > N} \sup_{z \in S_{1/k}} \| (im)^j e^{imz} P_m(a) \|_A.
\]
Using Lemma \ref{5.11}, we have $\|P_m(a)\|_A \le C_a \rho_a^{-|m|}$ and $|e^{imz}| \le e^{|m| \cdot |\text{Im}(z)|} \le e^{|m|/k}$.
\[
\|a - S_N(a)\|_k \le C_a \sum_{j=0}^k \sum_{|m| > N} |m|^j e^{|m|/k} \rho_a^{-|m|}.
\]
We require the term $e^{|m|/k}\rho_a^{-|m|} = (e^{1/k}\rho_a^{-1})^{|m|} = (\rho')^{-|m|}$ to be a decay factor. Since $\rho_a > 1$, we can always choose $k$ large enough such that $\rho_a > e^{1/k}$, ensuring $\rho' > 1$. The polynomial factor $|m|^j$ is overcome by the exponential decay of $\rho^{-|m|}$. Thus, the series converges absolutely in $B_k$, and the tail $\sum_{|m| > N}$ converges to 0 as $N \rightarrow \infty$. This proves convergence in the $B_k$ norm and, consequently, in the Silva topology of $A_{\omega}$.
\end{proof}

\subsection{Construction of Manin triples from $C^*$-algebras}\label{Construction of Manin triples from $C^*$-algebras}

\begin{theorem}
Let $A$ be a unital $C^*$-algebra equipped with a strongly continuous circle action $\alpha: \mathbb{S}^1 \rightarrow \text{Aut}(A)$ and a faithful, tracial, $\alpha$-invariant state $\omega: A \rightarrow \mathbb{C}$. Let $\pi_0: A_{\infty} \rightarrow A_{\infty}^0$ be the projection onto the zero-weight elements, given by
\[
\pi_0(a) := \frac{1}{2\pi} \int_0^{2\pi} \alpha_t(a) \, dt.
\]
Consider the Lie algebra $\mathfrak{g} := A_{\infty} \oplus A_{\infty}$ equipped with the componentwise Lie bracket $[(x_1, x_2), (y_1, y_2)] := ([x_1, y_1], [x_2, y_2])$ and the symmetric bilinear form
\[
\langle\langle (a, a'), (x, y) \rangle\rangle := \omega(ax) - \omega(a'y).
\]
Define the Lie subalgebras $\mathfrak{p}_1, \mathfrak{p}_2 \subset \mathfrak{g}$ by
\begin{align*}
\mathfrak{p}_1 &:= \{(a, a) \mid a \in A_{\infty}\}, \\
\mathfrak{p}_2 &:= \{(x, y) \mid x \in A_{\infty}^{\ge 0}, \, y \in A_{\infty}^{\le 0}, \, \pi_0(x) + \pi_0(y) = 0\},
\end{align*}
where $A_{\infty}^{\ge 0} := \bigoplus_{m \ge 0} A_{\infty}^m$ and $A_{\infty}^{\le 0} := \bigoplus_{m \le 0} A_{\infty}^m$. 

If the zero-weight subspace $A_{\infty}^0$ is abelian with respect to the canonical commutator bracket (i.e., $[a, b] = 0$ for all $a, b \in A_{\infty}^0$), then $(\mathfrak{g}, \mathfrak{p}_1, \mathfrak{p}_2)$ is a Manin triple.
\end{theorem}

\begin{proof}
We must verify that the triple $(\mathfrak{g}, \mathfrak{p}_1, \mathfrak{p}_2)$ satisfies the four defining conditions of a Manin triple:

(1) \textit{Lie subalgebras}. For $\mathfrak{p}_1 = \{(a, a) \mid a \in A_{\infty}\}$, the componentwise bracket yields
\[
[(a, a), (b, b)] = ([a, b], [a, b]) \in \mathfrak{p}_1,
\]
thus $\mathfrak{p}_1$ is closed under the Lie bracket. For $\mathfrak{p}_2$, let $(x_1, y_1), (x_2, y_2) \in \mathfrak{p}_2$. The bracket is
\[
[(x_1, y_1), (x_2, y_2)] = ([x_1, x_2], [y_1, y_2]).
\]
Writing $x_1 = \sum_{m \ge 0} x_{1,m}$ and $x_2 = \sum_{n \ge 0} x_{2,n}$ with elements in the respective weight spaces $A_{\infty}^m$ and $A_{\infty}^n$, we have
\[
[x_1, x_2] = \sum_{m, n \ge 0} [x_{1,m}, x_{2,n}].
\]
Because the circle action acts via $\alpha_t([x_{1,m}, x_{2,n}]) = e^{i(m+n)t}[x_{1,m}, x_{2,n}]$ and $m+n \ge 0$, we conclude that $[x_1, x_2] \in A_{\infty}^{\ge 0}$. By identical logic, $[y_1, y_2] \in A_{\infty}^{\le 0}$. Furthermore, because $[A_{\infty}^0, A_{\infty}^0] = 0$, the projection of the bracket onto the zero-weight space vanishes: $\pi_0([x_1, x_2]) = [\pi_0(x_1), \pi_0(x_2)] = 0$. Hence, $\pi_0([x_1, x_2]) + \pi_0([y_1, y_2]) = 0 + 0 = 0$, proving $\mathfrak{p}_2$ is a closed Lie subalgebra.

(2) \textit{Isotropy}. For any two elements $(a, a), (b, b) \in \mathfrak{p}_1$, the bilinear form gives
\[
\langle\langle (a, a), (b, b) \rangle\rangle = \omega(ab) - \omega(ab) = 0.
\]
For $\mathfrak{p}_2$, let $(x, y), (x', y') \in \mathfrak{p}_2$. Then $x, x' \in A_{\infty}^{\ge 0}$ and $y, y' \in A_{\infty}^{\le 0}$, and we evaluate
\[
\langle\langle (x, y), (x', y') \rangle\rangle = \omega(xx') - \omega(yy').
\]
Because the state $\omega$ is $\alpha$-invariant, $\omega(a) = \omega(\pi_0(a))$ for all $a \in A_\infty$). Since $x$ and $x'$ only contain weights $m, n \ge 0$, their product $xx'$ only has a zero-weight component if $m=0$ and $n=0$. Therefore, $\pi_0(xx') = \pi_0(x)\pi_0(x')$, which implies $\omega(xx') = \omega(\pi_0(x)\pi_0(x'))$. Similarly, $\omega(yy') = \omega(\pi_0(y)\pi_0(y'))$. Since $\pi_0(x) = -\pi_0(y)$ and $\pi_0(x') = -\pi_0(y')$, we obtain
\[
\omega(\pi_0(x)\pi_0(x')) - \omega(\pi_0(y)\pi_0(y')) = \omega((-\pi_0(y))(-\pi_0(y'))) - \omega(\pi_0(y)\pi_0(y')) = 0.
\]
Therefore, both $\mathfrak{p}_1$ and $\mathfrak{p}_2$ are isotropic.

(3) \textit{Invariance under the adjoint action}. The bilinear form is invariant if for all $(u, v), (a, b), (x, y) \in \mathfrak{g}$, we have
\[
\langle\langle [(u, v), (a, b)], (x, y) \rangle\rangle + \langle\langle (a, b), [(u, v), (x, y)] \rangle\rangle = 0.
\]
Expanding the brackets and the bilinear form yields
\[
\omega([u, a]x) - \omega([v, b]y) + \omega(a[u, x]) - \omega(b[v, y]).
\]
Because the state $\omega$ is tracial, we can cyclically permute the elements: $\omega([u,a]x) = \omega(uax) - \omega(aux)$ and $\omega(a[u,x]) = \omega(aux) - \omega(axu) = \omega(aux) - \omega(uax)$. Summing these two terms gives zero. The same cancellation occurs also for the $v, b, y$ terms, proving invariance.

(4) \textit{Weak non-degeneracy and topological direct sum}. Weak non-degeneracy of the pairing follows directly from the faithfulness of the state $\omega$. Indeed, if $\omega(ax) - \omega(by) = 0$ for all $(x, y) \in \mathfrak{g}$, setting $y=0$ and $x=a^*$ gives $\omega(aa^*) = 0$, which implies $a=0$ (and symmetrically $b=0$). 

For the topological direct sum, we must show that any element $(u, v) \in A_{\infty} \oplus A_{\infty}$ decomposes uniquely into $(a, a) \in \mathfrak{p}_1$ and $(x, y) \in \mathfrak{p}_2$. Let $P_+ = \sum_{m > 0}$ and $P_- = \sum_{m < 0}$ denote the projections onto the strictly positive and strictly negative weight spaces. The unique decomposition is given explicitly by
\[
a = P_+ v + \frac{1}{2}\pi_0(u+v) + P_- u,
\]
\[
x = P_+ (u-v) + \frac{1}{2}\pi_0(u-v),
\]
\[
y = \frac{1}{2}\pi_0(v-u) + P_- (v-u).
\]
Note that, $x \in A_{\infty}^{\ge 0}$ and $y \in A_{\infty}^{\le 0}$. Moreover, $\pi_0(x) + \pi_0(y) = \frac{1}{2}\pi_0(u-v) + \frac{1}{2}\pi_0(v-u) = 0$, so $(x,y) \in \mathfrak{p}_2$. Furthermore, $a + x = u$ and $a + y = v$. Both $\mathfrak{p}_1$ and $\mathfrak{p}_2$ are closed subspaces because the projections $P_+, P_-,$ and $\pi_0$ are continuous.
\end{proof}

\subsection{Cobracket formula}
\begin{proposition}[Cobracket on $\mathfrak{p}_1$]
Let $(\mathfrak{p}_1, \mathfrak{p}_2, \delta)$ be the Lie bialgebra structure on $\mathfrak{p}_1 \cong A_{\infty}$ induced from the Manin triple. For $(X, X) \in \mathfrak{p}_1$ with Fourier expansion $X(\theta) = \sum_k X_k e^{ik\theta}$, the cobracket acts on pairs $(\xi, \eta), (\zeta, \chi) \in \mathfrak{p}_2$ as
\[
\delta(X)((\xi, \eta), (\zeta, \chi)) = \frac{1}{2\pi} \int_0^{2\pi} \left( \omega(X(\theta)[\xi(\theta), \zeta(\theta)]) - \omega(X(\theta)[\eta(\theta), \chi(\theta)]) \right) \, d\theta.
\]
In Fourier modes, this becomes
\[
\delta(X)((\xi, \eta), (\zeta, \chi)) = \sum_{k+m+p=0} \omega(X_k[\xi_m, \zeta_p]) - \sum_{k+n+q=0} \omega(X_k[\eta_n, \chi_q]),
\]
where $\xi = \sum_{m \ge 0} \xi_m e^{im\theta}$, $\eta = \sum_{n \le 0} \eta_n e^{in\theta}$, $\zeta = \sum_{p \ge 0} \zeta_p e^{ip\theta}$, $\chi = \sum_{q \le 0} \chi_q e^{iq\theta}$.
\end{proposition}

\begin{proof}
The cobracket is defined dually by the pairing and bracket structure. For $(X, X) \in \mathfrak{p}_1$ and pairs of elements in $\mathfrak{p}_2$, the pairing with the commutator gives
\[
\langle\langle (X, X), [(\xi, \eta), (\zeta, \chi)] \rangle\rangle = \langle\langle (X, X), ([\xi, \zeta], [\eta, \chi]) \rangle\rangle = \omega(X[\xi, \zeta]) - \omega(X[\eta, \chi]).
\]
Expanding via Fourier modes and using weight orthogonality (terms with $k+m+p \neq 0$ vanish, and similarly for the $\eta, \chi$ component), we obtain the stated formula.
\end{proof}

\subsection{Examples of Manin triples}
\begin{example}[Matrix algebra with diagonal torus action]
Let $A = M_n(\mathbb{C})$ equipped with the normalized trace $\omega(a) = \frac{1}{n}\operatorname{Tr}(a)$ as the faithful tracial state. Consider the diagonal torus action $\alpha: \mathbb{S}^1 \rightarrow \operatorname{Aut}(A)$ given by
\[
\alpha_t(a) = \operatorname{diag}(t^{k_1}, \dots, t^{k_n}) \, a \, \operatorname{diag}(t^{-k_1}, \dots, t^{-k_n}),
\]
where $k_1, \dots, k_n \in \mathbb{Z}$. The corresponding weight spaces are given by
\[
A_m = \operatorname{span}\{E_{ij} \mid k_i - k_j = m\}.
\]
We define the positive and negative weight spaces as direct sums
\[
A_{\ge 0} := \bigoplus_{m \ge 0} A_m, \quad A_{\le 0} := \bigoplus_{m \le 0} A_m.
\]
The zero-weight space $A_0 = \operatorname{span}\{E_{ii}\}$ consists of the diagonal matrices, which satisfies $[A_0, A_0] = 0$. 

The constraint $\pi_0(X) + \pi_0(Y) = 0$ for $(X, Y) \in \mathfrak{p}_2$ means that the diagonal parts of $X$ and $Y$ must be strictly opposite. For an element $(E_{ij}, E_{ij}) \in \mathfrak{p}_1$ and $(X, Y) \in \mathfrak{p}_2$, the pairing evaluates to
\[
\langle\langle (E_{ij}, E_{ij}), (X, Y) \rangle\rangle = \omega(E_{ij}X) - \omega(E_{ij}Y) = \frac{1}{n}(X_{ji} - Y_{ji}).
\]
Using this, the cobracket formula in Fourier modes is given by evaluating the pairing on the commutator $([X, X'], [Y, Y'])$
\[
\delta(E_{ij})((X, Y), (X', Y')) = \omega(E_{ij}[X, X']) - \omega(E_{ij}[Y, Y']).
\]
\end{example}

\begin{example}[Manin triples of $\Gamma^\infty(T\mathbb{S}^1)$]
Consider the Fréchet space $C^{\infty}(\mathbb{S}^1)$ (with the compact-open $C^{\infty}$ topology) equipped with the Lie bracket
\[
[f, g] = f'g - g'f.
\]
This makes $C^{\infty}(\mathbb{S}^1)$ a topological Lie algebra. We note that the map
\[
\phi : C^{\infty}(\mathbb{S}^1) \rightarrow \Gamma^{\infty}(T\mathbb{S}^1), \quad f \mapsto X_f,
\]
defined by sending $f$ to the vector field $X_f := f \frac{d}{dt} \in \Gamma^\infty(T\mathbb{S}^1)$, is a continuous Lie algebra isomorphism. 

Let $A = C^{\infty}(\mathbb{S}^1)$ with the rotation action $(\alpha_t f)(\theta) = f(\theta - t)$. The corresponding weight spaces are given by
\[
A_m = \operatorname{span}\{e^{im\theta}\}.
\]
From this, we construct the Manin triple $(\mathfrak{g}, \mathfrak{p}_1, \mathfrak{p}_2)$ where
\begin{align*}
\mathfrak{p}_1 &= \{(f, f) \mid f \in C^{\infty}(\mathbb{S}^1)\}, \\
\mathfrak{p}_2 &= \left\{(X, Y) \mid X(\theta) = \sum_{k \ge 0} X_k e^{ik\theta}, Y(\theta) = \sum_{m \le 0} Y_m e^{im\theta}, \, X_0 + Y_0 = 0\right\}.
\end{align*}
The $0$-mode space $A_0$ represents the subspace of constant vector fields on $\mathbb{S}^1$, meaning $A_0 \cong \mathbb{C}$. Thus, $[A_0, A_0] = 0$, and $(\mathfrak{g}, \mathfrak{p}_1, \mathfrak{p}_2)$ forms a Manin triple.

The cobracket formula for $(X, X) \in \mathfrak{p}_1$ and pairs $(\xi, \eta), (\zeta, \chi) \in \mathfrak{p}_2$ is
\[
\delta(X)((\xi, \eta), (\zeta, \chi)) = \sum_{\substack{k<0, m \ge 0, p \ge 0 \\ k+m+p=0}} \omega(X_k[\xi_m, \zeta_p]) - \sum_{\substack{k>0, n \le 0, q \le 0 \\ k+n+q=0}} \omega(X_k[\eta_n, \chi_q]).
\]
The polynomial Fourier decay of smooth functions, combined with the continuity of the bilinear form $\omega$ and the Lie bracket, ensures that these formal sums converge.
\end{example}

\begin{example}[Manin triple of $\Gamma^\omega(T\mathbb{S}^1)$]
For $A = C^{\omega}(\mathbb{S}^1)$, we note that the torus action, the Manin triple, and the cobracket formula are identical to the smooth case. We have the following bound
\[
\|X_k\|_A \le C\rho^{-k}, \quad \|Y_m\|_A \le C\rho^{-|m|} \quad \text{for } \rho > 1, C\geq0.
\]
Thus, the sum in the cobracket formula converges absolutely
\[
\sum_{\substack{k<0, m \ge 0, p \ge 0 \\ k+m+p=0}} \|[\xi_m, \zeta_p]\|_A \cdot \|X_k\|_A \le C \sum_{k<0} \rho^{-|k|} < \infty.
\]
\end{example}

\begin{example}[Manin triple of $C^{\infty}(\mathbb{S}^1, \mathfrak{g})$]
(See also \cite[Example 3]{8}). Let $A = C^{\infty}(\mathbb{S}^1, \mathfrak{g})$. From the root space decomposition of the finite-dimensional Lie algebra $\mathfrak{g}$, we have
\[
\mathfrak{g} = \mathfrak{h} \oplus \bigoplus_{\alpha \in \Phi} \mathfrak{g}_{\alpha},
\]
where $\mathfrak{h}$ is a Cartan subalgebra (the maximal commuting subalgebra), and $\Phi \subset \mathfrak{h}'$ is the set of roots. Let $\langle \cdot, \cdot \rangle_{\mathfrak{g}} : \mathfrak{g} \times \mathfrak{g} \rightarrow \mathbb{K}$ be the Killing form on $\mathfrak{g}$. We define the invariant, non-degenerate symmetric bilinear form on $A$ via the Killing form
\[
\langle\langle f, g \rangle\rangle_A := \frac{1}{2\pi}\int_0^{2\pi} \langle f(t), g(t) \rangle_{\mathfrak{g}} \, dt.
\]
Let us define the function spaces
\begin{align*}
A_0 &:= C^{\infty}(\mathbb{S}^1, \mathfrak{h}), \\
A_{\alpha} &:= C^{\infty}(\mathbb{S}^1, \mathfrak{g}_{\alpha}).
\end{align*}
Then we have the following weight space decomposition
\[
C^{\infty}(\mathbb{S}^1, \mathfrak{g}) = A_0 \oplus \bigoplus_{\alpha \in \Phi} A_{\alpha}.
\]
Fixing a choice of positive roots $\Phi^+ \subset \Phi$ and negative roots $\Phi^-$, let us define
\begin{align*}
A_{\ge 0} &:= A_0 \oplus \bigoplus_{\alpha \in \Phi^+} A_{\alpha}, \\
A_{\le 0} &:= A_0 \oplus \bigoplus_{\alpha \in \Phi^-} A_{\alpha}.
\end{align*}
We define the Lie algebra $\mathfrak{p} := A \oplus A$, and its subalgebras
\begin{align*}
\mathfrak{p}_1 &:= \{(x, x) \mid x \in A\}, \\
\mathfrak{p}_2 &:= \{(x, y) \mid x \in A_{\ge 0}, \, y \in A_{\le 0}, \, \pi_0(x) + \pi_0(y) = 0\},
\end{align*}
where $\pi_0: A \rightarrow A_0$ is the canonical projection onto the Cartan component. We observe that $[A_0, A_0]_A = 0$ since $\mathfrak{h}$ is abelian and the Lie bracket on $A$ is defined pointwise. The inclusion $[A_{\alpha}, A_{\beta}] \subseteq A_{\alpha+\beta}$ follows from the identity $[\mathfrak{g}_{\alpha}, \mathfrak{g}_{\beta}] \subseteq \mathfrak{g}_{\alpha+\beta}$. Equipped with the split pairing 
\[
\langle\langle (x, x'), (y, y') \rangle\rangle_{\mathfrak{p}} := \langle\langle x, y \rangle\rangle_A - \langle\langle x', y' \rangle\rangle_A,
\]
the triple $(\mathfrak{p}, \mathfrak{p}_1, \mathfrak{p}_2)$ satisfies the required hypotheses of a weight space decomposition (Definition \ref{Weight space decomposition}) and thus forms a Manin triple by Theorem \ref{Weight space decomposition-manin}.
\end{example}

\begin{example}[Manin triple of $C^{\omega}(\mathbb{S}^1, \mathfrak{g})$]
Let $A = C^{\omega}(\mathbb{S}^1, \mathfrak{g})$ be the space of analytic loops taking values in $\mathfrak{g}$. We define the corresponding root spaces for these analytic loops as $A_{\alpha} := C^{\omega}(\mathbb{S}^1, \mathfrak{g}_{\alpha})$. Following a construction perfectly analogous to the smooth case, we define the subspaces $\mathfrak{p}, \mathfrak{p}_1,$ and $\mathfrak{p}_2$. By verifying the weight space decomposition and the invariance of the bilinear form in the analytic setting, we conclude that $(\mathfrak{p}, \mathfrak{p}_1, \mathfrak{p}_2)$ constitutes a Manin triple.
\end{example}
\section{Integration of Lie bialgebras}\label{6}
Let $G$ be a 1-connected regular convenient Lie group with Lie algebra $\mathfrak{g}$. Let $\mathfrak{b} \subseteq \mathfrak{g}'$ be a convenient subspace such that the coadjoint action $\operatorname{Ad}^* : G \times \mathfrak{g}' \to \mathfrak{g}'$ restricts to a smooth action on $\mathfrak{b}$. Differentiating this restricted action ensures that the infinitesimal coadjoint action $\operatorname{ad}^*$ of $\mathfrak{g}$ on $\mathfrak{b}$ is also well-defined and smooth. Consequently, these naturally induce well-defined, smooth adjoint actions $\operatorname{Ad}^{(2)}$ of $G$ on the space of skew-symmetric bilinear forms $L^2_{\text{skew}}(\mathfrak{b})$ and the $\operatorname{ad}^{(2)}$ action of the Lie algebra $\mathfrak{g}$ on $L^2_{\text{skew}}(\mathfrak{b})$.

Now, suppose $(\mathfrak{g}, \mathfrak{b}, \delta)$ is a convenient Lie bialgebra structure. By definition, the map $\delta: \mathfrak{g} \to L^2_{\text{skew}}(\mathfrak{b})$ is a 1-cocycle relative to the $\operatorname{ad}^{(2)}$ action. Applying Lemma \ref{4.5}, this 1-cocycle condition is strictly equivalent to the statement that the map
\[
\tilde{\delta}: \mathfrak{g} \to \mathfrak{g} \rtimes L^2_{\text{skew}}(\mathfrak{b}), \quad X \mapsto (X, \delta(X))
\]
is a bounded Lie algebra homomorphism into the semidirect product Lie algebra. 

Because the action $\operatorname{Ad}^{(2)}$ of $G$ on $L^2_{\text{skew}}(\mathfrak{b})$ is smooth, we can form the semidirect product Lie group $G \rtimes L^2_{\text{skew}}(\mathfrak{b})$. Furthermore, the regularity of $G$ and $L^2_{\text{skew}}(\mathfrak{b})$ ensure that $G \rtimes L^2_{\text{skew}}(\mathfrak{b})$ is the regular Lie group (by \ref{regularsemidirect}) corresponding to the Lie algebra $\mathfrak{g} \rtimes L^2_{\text{skew}}(\mathfrak{b})$. Since $G$ is 1-connected, we can then apply Proposition \ref{proposition 4.10} to integrate the Lie algebra homomorphism $\tilde{\delta}$. This yields a unique smooth Lie group homomorphism
\[
\tilde{\theta}: G \to G \rtimes L^2_{\text{skew}}(\mathfrak{b}).
\]

By Lemma \ref{4.4}, the second component of this Lie group homomorphism defines a smooth group 1-cocycle $\Theta: G \to L^2_{\text{skew}}(\mathfrak{b})$ relative to the $\operatorname{Ad}^{(2)}$ action (for which we derived the path-integral formula \eqref{1-cocycle}). 

This group 1-cocycle allows us to define a unique multiplicative bivector field $\pi$ by translating the values of the cocycle across the group. Specifically, $\pi \in \Gamma(L^2_{\text{skew}}(\mathbb{F}_{\mathfrak{b}}))$ is generated by setting $\pi(g) := R_{g}^{**}\Theta(g)$, where $\mathbb{F}_{\mathfrak{b}}$ is the vector bundle over $G$ whose fibers are generated by translating $\mathfrak{b}$ via the right action of $G$ (i.e., $\mathbb{F}_{\mathfrak{b}}(g) := R_{g^{-1}}^*\mathfrak{b}$). 

To establish that $(G, \mathbb{F}_{\mathfrak{b}}, \pi)$ constitutes a Poisson Lie group structure, it only remains to be shown that the integrated multiplicative bivector $\pi$ satisfies the Jacobi identity. We show that such $\pi$ satisfy the Jacobi idenity in the next section.

\subsection{The Jacobi identity}

\begin{proposition}
Let $G$ be a 1-connected regular Lie group with Lie algebra $\mathfrak{g}$ and let $(\mathfrak{g}, \mathfrak{b},\delta)$ be a Lie bialgebra structure as described in the previous section, which integrates to a multiplicative bivector $\pi \in \Gamma(L^2_{\text{skew}}(\mathbb{F}_{\mathfrak{b}}))$. Then, $\pi$ satisfies the Jacobi identity (condition (2) of Definition \ref{3.2}).
\end{proposition}

\begin{proof}
Let $\alpha, \beta \in \Gamma_l(\mathbb{F})$ be left-invariant sections, $x \in G$, and $Y \in \mathfrak{g}$. We evaluate the differential $d(\pi(\alpha, \beta))$ at $x$ in the direction of $L_{x*}Y$:
\begin{align*}
d(\pi(\alpha, \beta))(x)(L_{x*}Y) &= \frac{d}{dt}\bigg|_{t=0} \pi(x \exp(tY)) \Big( \alpha(x \exp(tY)), \beta(x \exp(tY)) \Big) \\
\intertext{Applying the multiplicativity condition $\pi(gh) = L_{g}^{**}\pi(h) + R_{h}^{**}\pi(g)$, this splits into two terms:}
&= \frac{d}{dt}\bigg|_{t=0} \Big[ L_{x}^{**}\pi(\exp(tY)) \Big( \alpha(x \exp(tY)), \beta(x \exp(tY)) \Big) \\
&\qquad\quad + R_{\exp(tY)}^{**}\pi(x) \Big( \alpha(x \exp(tY)), \beta(x \exp(tY)) \Big) \Big] \\
\intertext{By the definition of bivector pushforwards ($f^{**}\pi(\mu, \nu) = \pi(f^*\mu, f^*\nu)$), we transfer the actions to pullbacks on the sections:}
&= \frac{d}{dt}\bigg|_{t=0} \Big[ \pi(\exp(tY)) \Big( L_{x}^*\alpha(x \exp(tY)), L_{x}^*\beta(x \exp(tY)) \Big) \\
&\qquad\quad + \pi(x) \Big( R_{\exp(tY)}^*\alpha(x \exp(tY)), R_{\exp(tY)}^*\beta(x \exp(tY)) \Big) \Big] \\
\intertext{Using the invariance properties of the sections, the pullbacks simplify (e.g., $L_x^*\alpha(xh) = \alpha(h)$), leaving the second term independent of $t$:}
&= \frac{d}{dt}\bigg|_{t=0} \Big[ \pi(\exp(tY)) \Big( \alpha(\exp(tY)), \beta(\exp(tY)) \Big) \Big] + \frac{d}{dt}\bigg|_{t=0} \Big[ \pi(x) \Big( \alpha(x), \beta(x) \Big) \Big] \\
\intertext{Since the second term is constant with respect to $t$, its derivative is exactly $0$. The first term evaluates via the Leibniz rule (and the fact that $\pi(e)=0$) to:}
&= d\pi(e)(Y) \Big( \alpha(e), \beta(e) \Big) + 0 \\
&= \langle Y, [\alpha(e), \beta(e)]_{\mathfrak{b}} \rangle \qquad\text{(using equation \ref{linearzitaion_of_pi})}.
\end{align*}

This proves that $d(\pi(\alpha, \beta))(x) = L_{x^{-1}}^* \left( [\alpha(e), \beta(e)]_{\mathfrak{b}} \right)$. 

Because $[\alpha(e), \beta(e)]_{\mathfrak{b}} \in \mathfrak{b}$ and the subbundle fiber is defined as $\mathbb{F}_x = L_{x^{-1}}^*\mathfrak{b}$, the 1-form $d(\pi(\alpha, \beta))$ factors uniquely through $\mathbb{F}$. We now note that the factored section
\[
\widetilde{d(\pi(\alpha, \beta))}(x) = L_{x^{-1}}^* \left( [\alpha(e), \beta(e)]_{\mathfrak{b}} \right),
\]
is the unique left-invariant extension of $[\alpha(e), \beta(e)]_{\mathfrak{b}}$. Thus, the space $\Gamma_l(\mathbb{F})$ is closed under the bracket $[\alpha, \beta]_{\mathbb{F}} := \widetilde{d(\pi(\alpha, \beta))}$.

To establish the Jacobi identity for left-invariant sections $\alpha, \beta, \gamma \in \Gamma_l(\mathbb{F})$, we substitute the factored sections into the Jacobiator at $x \in G$:
\begin{align*} 
\text{Jac}(\alpha, \beta, \gamma)(x) &= \pi(\alpha, [\beta, \gamma]_{\mathbb{F}})(x) + \pi(\beta, [\gamma, \alpha]_{\mathbb{F}})(x) + \pi(\gamma, [\alpha, \beta]_{\mathbb{F}})(x) \\
&= L_{x^{-1}}^* \Big( [\alpha(e), [\beta(e), \gamma(e)]_{\mathfrak{b}}]_{\mathfrak{b}} + [\beta(e), [\gamma(e), \alpha(e)]_{\mathfrak{b}}]_{\mathfrak{b}} + [\gamma(e), [\alpha(e), \beta(e)]_{\mathfrak{b}}]_{\mathfrak{b}} \Big) \\
&= 0.
\end{align*}
For general closed local sections, the Jacobiator forms a smooth alternating 3-tensor field. Because $\mathbb{F}$ is globally spanned by left-invariant sections and the 3-tensor vanishes identically on this frame, it must be zero everywhere on $G$. Thus, condition (2) of Definition \ref{3.2} is satisfied.
\end{proof}
We summarize the result of Theorem \ref{D-0} and the result corresponding to the integration of a Lie bialgebra in the following theorems.

\subsection{Drinfeld theorem in infinite dimensions}

\begin{theorem}\label{D-1}
Let $(G, \mathbb{F}_{\mathfrak{b}}, \pi)$ be a convenient Poisson Lie group with the Lie algebra $\mathfrak g$. Let $\mathfrak b\subseteq \mathfrak g'$ be a convenient space such that the canonical map $\operatorname{incl}:\mathfrak b\to \mathfrak g', \alpha\mapsto \alpha$ is smooth, and the coadjoint action $\operatorname{Ad}^*: G\times \mathfrak g'\to \mathfrak g'$ restricts to the smooth action on $\mathfrak b$. Then the differential $\delta:=d(\pi)(e)$ of $\pi$ at the identity induces a unique (up to isomorphism) convenient Lie bialgebra structure $(\mathfrak{g}, \mathfrak{b}, \delta)$ on $\mathfrak{g}$ where $\mathbb F_{\mathfrak b}(e)=\mathfrak b$. Moreover, if $\mathfrak{g}$ and $\mathfrak{b}$ are Fréchet or Silva, then it defines a topological (continuous) Lie bialgebra structure.
\end{theorem}

\begin{theorem}\label{D-2}
Let $G$ be a 1-connected convenient regular Lie group. Let $\mathfrak b\subseteq \mathfrak g'$ be a convenient space such that the canonical map $\operatorname{incl}:\mathfrak b\to \mathfrak g', \alpha \mapsto \alpha$ is smooth, and the coadjoint action $\operatorname{Ad}^*: G\times \mathfrak g'\to \mathfrak g'$ restricts to the smooth action on $\mathfrak b$. Then, a convenient Lie bialgebra structure $(\mathfrak{g}, \mathfrak{b}, \delta)$ integrates to a unique (up to isomorphism) convenient Poisson Lie group structure $(G, \mathbb{F}_{\mathfrak{b}}, \pi)$ where $\mathbb F_{\mathfrak b}$ is generated by the right translation of $\mathfrak b$ by the action of $G$,  $\mathbb F_{\mathfrak b}(g):=R_{g}^*\mathfrak b$ and $\pi$ is given by the right translation of the group $1$-cocycle $\Theta$ of $G$, where $\Theta$ is the integration of the Lie algebra $1$-cocycle $\delta$ given by formula \ref{1-cocycle}. 
\end{theorem}

\begin{theorem}\label{D-3}
Let $G$ be a 1-connected regular Lie group with the convenient Lie algebra $\mathfrak g$. Let $\mathfrak b\subseteq \mathfrak g'$ be a convenient subspace such that the canonical map $\operatorname{incl}:\mathfrak b\to \mathfrak g', \alpha\mapsto \alpha$ is smooth and the coadjoint action $\operatorname{Ad}^*: G\times \mathfrak g'\to \mathfrak g'$ restricts to the smooth action on $\mathfrak b$.
Then the following assertions are equivalent up to isomorphism :
\begin{enumerate}
    \item $(G, \mathbb{F}_{\mathfrak{b}}, \pi)$ is a Poisson Lie group structure such that $\pi^{\sharp}(\mathbb{F}_x) \subset T_xG$ for each $x \in G$.
    \item $(\mathfrak{g}, \mathfrak{b}, \delta)$ is a Lie bialgebra structure and $\mathfrak{g}$ is a Lie-Poisson space with respect to $\mathfrak{b}$.
    \item $(\mathfrak{d} = \mathfrak{g} \oplus \mathfrak{b}, \mathfrak{g}, \mathfrak{b})$ is a Manin triple.
\end{enumerate}
\end{theorem}

The correspondence between Poisson Lie groups and Lie bialgebras can also be lifted to a correspondence between Poisson Lie group homomorphisms and Lie bialgebra morphisms. Let $(G, \mathbb{F}_1, \pi_1)$ and $(H, \mathbb{F}_2, \pi_2)$ be two Poisson Lie group structures, and let $\Phi: G \to H$ be a Poisson Lie group homomorphism. Then the induced map $\phi := d\Phi(e): \mathfrak{g} \to \mathfrak{h}$ is a Lie bialgebra morphism. Let $\mathfrak{b}_1 = \mathbb{F}_1(e)$ and $\mathfrak{b}_2 = \mathbb{F}_2(e)$. The compatibility condition on $\Phi$ (see Definition \ref{poissonhom}) states that
\[
\Theta_2 \circ \Phi = (\phi^*|_{\mathfrak{b}_2})_* \circ \Theta_1,
\]
where $\Theta_1$ and $\Theta_2$ are the $1$-cocycles associated with the bivectors $\pi_1$ and $\pi_2$, respectively. Differentiating this condition at the identity and using the fact that $\Theta_2(e) = 0$, we obtain
\[
\delta_2 \circ \phi = (\phi^*|_{\mathfrak{b}_2})_* \circ \delta_1,
\]
where $\delta_1 = d\Theta_1(e)$ and $\delta_2 = d\Theta_2(e)$. Thus, $\phi$ is indeed a Lie bialgebra morphism.

Conversely, assume $\phi: \mathfrak{g} \to \mathfrak{h}$ is a Lie bialgebra morphism. Since $G$ is connected and simply connected, and $H$ is a regular Lie group, $\phi$ integrates to a unique Lie group homomorphism $\Phi: G \to H$ such that $d\Phi(e) = \phi$. The compatibility condition of $\phi$ reads
\[
\delta_2 \circ \phi = (\phi^*|_{\mathfrak{b}_2})_* \circ \delta_1.
\]
We note that
\begin{align*}
(\delta_2 \circ \phi)([X,Y]) &= \delta_2([\phi(X), \phi(Y)]) \\
&= \operatorname{ad}_{\phi(X)}^{(2)}(\delta_2(\phi(Y))) - \operatorname{ad}^{(2)}_{\phi(Y)}(\delta_2(\phi(X))).
\end{align*}
Thus, $\delta_2 \circ \phi$ is a $1$-cocycle of the Lie algebra $\mathfrak{g}$ with values in $L^2_{\text{skew}}(\mathfrak{b}_2)$, relative to the action of $\mathfrak{g}$ on $L^2_{\text{skew}}(\mathfrak{b}_2)$ given through the Lie algebra morphism $\phi$ as
\[
X \cdot \eta := \operatorname{ad}^{(2)}_{\phi(X)}\eta,
\]
for $X \in \mathfrak{g}$ and $\eta \in L^2_{\text{skew}}(\mathfrak{b}_2)$. Since $d(\Theta_2 \circ \Phi)(e) = \delta_2 \circ \phi$, $G$ is simply connected, and $G \ltimes L^2_{\text{skew}}(\mathfrak{b}_2)$ is a regular Lie group, Proposition \ref{proposition 4.10} implies that the $1$-cocycle $\delta_2 \circ \phi$ integrates to the unique Lie group $1$-cocycle $\Theta_2 \circ \Phi$. 

Now consider the right-hand side, $(\phi^*|_{\mathfrak{b}_2})_* \circ \delta_1 : \mathfrak{g} \to L^2_{\text{skew}}(\mathfrak{b}_2)$. We compute
\begin{align*}
((\phi^*|_{\mathfrak{b}_2})_* \circ \delta_1)([X,Y]) &= (\phi^*|_{\mathfrak{b}_2})_*\left(\operatorname{ad}_{X}^{(2)}(\delta_1(Y)) - \operatorname{ad}_{Y}^{(2)}(\delta_1(X))\right)\\
&= \operatorname{ad}_{\phi(X)}^{(2)}\left((\phi^*|_{\mathfrak{b}_2})_*(\delta_1(Y))\right) - \operatorname{ad}_{\phi(Y)}^{(2)}\left((\phi^*|_{\mathfrak{b}_2})_*(\delta_1(X))\right).
\end{align*}
This last equality holds because the pushforward map intertwines with the coadjoint actions
\[
(\phi^*|_{\mathfrak{b}_2})_* \circ \operatorname{ad}_{X}^{(2)} = \operatorname{ad}^{(2)}_{\phi(X)} \circ (\phi^*|_{\mathfrak{b}_2})_* \in L(L^2_{\text{skew}}(\mathfrak{b}_1), L^2_{\text{skew}}(\mathfrak{b}_2)).
\]
Indeed, let $\eta \in L^2_{\text{skew}}(\mathfrak{b}_1)$ and $\alpha, \beta \in \mathfrak{b}_2$. Evaluating the left-hand side yields
\begin{align*}
((\phi^*|_{\mathfrak{b}_2})_* \circ \operatorname{ad}_{X}^{(2)}\eta)(\alpha, \beta) &= \operatorname{ad}_{X}^{(2)} \eta\left(\phi^*(\alpha), \phi^*(\beta)\right)\\
&= \eta\left(\operatorname{ad}^*_{X}(\phi^*(\alpha)), \phi^*(\beta)\right) + \eta\left(\phi^*(\alpha), \operatorname{ad}^*_{X}(\phi^*(\beta))\right).
\end{align*}
Because $\phi$ is a Lie algebra morphism, its dual map satisfies the identity $\operatorname{ad}^*_{X}(\phi^*(\alpha)) = \phi^*(\operatorname{ad}^*_{\phi(X)}\alpha)$. Substituting this into the equation, we obtain
\begin{align*}
&\eta\left(\phi^*(\operatorname{ad}^*_{\phi(X)}\alpha), \phi^*(\beta)\right) + \eta\left(\phi^*(\alpha), \phi^*(\operatorname{ad}^*_{\phi(X)}\beta)\right)\\
&= (\phi^*|_{\mathfrak{b}_2})_*(\eta)(\operatorname{ad}^*_{\phi(X)}\alpha, \beta) + (\phi^*|_{\mathfrak{b}_2})_*(\eta)(\alpha, \operatorname{ad}^*_{\phi(X)}\beta)\\
&= (\operatorname{ad}_{\phi(X)}^{(2)} \circ (\phi^*|_{\mathfrak{b}_2})_*\eta)(\alpha, \beta).
\end{align*}
Therefore,
\begin{align*}
((\phi^*|_{\mathfrak{b}_2})_* \circ \delta_1)([X,Y]) &= \operatorname{ad}_{\phi(X)}^{(2)}\left((\phi^*|_{\mathfrak{b}_2})_*(\delta_1(Y))\right) - \operatorname{ad}_{\phi(Y)}^{(2)}\left((\phi^*|_{\mathfrak{b}_2})_*(\delta_1(X))\right).
\end{align*}
Hence, $(\phi^*|_{\mathfrak{b}_2})_* \circ \delta_1 : \mathfrak{g} \to L^2_{\text{skew}}(\mathfrak{b}_2)$ is a Lie algebra $1$-cocycle of the Lie algebra $\mathfrak{g}$ relative to the action $\widetilde{\operatorname{ad}}^{(2)}_X(\eta) := \operatorname{ad}_{\phi(X)}^{(2)}\eta$, for $X \in \mathfrak{g}$ and $\eta \in L^2_{\text{skew}}(\mathfrak{b}_1)$. It integrates to the unique Lie group cocycle $(\phi^*|_{\mathfrak{b}_2})_* \circ \Theta_1 : G \to L^2_{\text{skew}}(\mathfrak{b}_2)$. Thus, the equality of Lie algebra $1$-cocycles
\[
\delta_2 \circ \phi = (\phi^*|_{\mathfrak{b}_2})_* \circ \delta_1
\]
integrates to the equality of Lie group $1$-cocycles
\[
\Theta_2 \circ \Phi = (\phi^*|_{\mathfrak{b}_2})_* \circ \Theta_1,
\]
which proves that $\Phi: G \to H$ is a Poisson Lie group homomorphism.\\

 Let us rewrite the correspondence between Poisson lie group homomorphism and Lie bialgebra morphisms in the following theorem:
\begin{theorem}
Let $(G, \mathbb F_1, \pi_1)$ and $(H, \mathbb F_2, \pi_2)$ be two convenient Poisson Lie group structures. Let $\Phi: G\to H$ be a Poisson Lie group homomorphism. Then $d\Phi(e): \mathfrak g \to \mathfrak h$  is a Lie bialgebra morphism. Conversely, if $G$ and $H$ are 1-connected convenient regular Lie groups, then a Lie bialgebra morphism $\phi:\mathfrak g \to \mathfrak h$ between Lie bialgebras $Lie(G)=\mathfrak g$ and $Lie(H)=\mathfrak h$ integrates uniquely to a Poisson Lie group homomorphism $\Phi: G\to H$ such that $d\Phi(e)=\phi$.
\end{theorem}

 \begin{center}
    \begin{tikzcd}[row sep=large, column sep=huge]
        G \arrow[r, "\Phi"] \arrow[d] \arrow[dd, "\Theta_1"', bend right=40] 
        & H \arrow[d] \arrow[dd, "\Theta_2", bend left=40] \\
        
        \mathfrak{g} \arrow[r, "\phi"] \arrow[d, "\delta_1" dashed] 
        & \mathfrak{h} \arrow[d,"\delta_2" dashed] \\
        
        L^2_{{skew}}(\mathfrak b_1) \arrow[r, "(\phi^*|_{\mathfrak{b}_2})_*"] 
        & L^2_{{skew}}(\mathfrak b_2)
    \end{tikzcd}
\end{center}

Let us rewrite this equivalence in the category theoretic language.  
\begin{definition}
Let $\CatPL$ be the category whose
\begin{itemize}
    \item \textbf{Objects} are triples $(G, \mathbb{F}, \pi)$, where $G$ is a 1-connected, convenient regular Lie group equipped with a Poisson Lie group structure $(G, \mathbb{F}, \pi)$.
    \item \textbf{Morphisms} are Poisson-Lie group homomorphisms $\Phi: G \to H$.
\end{itemize}
\end{definition}

\begin{definition}
Let $\CatLB$ be the category whose
\begin{itemize}
    \item \textbf{Objects} are triples $(\mathfrak{g}, \mathfrak{b}, \delta)$, where $\mathfrak{g} = \text{Lie}(G)$ is the Lie algebra of a convenient regular Lie group $G$, equipped with a  Lie bialgebra structure $(\mathfrak{g}, \mathfrak{b}, \delta)$.
    \item \textbf{Morphisms} are Lie bialgebra homomorphisms $\phi: \mathfrak{g} \to \mathfrak{h}$.
\end{itemize}
\end{definition}

\begin{theorem}[Equivalence between $\CatPL$ and $\CatLB$]\label{equivalence}
The differentiation functor $\mathcal{L}ie$ defined by taking the tangent map at the identity establishes an equivalence of categories between $\CatPL$ and $\CatLB$. Its inverse is a unique left adjoint functor (Integration, denoted ``Int'') up to natural isomorphism. This means every Lie bialgebra structure $(\mathfrak{g}, \mathfrak{b}, \delta)$ in $\CatLB$ integrates to a unique Poisson Lie structure $(G, \mathbb{F}, \pi)$ on the corresponding Lie group $G$ such that $\mathcal{L}ie(G)=\mathfrak g$.
\end{theorem}

\[
\begin{tikzcd}[row sep=huge, column sep=huge, ampersand replacement=\&]
    \parbox{3.5cm}{\centering \textbf{Category} $\CatPL$ \\ 
    \scriptsize Objects: $(G, \mathbb{F}, \pi)$ \\ 
    \scriptsize Morphisms: $\Phi: G\to H$}   
    \arrow[r, "\mathcal{L}ie"', shift right=1.0ex] 
    \arrow[r, phantom, "\simeq"]
    \& 
    \parbox{4.0cm}{\centering \textbf{Category} $\CatLB$ \\ 
    \scriptsize Objects: $(\mathfrak{g}, \mathfrak{b}, \delta)$ \\ 
    \scriptsize Morphisms: $\phi = d\Phi(e):\mathfrak g \to \mathfrak h$}
    \arrow[l, "\text{Int}"', shift right=1.5ex]
\end{tikzcd}
\]

\begin{proof}

Let $Y \in \mathrm{LieBiAlg_{reg}}$ be a Lie bialgebra and let $X \in \mathrm{PLGrp_{reg}^{sc}}$ be a Poisson-Lie group. To establish the adjunction $\mathrm{Int} \dashv \mathcal{L}ie$, we need to  construct a bijection
$$ \Phi_{Y,X} : \operatorname{Hom}_{\mathrm{PLGrp_{reg}^{sc}}}(\mathrm{Int}(Y), X) \xrightarrow{\sim} \operatorname{Hom}_{\mathrm{LieBiAlg_{reg}}}(Y, \mathcal{L}ie(X)) $$
that is natural in both $Y$ and $X$. 

First, consider the unit of the adjunction for the object $Y$. This is a natural transformation $\eta_Y : Y \to \mathcal{L}ie(\mathrm{Int}(Y))$. Because $\mathrm{Int}(Y)$ is constructed as the unique 1-connected Lie group integrating the Lie algebra of $Y$, the unit $\eta_Y$ is a canonical isomorphism of Lie bialgebras, and its inverse $\eta_Y^{-1} : \mathcal{L}ie(\mathrm{Int}(Y)) \to Y$ is well-defined.

\textit{Construction of $\Phi_{Y,X}$:} 
Given a Poisson Lie group homomorphism $\Psi : \mathrm{Int}(Y) \to X$, we define the corresponding Lie bialgebra homomorphism by applying the differentiation functor and pre-composing with the unit:
$$ \Phi_{Y,X}(\Psi) := \mathcal{L}ie(\Psi) \circ \eta_Y. $$
Since $\mathcal{L}ie(\Psi)$ and $\eta_Y$ are both Lie bialgebra morphisms, their composition is a morphism in $\mathrm{LieBiAlg_{reg}}$.

\textit{Construction of the inverse $\Phi_{Y,X}^{-1}$:}
Let $g \in \operatorname{Hom}_{\mathrm{LieBiAlg_{reg}}}(Y, \mathcal{L}ie(X))$ be a Lie bialgebra homomorphism. Because $\eta_Y$ is an isomorphism, we can construct the composition
$$ \tilde{g} := g \circ \eta_Y^{-1} : \mathcal{L}ie(\mathrm{Int}(Y)) \to \mathcal{L}ie(X). $$
Note that $\tilde{g}$ is a Lie bialgebra homomorphism mapping from the Lie algebra of $\mathrm{Int}(Y)$ to the Lie algebra of $X$. 

Because the domain group $\mathrm{Int}(Y)$ is 1-connected and $X$ is regular, there exists a unique smooth Lie group homomorphism $F(g) : \mathrm{Int}(Y) \to X$ such that its derivative satisfies $\mathcal{L}ie(F(g)) = \tilde{g}$. Furthermore, because $\tilde{g}$ is compatible with respect to the cocycles (condition (2) of Definition \ref{cocycle compatibility}), $F(g)$ respects the multiplicative Poisson structure (equation 4.2 of Definition \ref{Cocycle compatibility}), making it a  morphism in $\mathrm{PLGrp_{reg}^{sc}}$. 
We define $\Phi_{Y,X}^{-1}(g) := F(g)$.
\end{proof}

\subsubsection{Path-integral formula for the group 1-cocycle}

Let $(\mathfrak{g}, \mathfrak{b}, \delta)$ be a Lie bialgebra structure. By Theorem \ref{D-2}, it integrates to a unique Poisson Lie group structure $(G, \mathbb{F}, \pi)$. The corresponding $1$-cocycle $\theta: G \to L^2_{\text{skew}}(\mathfrak{b})$ is the second component of the Lie group homomorphism $\tilde{\theta}: G \to G \rtimes L^2_{\text{skew}}(\mathfrak{b})$, which is the integration of the Lie algebra morphism $\tilde{\delta}: \mathfrak{g} \to \mathfrak{g} \rtimes L^2_{\text{skew}}(\mathfrak{b})$ defined by $X \mapsto (X, \delta(X))$. 

Let $X \in C^\infty(\mathbb{R}, \mathfrak{g})$. Using the regularity of the Lie group $G$ (\cite[Definition 5.3]{22}), there exists a unique curve $g \in C^\infty(\mathbb{R}, G)$ that satisfies the differential equation
\begin{align*}
\begin{cases}   
g(0) = e, \\
g'(t) = R_{g(t)*}X(t).
\end{cases} 
\end{align*}
Let us denote the endpoint of this curve by $\operatorname{evol}_{G}(X) := g(1)$. By \cite[Corollary 5.9]{20}, the evolution of a curve $(X, Y) \in C^\infty(\mathbb{R}, \mathfrak{g} \rtimes L^2_{\text{skew}}(\mathfrak{b}))$ in the semidirect product group is given by
\[
\operatorname{evol}_{G \rtimes L^2_{\text{skew}}(\mathfrak{b})}(X, Y) = (g(1), h(1)),
\]
where the components are determined by
\begin{align*}
g(t) &= \operatorname{Evol}_{G}(X)(t), \\
h(t) &= \operatorname{Ad}^{(2)}_{g(t)} \int_{0}^{t} \operatorname{Ad}^{(2)}_{g(u)^{-1}} Y(u) \, du \in L^2_{\text{skew}}(\mathfrak{b}).
\end{align*}
Here, $\operatorname{Evol}_G: C^\infty(\mathbb{R}, \mathfrak{g}) \to C^\infty(\mathbb{R}, G)$ is the right evolution map, which is the inverse of the right logarithmic derivative
\[
\delta^r: C^\infty(\mathbb{R}, G) \to C^\infty(\mathbb{R}, \mathfrak{g}), \quad \delta^r(g)(t) := R_{g(t)^{-1}*} g'(t).
\]

Now, let $g_0 \in G$, and let $g \in C^\infty(\mathbb{R}, G)$ be a smooth curve such that $g(0) = e$ and $g(1) = g_0$. By taking the right logarithmic derivative of $g$, we obtain the curve $X \in C^\infty(\mathbb{R}, \mathfrak{g})$ given by $X(t) := R_{g(t)^{-1}*} g'(t)$.

Because the homomorphism $\tilde{\theta}$ integrates the Lie algebra morphism $\tilde{\delta}$, the evolution of the curve $(X(t), \delta(X(t)))$ in the Lie algebra maps exactly to the curve $(g(t), \theta(g(t)))$ in the Lie group. Setting $Y(t) := \delta(X(t)) \in C^\infty(\mathbb{R}, L^2_{\text{skew}}(\mathfrak{b}))$, the second component evaluates at $t=1$ to
\[
h(1) = \operatorname{Ad}^{(2)}_{g(1)} \int_{0}^{1} \operatorname{Ad}^{(2)}_{g(u)^{-1}} \delta(X(u)) \, du \in L^2_{\text{skew}}(\mathfrak{b}).
\]
Since $g(1) = g_0$ and $h(1) = \theta(g_0)$, we obtain the following path-integral formula for the Lie group $1$-cocycle $\theta$
\begin{equation}\label{1-cocycle}
    \theta(g_0) = \operatorname{Ad}^{(2)}_{g_0} \int_{0}^{1} \operatorname{Ad}^{(2)}_{g(t)^{-1}} \delta(X(t)) \, dt.
\end{equation}
We must verify the uniqueness of the formula \ref{1-cocycle} of the $1$ -cocycle $\Theta(g)$. By Lemma \ref{4.4}, $\Theta$ is a $1$-cocycle if and only if \[\tilde \Theta : G \to G\rtimes L^2_{skew}(\mathfrak b), g\to(g, \Theta(g)), g\in G. \] is a Lie group homomorphism. But since $G$ is a 1-connected regular Lie group, from Proposition \ref{proposition 4.10}, it follows that $\tilde \Theta $ is a unique lift of a Lie algebra morphism \[\tilde \delta : \mathfrak g \to \mathfrak g\rtimes L^2_{skew}(\mathfrak b), X\to (X,\delta(X)), X\in \mathfrak g,\] which is uniquely defined by $\delta$. Hence, formula \ref{1-cocycle} depends only on the endpoints of the path $g(t)$.

\section{Remarks on Poisson structures in the setting of Bastiani-calculus}\label{7}
For a manifold $M$ modeled on a locally convex space $E$, we encounter a problem that the bundle $p: L^2(T'M)\to M$ may not have a smooth manifold structure due to the fact that $T'M$ may not have a smooth vector bundle structure over $M$.

\begin{remark}

For a manifold $M$ modeled on a locally convex space $E$, its cotangent bundle $T'M$ may not have a smooth vector bundle structure over $M$. This is due to the fact that $E'$ may not have a topology for which the transition chart change maps of $T'M$ are smooth. But if the typical fiber over a point is of the same kind as $E$, then it has a smooth manifold structure depends on the trivialization. It also works if the cotangent bundle is a trivial vector bundle, then $T'M\cong M\times E'$. Thus, any locally convex topology on $E'$ provides a smooth manifold structure on $T'M$. However, if $M$ is a smooth manifold modeled on a convenient space $E$. Let $(U_\alpha,u_\alpha:U_\alpha\to E_\alpha)$ be a smooth atlas. The cotangent bundle $T'M$ carries a smooth manifold structure ( see \cite[33.1]{20}), where the cocycle of a transition function is given by  \[ \phi_{\alpha,\beta}:U_{\alpha}\cap U_{\beta} \to GL(E_\alpha', E_\beta'), x\mapsto d(u_{\beta}\circ u_{\alpha}^{-1})(u_\alpha(x))^*.\]

\end{remark}
We will see in the next section that in the special situation when a manifold $M$, which are regular as topological spaces and modeled on nuclear Fréchet or nuclear Silva space $E$, $L_{\text{alt}}^{k}(T_{x}'M) \cong \hat\wedge^{k}T_{x}M$ is a topological isomorphism \cite[Corollary 2, 21.2]{17}, where $\hat\wedge$ denotes the skew-symmetric tensor product completed in the projective tensor product topology. $\hat\wedge^{k}T_{x}M\cong \hat\wedge^kE$ carries nuclear Fr\'{e}chet structure when $E$ is a nuclear Fr\'{e}chet space and $\hat\wedge^{k}T_{x}M\cong \hat\wedge^kE$ carries a nuclear Silva structure when $E$ is nuclear Silva space (see \ref{Lemma9.13}). Then $\hat\wedge^2TM$ carries a smooth vector bundle structure according to \cite[Definition 3.7.1]{11}. Then the bundle projection $\pi_{\hat\wedge^kTM}: \hat\wedge^kTM \to M$ takes all $v\in \hat\wedge^kT_xM$ to $x\in M$. 
The local trivialization for each chart $\phi: U \to V \subseteq E$ of $M$ is given by

\begin{align*}\theta_\phi : \hat\wedge^kT(U) \to U \times \hat\wedge^kE, \\
\theta_\phi(v)= ((\pi_{\hat\wedge^k TM})_{|\hat\wedge^kTU},\hat\wedge ^kT_x\phi(v))
\end{align*} where $x=\pi_{\hat\wedge^k TM}(v)$ and $\hat\wedge ^kT_x\phi: \hat\wedge^kT_xM \to\hat\wedge^k E$ is the continuous linear extension of the differential $T_x\phi: T_xM\to E$ to the $k$-th exterior power.
 The inverse of $\theta_\phi$ is given by the following mapping: 
 \[U \times \hat\wedge^k E \to \hat\wedge^kTU,\\
(x,y)\to \hat\wedge ^k(T(\phi)^{-1})(\phi(x),y)\]
where $\hat\wedge ^k(T_{x}(\phi)^{-1})$ is the continuous linear extension of the map $T(\phi)^{-1}: U\times E\to TU$ to the wedge product.

\section{Drinfeld theorem for Poisson Lie groups modeled on a nuclear Fréchet and a nuclear Silva space}\label{8}

In this section, we prove the Drinfeld correspondence within the framework of Lie groups modeled on a nuclear Fréchet or a nuclear Silva space, which resembles the Drinfeld correspondence in finite dimensions. We provide an alternative proof to that presented in Section 7 for the integration of a Lie bialgebra structure $(\mathfrak{g}, \mathfrak{g}', \delta)$ on the Lie algebra $\mathfrak{g}$ of a 1-connected regular Lie group $G$ into a compatible Poisson structure $(G, T'G, \pi)$.

We establish a criterion for a bivector $\pi \in \Gamma(\hat\wedge^2TM)$ to define a Poisson structure in terms of the vanishing of the Schouten bracket $[\pi, \pi]_s=0$. We begin by proving the isomorphism between the topological module ${\hat\wedge}^k_{C^{\infty}(M)}\Gamma(TM)$ and the space of sections $\Gamma\left({\hat\wedge}^k TM\right)$, which is required to define the Schouten bracket on the graded commutative algebra $\bigoplus_{n\ge 0}\Gamma(\hat\wedge^n TM)$.

\subsection{A version of the Serre-Swan theorem in infinite dimensions}\label{S-S}
Let $M$ be a smoothly regular manifold modeled on a nuclear topological vector space $E$. We aim to demonstrate the following topological isomorphism
\[
{\hat\wedge}^k_{C^{\infty}(M)}\Gamma(TM) \cong \Gamma\left({\hat\wedge}^k TM\right).
\]
Let $U \subset E$ be an open subset, let $F$ be a nuclear topological vector space, and let $\hat\wedge^k_{C^\infty(M)}$ denote the skew-symmetric $C^\infty(M)$-linear tensor product completed in the projective tensor product topology. We first need to establish the following topological isomorphism for the local mapping spaces
\[
C^{\infty}(U, F) \cong C^{\infty}(U, \mathbb{R}) \hat{\otimes}_{\pi} F.
\]

\paragraph{Density of $C^{\infty}(U, \mathbb{R}) \otimes F$ in $C^{\infty}(U, F)$.}
\begin{remark}[{\cite[4.6]{14}}]
Let $E$ and $F$ be topological vector spaces. The $\epsilon$-topology on $E \tilde{\otimes}_{\epsilon} F$ is the initial topology with respect to the linear map
\[
\psi : E \otimes F \rightarrow L(E'_{\tau}, F)_{\epsilon}
\]
given by $\psi(x, y)(\lambda) := \lambda(x)y$ for $x \in E, y \in F$ and $\lambda \in E'$, where $E'_{\tau}$ denotes the space of continuous linear functionals endowed with the Mackey topology (the topology of uniform convergence on equicontinuous subsets of $E'$). The locally convex topology on $L(E', F)$ is given by the seminorms
\[
\|\alpha\|_{S,q} := \sup_{\lambda \in S} q(\alpha(\lambda))
\]
where $S$ varies over the set of equicontinuous subsets of $E'$, and $q$ varies over the set of continuous seminorms on $F$. Using the Bipolar theorem, we can deduce that the seminorm $p$ on $E$ can be written as
\[
p(y) = \sup_{\lambda \in S} |\lambda(y)|
\]
for all $y \in E$ (see \cite[Section 4.7]{14}).
\end{remark}

\begin{remark}
Let $\mathcal{K} \in \mathcal{K}(U)$ and $\mathcal{L} \in \mathcal{K}(E)$ be cofinal compact subsets. Then the expression
\[
\|\gamma\|_{C^j, \mathcal{K}, \mathcal{L}, q} := \max\{\|\gamma\|_{1, \mathcal{K}, \mathcal{L}, q}, \dots, \|\gamma\|_{C^j, \mathcal{K}, \mathcal{L}, q}\}
\]
defines a continuous seminorm on $C^{\infty}(U, F)$, where
\[
\|\gamma\|_{i, \mathcal{K}, \mathcal{L}, q} = \sup_{x \in \mathcal{K} \times \mathcal{L}^i} q(D^i(\gamma)(x)).
\]
\end{remark}

Let $\gamma \in C^{\infty}(U, F)$ and let $I_F$ be the identity operator on $F$. Since $F$ is a nuclear Fréchet (or Silva) space, we have $L_{\beta}(F, F) \cong F' \hat{\otimes}_{\pi} F$ (see \cite[Proposition 50.5]{31}), meaning $I_F$ can be approximated by finite-rank operators $i_n \in F' \hat{\otimes}_{\pi} F$. The fact that $C^{\infty}(U, \mathbb{R}) \otimes F$ is dense in $C^{\infty}(U, F)$ follows from the following lemma.

\begin{lemma}
For $\gamma_n := i_n \circ \gamma \in C^{\infty}(U) \otimes F$, the sequence $\gamma_n \rightarrow \gamma$ converges with respect to the compact-open $C^k$-topology of $C^{\infty}(U, F)$.
\end{lemma}

\begin{proof}
Let $j \in \mathbb{Z}_+$, let $\mathcal{K} \in \mathcal{K}(U)$ and $\mathcal{L} \in \mathcal{K}(E)$ be compact subsets, and let $q$ be a seminorm on $F$. Then, for the seminorm $\|\cdot\|_{C^j, \mathcal{K}, \mathcal{L}, q}$ on $C^{\infty}(U, F)$ (see \cite[Lemma 2.14, c]{14}), we have
\[
\|\gamma_n - \gamma\|_{C^j, \mathcal{K}, \mathcal{L}, q} = \max_{1 \le i \le j}\left\{\sup_{x \in \mathcal{K} \times \mathcal{L}^i} q(D^i(i_n \circ \gamma - \gamma)(x))\right\} = \max_{1 \le i \le j}\left\{\sup_{x \in \mathcal{K} \times \mathcal{L}^i} q((i_n - I_F)(D^i\gamma(x)))\right\}.
\]
Because each total derivative $D^i\gamma$ is continuous and the domain $x \in \mathcal{K} \times \mathcal{L}^i$ is compact, the image set $D^i\gamma(\mathcal{K} \times \mathcal{L}^i)$ lies within a compact subset of $F$. Because $i_n \rightarrow I_F$ uniformly on compact subsets of $F$, it follows that $\|\gamma_n - \gamma\|_{C^j, \mathcal{K}, \mathcal{L}, q} \rightarrow 0$.
\end{proof}

\paragraph{Topological isomorphism $C^{\infty}(U, F) \cong C^{\infty}(U, \mathbb{R}) \hat{\otimes}_{\pi} F$.}
Since $U$ is a $k_{\mathbb{R}}$-space (being a Silva space) and $F$ is complete, the space $C^{\infty}(U, F)$ is complete. The algebraic tensor product $C^{\infty}(U, \mathbb{R}) \otimes F$ is a dense vector subspace of both $C^{\infty}(U, F)$ and $C^{\infty}(U, \mathbb{R}) \hat{\otimes}_{\pi} F$ (since $F$ is nuclear). Thus, it suffices to show that $C^{\infty}(U, F)$ and $C^{\infty}(U, \mathbb{R}) \hat{\otimes}_{\pi} F$ induce the identical topology on $C^{\infty}(U, \mathbb{R}) \otimes F$. The unique continuous linear map
\[
\tilde{\Lambda}: C^{\infty}(U, F) \rightarrow F \tilde{\otimes}_{\epsilon} C^{\infty}(U) \cong F \hat{\otimes}_{\pi} C^{\infty}(U)
\]
determined by $\Lambda(\gamma \cdot v) = v \otimes \gamma$ for $v \in F$ and $\gamma \in C^{\infty}(U)$ is then a topological isomorphism (see \cite[4.6]{14} for the $\epsilon$-topology on the tensor product, which is isomorphic to the projective tensor product topology since $F$ is nuclear).

Since $F$ is nuclear, we have the topological embedding
\[
\psi : F \otimes C^{\infty}(U) \rightarrow L(F'_{\tau}, C^{\infty}(U))_{\epsilon}.
\]
It now suffices to show that the linear map $h := \psi \circ \Lambda$ is a topological embedding, where $h(\gamma)(\lambda) = (\psi \circ \Lambda)(\gamma)(\lambda) = \lambda \circ \gamma$ for all $\gamma \in Y := F \otimes C^{\infty}(U) \subseteq C^{\infty}(U, F)$ and $\lambda \in F'$. This follows from the equivalence of seminorms between $F \tilde{\otimes}_{\epsilon} C^{\infty}(U)$ and $C^{\infty}(U, F)$, specifically
\[
\|h(\gamma)\|_{S, C^j, \mathcal{K}, \mathcal{L}} = \|\gamma\|_{C^j, \mathcal{K}, \mathcal{L}, q}
\]
for each $j \in \mathbb{N}_0$, compact subsets $\mathcal{K} \in \mathcal{K}(U)$, $\mathcal{L} \in \mathcal{K}(E)$, and continuous seminorm $q$ on $F$, with $S := B_q(0)^{\circ} \subseteq F'$. Let us compute
\begin{align*}
\|h(\gamma)\|_{S, C^j, \mathcal{K}, \mathcal{L}} &= \sup_{\lambda \in S} \|\lambda \circ \gamma\|_{C^j, \mathcal{K}, \mathcal{L}} = \sup_{\lambda \in S} \max_{0 \le i \le j} \left| \sup_{x \in \mathcal{K} \times \mathcal{L}^i} D^i(\lambda \circ \gamma)(x) \right| \\
&= \sup_{\lambda \in S} \max_{0 \le i \le j} \left| \sup_{x \in \mathcal{K} \times \mathcal{L}^i} \lambda(D^i(\gamma)(x)) \right| \\
&= \max_{0 \le i \le j} \left\{ \sup_{x \in \mathcal{K} \times \mathcal{L}^i} \left| \sup_{\lambda \in S} \lambda(D^i(\gamma)(x)) \right| \right\} \\
&= \max_{0 \le i \le j} \left\{ \sup_{x \in \mathcal{K} \times \mathcal{L}^i} q(D^i(\gamma)(x)) \right\} = \|\gamma\|_{C^j, \mathcal{K}, \mathcal{L}, q}.
\end{align*}

\paragraph{Extension to the wedge product: $C^{\infty}(U, \hat\wedge^2 TU) \cong \hat\wedge^2_{C^{\infty}(U)}C^{\infty}(U, TU)$.}
We now define an equivalence relation on the space $\otimes^2_{\pi}C^{\infty}(U, F)$ such that the projective tensor product is $C^{\infty}(U)$-linear on the quotient space. Consider the map 
\[
\tau : \otimes^2_{\pi}(C^{\infty}(U) \otimes_{\pi} F) \rightarrow C^{\infty}(U) \otimes_{\pi} (F \otimes_{\pi} F), \quad (f \otimes u) \otimes (g \otimes v) \mapsto fg \otimes u \otimes v,
\]  
and its extension to the quotient 
\[
\tilde{\tau}: \otimes^2_{\pi}(C^{\infty}(U) \otimes_{\pi} F)/\ker(\tau) \rightarrow C^{\infty}(U)\otimes_{\pi}(\otimes_{\pi}^2 F).
\]
Let $\tau'$ be the extension of $\tau$ to the completion, and let $\tilde{\tau}'$ be the extension of $\tau'$ to the quotient
\begin{align*} 
\tau'&: \hat{\otimes}_{\pi}^2 C^{\infty}(U, F) \rightarrow C^{\infty}(U, \hat{\otimes}_{\pi}^2 F), \\
\tilde{\tau}'&: \hat{\otimes}_{\pi}^2 C^{\infty}(U, F)/\ker(\tau') \rightarrow C^{\infty}(U, \hat\otimes^2_{\pi}F).
\end{align*}
To show that $\tilde{\tau}'$ is a topological isomorphism, we must construct its inverse. We note the map
\begin{align*} 
I : C^{\infty}(U) \otimes_{\pi} (\otimes^2_{\pi}F) &\rightarrow \otimes^2_{\pi}(C^{\infty}(U) \otimes_{\pi} F)/\ker(\tau), \\
f \otimes (u \otimes v) &\mapsto [(f \otimes u) \otimes (1 \otimes v)] 
\end{align*} 
is clearly injective. Let $[(f \otimes u) \otimes (g \otimes v)]$ be an arbitrary equivalence class in the quotient. Since $[(f \otimes u) \otimes (g \otimes v)] = [(fg \otimes u) \otimes (1 \otimes v)]$, $I$ is also surjective.

We now consider the topological embedding 
\[
i : C^{\infty}(U)\otimes_{\pi} (\otimes_{\pi}^2 F) \rightarrow \otimes^2_{\pi}(C^{\infty}(U) \otimes_{\pi} F), \quad f \otimes u \otimes v \mapsto (f \otimes u) \otimes (1 \otimes v).
\] 
Its extension $\tilde{i} : C^{\infty}(U, \hat{\otimes}_{\pi}^2 F) \rightarrow \hat{\otimes}_{\pi}^2 C^{\infty}(U, F)$ is also an embedding. It is straightforward to see that $I$ is the continuous linear inverse of $\tilde{\tau}$ prior to completion. Let 
\[
q : \hat\otimes^2_{\pi}C^{\infty}(U, F) \rightarrow \hat\otimes^2_{\pi}C^{\infty}(U, F)/\ker(\tau')
\] 
be the quotient map. To see that $I' := q \circ \tilde{i}$ is the continuous right inverse of the map $\tilde{\tau}'$, we first observe the exact sequence relations
\begin{align*}
\ker(\tau') \cap \otimes^2_{\pi}(C^{\infty}(U) \otimes_{\pi} F) &= \ker(\tau), \\
\otimes^2_{\pi}C^{\infty}(U, F)/\ker(\tau) &\subset \hat{\otimes}_{\pi}^2 C^{\infty}(U, F)/\ker(\tau').
\end{align*} 
Therefore, the completion of the pre-quotient aligns with the quotient of the completion
\[
\overline{\otimes^2_{\pi}C^{\infty}(U, F)/\ker(\tau)}\cong \hat{\otimes}_{\pi}^2 C^{\infty}(U, F)/\ker(\tau').
\]
We then note that $\tilde{\tau}'|_{\otimes^2_{\pi}(C^{\infty}(U) \otimes_{\pi} F)} = \tilde{\tau}$ is the continuous right inverse of $I'|_{C^{\infty}(U)\otimes_{\pi}^2 F} = I$, meaning their extensions $\tilde{\tau}'$ and $I'$ are continuous mutual inverses by \cite[Lemma 4.5]{14}. Thus, we obtain the topological isomorphism

\begin{equation}\label{wedgetensor}
\tilde{\tau}': \hat{\otimes}_{\pi}^2{}_{C^{\infty}(U)}C^{\infty}(U, F) := (\hat{\otimes}_{\pi}^2 C^{\infty}(U, F))/\ker(\tau') \cong C^{\infty}(U, \hat{\otimes}^2_{\pi}F).
\end{equation}
For any locally convex space $E$ and nuclear space $F$, let $\gamma : F \hat{\otimes}_{\pi} F \rightarrow F \hat{\otimes}_{\pi} F$ be the unique continuous extension of the flip map $x \otimes y \mapsto y \otimes x$. Then we verify that

\begin{align*}
C^{\infty}(U)\hat{\otimes}_{\pi} (\hat{\wedge}^2 F) &:= C^{\infty}(U)\hat{\otimes}_{\pi} (\hat{\otimes}_{\pi}^2 F)/(\ker (Id_{\hat{\otimes}_{\pi}^2 F}-\gamma)) \\ 
&\cong (\hat{\otimes}_{\pi}^2{}_{C^{\infty}(U)}C^{\infty}(U, F))/(\ker(Id_{\hat{\otimes}_{\pi}^2{}_{C^{\infty}(U, F)}}- \tilde{\gamma})) =: \hat{\wedge}^2_{C^{\infty}(U)}C^{\infty}(U, F),
\end{align*}
where $\tilde{\gamma}$ is the unique extension of the map 
\[
\tilde{\gamma}[(f \otimes u) \otimes (g \otimes v)] := [(g \otimes v) \otimes (f \otimes u)] = [(fg \otimes v) \otimes (1 \otimes u)] = I(fg \otimes \gamma(u \otimes v))
\] 
to the completion. 

The isomorphism
\begin{align*}
&C^{\infty}(U)\hat{\otimes}_{\pi} (\hat{\otimes}_{\pi}^2 F)/ \ker (Id_{C^\infty(U)} \otimes Id_{\hat{\otimes}_{\pi}^2 F} - Id_{C^\infty(U)} \otimes \gamma) \\
&\cong (\hat{\otimes}_{\pi}^2{}_{C^{\infty}(U)}C^{\infty}(U, F))/ \ker(Id_{\hat{\otimes}_{\pi}^2{}_{C^{\infty}(U, F)}}- \tilde{\gamma})
\end{align*}
follows from the topological isomorphism $\tilde{\tau}'$ defined in \eqref{wedgetensor}. To see that $I'$ (the inverse of $\tilde{\tau}'$) passes to the quotients, we note that it intertwines the respective flip maps. For any simple tensor $f \otimes u \otimes v$, we have
\begin{align*}
\tilde{\gamma}(I'(f \otimes u \otimes v)) &= \tilde{\gamma}([(f \otimes u) \otimes (1 \otimes v)]) = [(1 \otimes v) \otimes (f \otimes u)] \\
&= [(f \otimes v) \otimes (1 \otimes u)] = I'(f \otimes v \otimes u) \\
&= I'( (Id_{C^\infty(U)} \otimes \gamma)(f \otimes u \otimes v) ).
\end{align*}
By continuity, $I' \circ (Id \otimes \gamma) = \tilde{\gamma} \circ I'$ on the entire completion. Consequently, $I'$ maps $\ker(Id - Id \otimes \gamma)$ bijectively onto $\ker(Id - \tilde{\gamma})$.

Furthermore, because $F$ is a nuclear space, its twofold completed projective tensor product $\hat{\otimes}_{\pi}^2 F$ is also nuclear. This guarantees that the completed projective tensor product preserves the exact sequence $0 \rightarrow \ker(Id - \gamma) \rightarrow \hat{\otimes}_{\pi}^2 F$ when tensoring with $C^\infty(U)$, meaning we can identify
\[
\ker (Id_{C^\infty(U)} \otimes (Id_{\hat{\otimes}_{\pi}^2 F} - \gamma)) = C^\infty(U)\hat{\otimes}_\pi\ker(Id_{\hat{\otimes}_{\pi}^2 F}-\gamma).
\]
Therefore, the image of $C^\infty(U)\hat{\otimes}_\pi\ker(Id_{\hat{\otimes}_{\pi}^2 F}-\gamma)$ under $I'$ coincides with $\ker(Id_{\hat{\otimes}_{\pi}^2{}_{C^{\infty}(U, F)}}- \tilde{\gamma})$. Since $I'$ is a topological isomorphism mapping these closed subspaces onto each other, it induces the desired topological isomorphism on the quotients. Consequently, it can be passed to the wedge product. We thus obtain the topological isomorphisms
\begin{align*}
C^{\infty}(U, \hat\wedge^2 TU) &\cong \hat\wedge^2_{C^{\infty}(U)}C^{\infty}(U, TU), \\
\Gamma(U, \hat\wedge^2 TU) &\cong \hat\wedge^2 \Gamma(U, TU). 
\end{align*}
Since $M$ is smoothly regular, smooth cut-off functions exist. Thus, we can extend this local trivialization globally to obtain the topological isomorphism for the manifold
\begin{equation}\label{Iso}
\psi : \Gamma(M, \wedge^2 TM) \rightarrow \wedge^2_{C^{\infty}(M)}(\Gamma(TM)).
\end{equation}

The following commutative diagram outlines the construction of the topological isomorphisms passing to the quotient wedge products:

\begin{figure}[h]
    \centering
    \begin{tikzcd}[row sep=large, column sep=large]
        0 \arrow[r] 
        & \ker(Id - \tilde{\gamma}) \arrow[r, hook] \arrow[d, "I'|"', "\cong"] 
        & \frac{\hat{\otimes}_{\pi}^2 C^{\infty}(U, F)}{\ker(\tau')} \arrow[r, "q_1"] \arrow[d, "I'"', shift right=1ex] 
        & \wedge^2_{C^{\infty}(U)}C^{\infty}(U, F) \arrow[r] \arrow[d, "\cong", "\tilde{\tau}'_{wedge}"'] 
        & 0 \\
        0 \arrow[r] 
        & C^\infty(U)\hat{\otimes}_\pi\ker(Id - \gamma) \arrow[r, hook] 
        & C^{\infty}(U) \hat{\otimes}_{\pi} (\hat{\otimes}_{\pi}^2 F) \arrow[r, "q_2"] \arrow[u, "\tilde{\tau}'"', shift right=1ex] 
        & C^{\infty}(U) \hat{\otimes}_{\pi} (\hat\wedge^2 F) \arrow[r] 
        & 0
    \end{tikzcd}
    \caption{Isomorphisms for the infinite-dimensional Serre-Swan theorem.}
    \label{fig:wedge_isomorphism}
\end{figure}
\subsection{The Schouten bracket}\cite{24}
Let $M$ be a smooth manifold modeled on a convenient space $E$. Let $p:TM\to M$ be the kinematic tangent bundle over $M$ of $TM$.

Consider the vector bundle \[\hat\wedge^k TM = \bigcup_{x \in M}  \hat\wedge^k_{\beta} T_xM\] over $M$, where $ \hat\wedge^k_{\beta} T_xM$ denotes the completion of the space of skew-symmetric tensors within $ \otimes^k_{\beta} T_xM$; the space of bornological tensor products of $T_xM$ (see \cite[Chapter 1, 5.7]{20}). The vector bundle  $p':\hat\wedge^k TM\to M$ carries a smooth convenient vector bundle structure over $M$ (see \cite[section 1.3]{24})

\begin{proposition}
Let $M$ be smoothly regular convenient manifold. The unique extension of the map
\begin{align*}
    [-,-]_s : \wedge^k \Gamma (TM) \times \wedge^l \Gamma(TM) \rightarrow \wedge^{k+l-1} \Gamma(TM),\\
[X_1 \wedge \dots \wedge X_k, Y_1 \wedge \dots \wedge Y_l] = \sum_{i,j} (-1)^{i+j} [X_i, Y_j] \wedge X_1 \wedge \dots \hat{X}_i \dots X_k \wedge Y_1 \dots \hat{Y}_j \dots Y_l
\end{align*}
defines a map
\[
[-,-]_s : \Gamma(\hat\wedge^k TM) \times \Gamma(\hat\wedge^l TM) \rightarrow \Gamma(\hat\wedge^{k+l-1} TM),
\]
on the  graded commutative algebra $\bigoplus_{n\geq0}\Gamma(\hat\wedge^n TM)$, which we call the Schouten bracket.
\end{proposition}

\begin{proof}
The map on the simple tensors extends uniquely to the completed projective tensor products of the vector fields by the universal property of the projective tensor product. To see that the Schouten bracket can be extended to $\wedge^k_{C^{\infty}(M)}\Gamma(TM) \times \wedge^l_{C^{\infty}(M)}\Gamma(TM)$, we note that the Schouten bracket is $C^{\infty}(G)$-linear in each component.
\[
[X_1 \wedge \dots f X_i \dots \wedge X_k, Y_1 \wedge \dots Y_l]_s = f[X_1 \dots \wedge X_k, Y_1 \wedge \dots Y_l] + (-1)^{k-i} i(df)(X_1 \wedge \dots X_k) \wedge Y_1 \wedge \dots Y_l
\]
where $i(df)(X_1 \wedge \dots X_k) := \sum_j (-1)^{j-1} df(X_j) \wedge X_1 \dots \hat{X}_j \dots X_k$. Now, using the isomorphism \ref{Iso} in Section 9.1, we can extend $[-,-]_s$ to $\Gamma(\hat\wedge^k TM) \times \Gamma(\hat\wedge^l TM)$.
\subsection{Approximation of left/right-invariant multi-vector fields}
Recall that a bivector field $\pi \in \Gamma(L^2_{\text{skew}}(T'G))$ is called left invariant if $\pi(gh) = L_g^{**}\pi(h)$, and right invariant if $\pi(gh) = R_h^{**}\pi(g)$ for every $g, h \in G$.
\end{proof}

\begin{lemma}\label{8.8}
The space of all sections $\Gamma(L^2_{\text{skew}}(T'G))$ is topologically isomorphic to $C^{\infty}(G, L^2_{\text{skew}}(\mathfrak{g}'))$. Under this isomorphism, the subspace of left-invariant skew-symmetric bivectors $\Gamma_l(L^2_{\text{skew}}(T'G))$ is topologically isomorphic to the vector space $L^2_{\text{skew}}(\mathfrak{g}').$
\end{lemma}

\begin{proof}
We first show that the bundle $L^2_{\text{skew}}(T'G)$ is trivializable as $G \times L^2_{\text{skew}}(\mathfrak{g}').$ Consider the following action of $G$ on $L^2_{\text{skew}}(T'G)$:
\[
R : G \times L^2_{\text{skew}}(T'G) \rightarrow L^2_{\text{skew}}(T'G), \quad (g, \omega) \mapsto R_{g^{-1}}^{**}\omega,
\] 
where $R_{g^{-1}}^{**}$ denotes the pushforward of bivectors by right translation. Since $L^2_{\text{skew}}(T'G)$ and $G \times L^2_{\text{skew}}(\mathfrak{g}')$ are both Fréchet (resp. Silva) manifolds, it suffices to show that the map is smooth in the convenient sense. According to \cite[Theorem 5.19]{20}, $R$ is smooth if and only if for each pair of covectors $(\alpha_1, \alpha_2) \in T'G \times_G T'G$, 
\[ 
\hat{R} : G \times L^2_{\text{skew}}(T'G) \rightarrow \mathbb{R}, \quad (g, \omega) \mapsto (R_{g^{-1}}^{**}\omega)(\alpha_1, \alpha_2) := \omega(R_{g^{-1}}^*\alpha_1, R_{g^{-1}}^*\alpha_2) 
\] 
is smooth for every $g \in G$ and $\omega \in L^2_{\text{skew}}(T'G)$. This holds because the pullback action $R_{g^{-1}}^*$ of $G$ on covectors and the evaluation map are both smooth.

Consequently, letting $p: L^2_{\text{skew}}(T'G) \to G$ denote the bundle base-point projection, the trivialization map 
\[
L^2_{\text{skew}}(T'G) \rightarrow G \times L^2_{\text{skew}}(\mathfrak{g}'), \quad a \mapsto (p(a), R_{p(a)^{-1}}^{**}a)
\] 
is smooth, as the action of $G$ on $L^2_{\text{skew}}(T'G)$ is smooth. It possesses a smooth inverse given by
\[
(g, a_0) \mapsto R_g^{**}a_0.
\] 
We therefore conclude that the space of sections $\Gamma(L^2_{\text{skew}}(T'G)) \cong C^{\infty}(G, L^2_{\text{skew}}(\mathfrak{g}')).$ 

Now, we identify the subspace of left-invariant sections under this isomorphism. Let $\tilde{\pi} \in C^{\infty}(G, L^2_{\text{skew}}(\mathfrak{g}'))$ be the function corresponding to a left-invariant section $\pi \in \Gamma_l(L^2_{\text{skew}}(T'G))$. By the definition of left-invariance, $\pi_g = L_g^{**} \pi_e$. Under the right trivialization, the corresponding function evaluates to:
\[
\tilde{\pi}(g) = R_{g^{-1}}^{**} \pi_g = R_{g^{-1}}^{**} L_g^{**} \pi_e = \operatorname{Ad}^{(2)}_g \pi_e.
\]
Since $\tilde{\pi}(g) = \operatorname{Ad}^{(2)}_g \pi_e$ (where $\operatorname{Ad}^{(2)}_g$ is the induced adjoint action on bivectors), any left-invariant section is uniquely and completely determined by its value at the identity $e$. Therefore, the evaluation map $\pi \mapsto \pi(e)$ yields a topological isomorphism between the space of left-invariant sections $\Gamma_l(L^2_{\text{skew}}(T'G))$ and the fiber $L^2_{\text{skew}}(\mathfrak{g}').$
\end{proof}

\begin{proposition}[Approximation of left/right-invariant multi-vector fields]
A left-invariant (resp. right-invariant) $k$-vector field $\pi \in \Gamma_l(\hat{\wedge}^k TG)$ can be approximated by elements in the algebraic wedge product $\wedge^k \Gamma_l(TG)$ (resp. $\wedge^k \Gamma_r(TG)$), where $\Gamma_l(TG)$ (resp. $\Gamma_r(TG)$) denotes the real vector space of left-invariant (resp. right-invariant) vector fields on $G$.
\end{proposition}

\begin{proof}
We denote by $\Gamma_l(\hat{\wedge}^k TG)$ the subspace of left-invariant elements in the space of sections $\Gamma(\hat{\wedge}^k TG)$. 

By the result of Lemma \ref{8.8}, evaluating a left-invariant section at the identity element $e \in G$ yields a topological isomorphism between the space of left-invariant sections and the fiber over the identity. For the completed exterior algebra of the tangent bundle, this gives the topological isomorphism:
\[
\Gamma_l(\hat{\wedge}^k TG) \cong \hat{\wedge}^k \mathfrak{g}.
\]
By the definition of the completed topological wedge product, the algebraic wedge product $\wedge^k \mathfrak{g}$ is a dense subspace of the completion $\hat{\wedge}^k \mathfrak{g}$. 

Since the Lie algebra $\mathfrak{g}$ is canonically isomorphic as a topological vector space to the space of left-invariant vector fields $\Gamma_l(TG)$, the algebraic wedge product $\wedge^k \mathfrak{g}$ corresponds exactly to $\wedge^k \Gamma_l(TG)$. 

Consequently, any element in the completed space $\hat{\wedge}^k \mathfrak{g}$ can be approximated by elements in $\wedge^k \mathfrak{g}$. Passing this density back through our topological isomorphism, it follows that any left-invariant section $\pi \in \Gamma_l(\hat{\wedge}^k TG)$ can be approximated by elements in the algebraic wedge product $\wedge^k \Gamma_l(TG)$. 

The proof for right-invariant multi-vector fields follows an identical argument by substituting right translations and using $\Gamma_r(TG)$.
\end{proof}

\subsection{Multiplicative bivector and the Schouten bracket}

\begin{lemma}\label{8.10}
$s \in \Gamma(\hat\wedge^k TG)$ is left (right) invariant if and only if $L_X s = 0$ for every right (left) invariant $X \in \Gamma(TG)$.
\end{lemma}

\begin{lemma}\label{8.11}
If $s \in \Gamma(\hat\wedge^k TG)$ is left invariant and $t \in \Gamma(\hat\wedge^l TG)$ is right invariant, then $[t, s]_s = 0$.
\end{lemma}

\begin{lemma}\label{8.12}
$s \in \Gamma(\hat\wedge^k TG)$ is multiplicative if and only if $s(e) = 0$ and for each left invariant section $X \in \Gamma_l(TG)$, $L_X s$ is left invariant. Equivalently, if and only if $s(e) = 0$ for each right invariant section $Y \in \Gamma_r(TG)$, $L_Y s$ is right invariant.
\end{lemma}

\begin{remark}

For proofs of the forward directions of Lemmas \ref{8.10}, \ref{8.11}, and \ref{8.12}, we refer to \cite[Proposition 10.5]{33}. To prove the converse, we observe that since our Lie groups are regular, the exponential map $\exp : \mathfrak{g} \rightarrow G$ exists. 
Following the reasoning in \cite[Proposition 10.5]{33}, the converse can be established within the subspace $\langle \exp(\mathfrak{g}) \rangle$ of the Lie group $G$ generated by $\exp(\mathfrak g)$.

If this subgroup $\langle \exp(\mathfrak{g}) \rangle$ is open in $G$, it is also closed, and thus it coincides with the identity component $G_0$ of $G$. This allows us to extend the result from the infinitesimal level to the identity component $G_0$, and thus to the entire group if $G$ is connected. For instance, $\langle \exp(\mathfrak{g}) \rangle$ is an open subset of $G$ when $G$ is a smooth (or analytic) loop group of a finite-dimensional Lie group, or when $G$ is the group of smooth (or analytic) diffeomorphisms of the circle $\mathbb{S}^1$.
\end{remark}

\begin{proposition}\label{[]=0}
If $\pi \in \Gamma(\hat\wedge^2 TG)$ is a multiplicative bivector, then $[\pi, \pi]_s \in \Gamma(\hat\wedge^3 TG)$ is a multiplicative bivector.
\end{proposition}

\begin{proof}
From the graded Jacobi identity of the Schouten bracket (\cite[Theorem 1.4]{24}), it follows that the Lie derivative of the Schouten bracket $[\pi, \pi']_s$ satisfies
\[
L_X[\pi, \pi'] = [L_X\pi, \pi']_s + [\pi, L_X\pi']
\]
for $X \in \Gamma(TG), \pi \in \Gamma(\hat\wedge^k TG), \pi' \in \Gamma(\hat\wedge^l TG)$. Now, if $\pi \in \Gamma(\hat\wedge^2 TG)$ is multiplicative, then for a left-invariant vector field $X \in \Gamma_l(TG)$, and a right invariant vector field $Y \in \Gamma_r(TG)$, we have
\[
L_Y L_X[\pi, \pi] = 2([L_Y L_X\pi, \pi]_s + [L_Y\pi, L_X\pi]) = 0
\]
which follows from the last three lemmas. By Lemma 6.8, we conclude that $[\pi, \pi]_s$ is multiplicative.
\end{proof}

\subsection{Drinfeld theorem for regular Poisson Lie groups modeled on a nuclear Fréchet space or a nuclear Silva space}
In this section we will prove a special case of Theorem \ref{D-3} when a Lie group is modeled on a nuclear Fréchet space or a nuclear Silva space.
Let $G$ be smooth Lie group with the Lie algebra $\mathfrak g$, modeled on a nuclear Fr\'echet or a nuclear Silva space.   We begin by showing that the vector bundle $L^2_{skew}(T'G)$ is topologically isomorphic to $\hat\wedge^k TG$. To establish this, we first identify the fiber over identity $L^2_{skew}(\mathfrak g')$ with the space $\hat\wedge^k \mathfrak g$ using the following two lemmas.
\begin{lemma}
If $E$ is a nuclear Fréchet or a nuclear Silva space, then $L^2_b(E'_b)$ is nuclear Fréchet or nuclear Silva resp. and there is a topological isomorphism between
\[
L^2_c(E'_c) \cong L^2_b(E'_b) \cong E \hat{\otimes}_{\pi} E.
\]
\end{lemma}
where $E \hat{\otimes}_{\pi} E$ denotes the completed projective tensor product of $E$ with itself. Moreover, \[ L^2_{skew}(E')\cong \hat\wedge^2E\] where $\hat\wedge^2E$ denotes the space of completed skew-symmetric tensors in $E \hat{\otimes}_{\pi} E$.

\begin{proof}
We first note that every nuclear space is a semi-Montel space, according to \cite[Propositions 50.2, 50.12]{31}. Being either a Fréchet or a Silva space, $E$ is necessarily barrelled. Consequently, $E$ is a semi-Montel, barrelled, and quasi-complete space, which implies that $E$ is a Montel space and, therefore, reflexive.

Since every bounded subset of a Montel space is relatively compact, the topology of uniform convergence on bounded subsets coincides with the topology of uniform convergence on compact subsets. This equivalence yields the topological isomorphism $L^2_c(E'_c) \cong L^2_b(E'_b)$. 

Let $E$ be a nuclear Fréchet space and set $F = E'_b$. Since $F$ is a Silva space, it is reflexive with respect to the topology of uniform convergence on bounded subsets on $F'_b$. By \cite[9.9 Corollary 1]{28}, $F \hat{\otimes}_{\pi} F$ is reflexive, and the following are topological isomorphisms:
\begin{align*}
(F \hat{\otimes}_{\pi} F) &\cong ((F'_b \hat{\otimes}_{\pi} F'_b)'_b)'_b, \\
F'_b \hat{\otimes}_{\pi} F'_b &\cong E \hat{\otimes}_{\pi} E.
\end{align*}
Since $F'_b = E$ is a nuclear Fréchet space, it follows that $(F'_b \hat{\otimes}_{\pi} F'_b)'_b \cong (F'_b)'_b \hat{\otimes}_{\pi} (F'_b)'_b \cong F \hat{\otimes}_{\pi} F$. Thus, we obtain the topological isomorphism $(F \hat{\otimes}_{\pi} F)'_b \cong E \hat{\otimes}_{\pi} E$. Furthermore, we have the identification $(F \hat{\otimes}_{\pi} F)'_b \cong L^2(E'_b)_b$. Thus, $L^2(E'_b)_b\cong E \hat{\otimes}_{\pi} E.$

Now, if $E$ is a nuclear Silva space, then $F := E'$ is a nuclear Fréchet space. According to \cite[Theorem 9.9]{28}, there exists a topological isomorphism between $F'_b \hat{\otimes}_{\pi} F'_b$ and $L^2(F)_b$. Since $E$ is a nuclear Silva space, it is reflexive, which implies $F'_b = (E'_b)'_b \cong E$. Thus, $L^2(E'_b)_b\cong E \hat{\otimes}_{\pi} E$. \\
\textit{Extension to the skew-symmetric subspaces:}
To see that this isomorphism extends to the wedge product, consider the continuous flip map $\tau: E \hat{\otimes}_{\pi} E \to E \hat{\otimes}_{\pi} E$, which is the unique extension of $x \otimes y \mapsto y \otimes x$. Similarly, define the transposition map $\sigma: L^2(E'_b) \to L^2(E'_b)$ by $\sigma(B)(\lambda, \mu) = B(\mu, \lambda)$.

The isomorphism $\Psi: E \hat{\otimes}_{\pi} E \to L^2(E'_b)$ constructed above satisfies the equivariance condition $\Psi \circ \tau = \sigma \circ \Psi$. The space $\hat\wedge^2 E$ is defined as the closed subspace of antisymmetric tensors (the $(-1)$-eigenspace of $\tau$), and $L^2_{\text{skew}}(E'_b)$ is the subspace of alternating bilinear forms (the $(-1)$-eigenspace of $\sigma$). Since $\Psi$ is a topological isomorphism that maps the eigenspaces of $\tau$ to the corresponding eigenspaces of $\sigma$, it restricts to a topological isomorphism between the skew-symmetric subspaces:
\[
\hat\wedge^2 E \cong L^2_{\text{skew}}(E'_b).
\]
\end{proof}
\begin{lemma}\label{Lemma9.13}
If $\mathfrak{g}$ is a nuclear  Fréchet (or Silva) space, then $\hat\wedge^2 \mathfrak g$ is also endowed with a Fréchet (resp. Silva) topological vector space structure.
\end{lemma}

\begin{proof}
\textit{Case 1:}
Suppose $\mathfrak{g}$ is a nuclear Fréchet space. Let $\tau : \mathfrak{g} \otimes \mathfrak{g} \rightarrow \mathfrak{g} \otimes \mathfrak{g}$ be the flip map defined by $a \otimes b \mapsto b \otimes a$, and let $\tau$ also denote its unique continuous extension $\tau : \mathfrak{g} \hat{\otimes}_{\pi} \mathfrak{g} \rightarrow \mathfrak{g} \hat{\otimes}_{\pi} \mathfrak{g}$. We define the exterior power as the quotient $\hat\wedge^2_{\pi} \mathfrak{g} := (\mathfrak{g} \hat{\otimes}_{\pi} \mathfrak{g})/\ker(\tau)$. Since $\mathfrak{g} \hat{\otimes}_{\pi} \mathfrak{g}$ is a Fréchet space and $\ker(\tau)$ is a closed subspace, the quotient $\hat\wedge^2_{\pi} \mathfrak{g}$ is itself a Fréchet space.

\textit{Case 2:}
Suppose $\mathfrak{g}$ is a nuclear Silva space. Let $E$ be a Silva space and $q : E \rightarrow E/\ker(\tau)$ be the quotient map. We will show that $W := E/\ker(\tau) = \varinjlim E_n/(E_n \cap \ker(\tau))$ is a regular direct limit of Banach spaces. 

Recall that $E = \varinjlim E_n$ is a regular direct limit of Banach spaces $E_n$ with compact linking maps $i_n : E_n \rightarrow E_{n+1}$. The topology on $E$ is final with respect to the inclusion maps $q_n : E_n \rightarrow E$, and the topology on $E/\ker(\tau)$ is final with respect to the quotient map $q$. Consequently, the topology on $E/\ker(\tau)$ is final with respect to the compositions $\tilde{q}_n := q \circ q_n$. These maps lift to injections $\tilde{i}_n : E_n/(\ker(\tau) \cap E_n) \rightarrow E/\ker(\tau)$, and the topology on $E/\ker(\tau)$ is final with respect to these induced mappings.

It remains to show that the induced linking maps are compact. Let $p_n : E_n \to E_n/(\ker(\tau) \cap E_n)$ be the canonical quotient map. Since $p_n$ is a continuous open map between Banach spaces, it maps bounded sets to bounded sets. Let $B \subset E_n/(\ker(\tau) \cap E_n)$ be a bounded set. We note that that $p_n^{-1}(B)$ is bounded in $E_n$, indeed,  let $V \subset E_n$ be a 0-neighborhood, then ${p}_n(V)$ is a 0-neighborhood of $\ker(\tau) \cap E_n$, hence there exists some $\lambda \in \mathbb{R}_+$ such that $B \in \lambda {p}_n(V)$, this implies ${p}_n^{-1}(B) \in \lambda V$. Thus, $p_n^{-1}(B)$ is bounded.

Since $i_n$ is a compact operator, the image $i_n(p_n^{-1}(B))$ is relatively compact in $E_{n+1}$. Because the linking map for the quotient spaces is given by $\tilde{i}_n(B) = p_{n+1}(i_n(p_n^{-1}(B)))$, and $p_{n+1}$ is continuous, $\tilde{i}_n(B)$ is relatively compact. Thus, $\hat\wedge^2_{\pi} \mathfrak{g}$ is a Silva space.\end{proof}

\textit{Nuclear Fr\'echet (res. Silva ) vector bundle structure on $L^2_{skew}(T'G)\cong\hat\wedge^2 TG$ .}\\
We observe that the canonical map $T : \hat\wedge^2 TG \rightarrow G \times \hat\wedge^2_{\pi} \mathfrak{g}$, defined by $T(\alpha) := (p(\alpha), dR_{p(\alpha)^{-1}}(p(\alpha))\alpha)$, is a bijection. Since the right translation map $dR_{g^{-1}}$ is a topological isomorphism for each $g \in G$, it commutes with the completion of the exterior power. This allows us to endow $\hat\wedge^2 TG$ with a smooth vector bundle structure (following \cite[3.7]{11}), which admits a global trivialization with the typical fiber $\hat\wedge^2_{\pi} \mathfrak{g}$.

We combine the results of last two lemmas to obtain a topological identification of vector bundle $L^2_{skew}(T'G)$ and $\hat\wedge^2 TG$.
\begin{proposition}
  Let $G$ be smooth Lie group with the Lie algebra $\mathfrak g$, modeled on a nuclear Fr\'echet or a nuclear Silva space.  Then the vector bundle $L^2_{skew}(T'G)$ is topologically isomorphic to $\hat\wedge^2 TG$.
\end{proposition}

We are now equipped to prove the Drinfeld theorem for a Lie group modeled on a nuclear Fréchet or a nuclear Silva space.

Let $G$ be a 1-connected Lie group modeled on a nuclear Fréchet or a nuclear Silva space, and let $\mathfrak{g}$ be its Lie algebra. Suppose $(G, \pi)$ is a smooth Poisson Lie group structure, where $\pi \in \Gamma(\hat\wedge^2 TG)$. By Theorem \ref{D-0}, there exists a continuous Lie bialgebra structure $(\mathfrak{g}, \mathfrak{g}', \delta)$ on $\mathfrak{g}$.

Conversely, assume $G$ is the unique 1-connected regular Lie group corresponding to the Lie algebra $\mathfrak{g}$, and let $(\mathfrak{g}, \mathfrak{g}', \delta)$ be a continuous Lie bialgebra structure on $\mathfrak{g}$. The cocycle $\delta$ integrates to a multiplicative bivector $\pi \in \Gamma(\hat\wedge^2 TG)$ by the results of Lemma \ref{4.4}, Lemma \ref{4.5}, and Proposition \ref{proposition 4.10} in section \ref{3}. To conclude that $(G, T'G, \pi)$ constitutes a Poisson structure on $G$, we must show that $\pi$ satisfies the Jacobi identity. This follows from the following proposition.

\begin{proposition}
Let $G$ be a 1-connected regular Lie group modeled on a nuclear Fréchet or a nuclear Silva Lie algebra $\mathfrak{g}$. If $(\mathfrak{g}, \mathfrak{g}', \delta)$ is a Lie bialgebra structure, then $\delta$ integrates to a smooth multiplicative bivector $\pi \in \Gamma(\hat{\wedge}^2 TG)$ such that $\pi$ is a Poisson bivector (which means the bracket $\{f, g\} := \pi(df, dg)$ satisfies the Jacobi identity). Consequently, $(G, T'G, \pi)$ is a Poisson Lie group structure.
\end{proposition}

\begin{proof}
By \cite[Theorem 3.6]{22}, it follows that $\pi$ is a Poisson bivector if and only if its Schouten bracket vanishes: $[\pi, \pi]_S = 0$.

Since $[\pi, \pi]_S$ is multiplicative (by \cref{[]=0}), the map $\theta(x) := R_{x^{-1}}^{**}[\pi, \pi]_S(x)$ defines a smooth group 1-cocycle. Consequently, the map $\hat{\theta}(x) := (x, \theta(x))$ defines a Lie group homomorphism into the semidirect product
\[
\hat{\theta} : G \rightarrow G \rtimes L^3_{\text{skew}}(\mathfrak{g}').
\]

Differentiating at the identity, we obtain a Lie algebra homomorphism $\hat{\delta} := d_e \hat{\theta} : \mathfrak{g} \rightarrow \mathfrak{g} \rtimes L^3_{\text{skew}}(\mathfrak{g}')$. It follows that the second component, $[\delta, \delta] := \operatorname{pr}_2(\hat{\delta})$, is a Lie algebra 1-cocycle relative to the adjoint action $\operatorname{ad}^{(3)}$ of $\mathfrak{g}$ on $L^3_{\text{skew}}(\mathfrak{g}')$ given by 
\begin{align*}
\operatorname{ad}^{(3)} &: \mathfrak{g} \to \mathfrak{gl}(L^3_{\text{skew}}(\mathfrak{g}')), \\
\operatorname{ad}^{(3)}_X \eta (\alpha, \beta, \gamma) &= \eta(\operatorname{ad}^*_X \alpha, \beta, \gamma) + \eta(\alpha, \operatorname{ad}^*_X \beta, \gamma) + \eta(\alpha, \beta, \operatorname{ad}^*_X \gamma)
\end{align*}
for $X \in \mathfrak{g}$, $\eta \in L^3_{\text{skew}}(\mathfrak{g}')$, and $\alpha, \beta, \gamma \in \mathfrak{g}'$. 

Evaluating this cocycle on an element $X \in \mathfrak{g}$ and dual elements $\alpha, \beta, \gamma \in \mathfrak{g}'$, we obtain
\[
[\delta, \delta](X)(\alpha, \beta, \gamma) = [ [\alpha, \beta]_{\mathfrak{g}'}, \gamma]_{\mathfrak{g}'} (X) + \text{cyclic permutations}(\alpha, \beta, \gamma).
\]
Since the expression on the right-hand side of the equality corresponds exactly to the evaluation of $X$ on the Jacobiator of the Lie bracket on $\mathfrak{g}'$, and $\mathfrak{g}'$ satisfies the Jacobi identity by definition, we obtain $[\delta, \delta] = 0$. 

Thus, the Lie algebra homomorphism $\hat{\delta}$ is trivial in its second component. Since $G$ is 1-connected, this implies that the corresponding Lie group homomorphism $\hat{\theta}$, and consequently the Lie group 1-cocycle $\theta$, are also trivial. Therefore, $\theta(x) = 0$ for all $x \in G$, which implies $[\pi, \pi]_S = 0$ everywhere.
\end{proof}
 We rewrite the Drinfeld correspondence in the context of this section in the following theorem:

\begin{theorem}\label{D-4}
Let $G$ be a Lie group modeled on a nuclear Fréchet or a nuclear Silva space, let $\mathfrak{g}$ be the Lie algebra of $G$ and $(G, T'G, \pi)$ be a smooth Poisson Lie group structure, where $\pi \in {\Gamma(\hat\wedge}^2TG)$. Then there exists a continuous Lie bialgebra structure $(\mathfrak{g}, \mathfrak{g}', \delta)$ on $\mathfrak{g}$. Conversely, if $G$ is the unique 1-connected regular Lie group of a Lie algebra $\mathfrak{g}$ and $(\mathfrak{g}, \mathfrak{g}', \delta)$ is a continuous Lie bialgebra structure on $\mathfrak{g}$. Then there exists a unique smooth Poisson Lie group structure $(G, T'G, \pi)$ on $G$ such that $\pi \in {\Gamma(\hat\wedge}^2TG)$ is a smooth bivector.
\end{theorem}

\subsection{Examples}
The main examples fit in the setting of this setting include smooth loop groups $C^{\infty}(\mathbb{S}^{1},G)$, analytic loop groups $C^{\omega}(\mathbb{S}^{1},G)$ of a 1-connected real Lie group $G$, as well as $\mathrm{Diff}^{\infty}(M)$ and $\mathrm{Diff}^{\omega}(M)$.

We recall that assumptions on a Lie groups $G$ for which we proved the Drinfeld correspondence are that the modeling space of $G$ must be a nuclear Fréchet or a nuclear Silva space, $G$ should be a 1-connected regular Lie group, and the group generated by the subspace $\exp(\mathfrak{g})$ of $G$ must be equal to $G$.\\

\textit{Examples: (1) $LG := C^{\infty}(\mathbb{S}^1, G)$, (2) $ALG := C^{\omega}(\mathbb{S}^1, G)$. }
\begin{theorem}
Let $G$ be a real smooth Lie group, then the smooth loop group $LG := C^{\infty}(\mathbb{S}^1, G)$ is a regular Lie group with the Lie algebra $C^{\infty}(\mathbb{S}^1, \mathfrak{g})$, which is a nuclear Fréchet space. And the analytic loop group $ALG := C^{\omega}(\mathbb{S}^1, G)$ is a regular Lie group with Lie algebra $C^{\omega}(\mathbb{S}^1, \mathfrak{g})$, which is a nuclear Silva topological space. To see the regularity of $C^{\infty}(\mathbb{S}^1, G)$ and $C^{\omega}(\mathbb{S}^1, G)$, we refer to \cite{12} Examples a and f.
\end{theorem}

\begin{lemma}
$C^{\infty}(\mathbb{S}^1, \mathfrak{g})$ is a nuclear Fréchet space (see also \cite[Lemma 30.3]{20}).
\end{lemma}

\begin{proof}
 $C^{\infty}(\mathbb{S}^1, \mathbb{R})$ is topologically isomorphic to the space of rapidly decreasing real sequences $\mathfrak s$ by \cite[Theorem 51.3]{31}. Therefore, $C^{\infty}(\mathbb{S}^1, \mathbb{R})$ is a nuclear Fréchet space (see \cite[Remark 51.1]{31}).
\end{proof}

\begin{proposition}
$C^{\omega}(\mathbb{S}^1, \mathfrak{g})$ is a nuclear Silva space (see also \cite[Lemma 30.3]{20}).
\end{proposition}

\begin{proof}
It suffices to show that the scalar-valued space $C^{\omega}(\mathbb{S}^1, \mathbb{R})$ is a nuclear Silva space, as tensoring with the finite-dimensional space $\mathfrak{g}$ preserves this property (see also \cite[Theorem 11.4]{20}). 

Let $A_r = \{z \in \mathbb{C} \mid 1-r < |z| < 1+r\}$ be the open annulus of thickness $r \in (0, 1)$ centered around the unit circle. We define the Hilbert space of square-integrable holomorphic functions on this annulus as
\[
H^2(A_r) := \left\{ f : A_r \rightarrow \mathbb{C} \mid f \text{ is holomorphic and } \int_{A_r} |f(z)|^2 \, dA < \infty \right\}.
\]
We will show that $C^{\omega}(\mathbb{S}^1, \mathbb{R})$ is the locally convex direct limit of $H^2(A_r)$ as $r \to 0$, and that the dual linking maps are nuclear (specifically, trace-class) operators.

\paragraph{Direct limit of Hilbert spaces.}
Let $f \in C^{\omega}(\mathbb{S}^1, \mathbb{R})$. Then there exists some $r \in (0, 1)$ such that $f$ extends to a bounded holomorphic function $F$ on the annulus $A_r$. Writing the Fourier expansion on the circle
\[
f(e^{i\theta}) = \sum_{k \in \mathbb{Z}} \hat{f}(k)e^{ik\theta},
\]
the fact that $f$ admits a holomorphic extension to $A_r$ implies, by the Paley–Wiener theorem (see, for example, \cite[Theorem 3.1]{29}), that there are constants $C, \lambda > 0$ such that
\[
|\hat{f}(k)| \le Ce^{-\lambda|k|} \quad \text{for all } k \in \mathbb{Z}.
\]
This exponential decay ensures that $\sum_{k \in \mathbb{Z}} |\hat{f}(k)|^2 < \infty$. By Parseval’s identity, the holomorphic extension $F$ has finite $L^2$-norm over $A_{r'}$ for any $0 < r' < r$, meaning $F \in H^2(A_{r'})$. We therefore obtain the topological identification
\[
C^{\omega}(\mathbb{S}^1, \mathbb{R}) = \varinjlim_{r \to 0} H^2(A_r)
\]
as a locally convex regular direct limit of Hilbert spaces.

\paragraph{Nuclearity of the dual space.}
Since each $H^2(A_r)$ is a Hilbert space, by the Riesz Representation theorem, it is anti-linearly isometric to its dual $H^2(A_r)'$. The dual space of the direct limit is exactly the projective limit of the duals:
\[
C^{\omega}(\mathbb{S}^1, \mathbb{R})' = \varprojlim_{r \to 0} H^2(A_r)'.
\]
Consider two radii $0 < r < R < 1$. The linking map in the direct system is the inclusion $H^2(A_R) \hookrightarrow H^2(A_r)$. The corresponding dual linking map is the restriction operator $T : H^2(A_R) \rightarrow H^2(A_r)$ defined by $T(f) = f|_{A_r}$. We will show that $T$ is a trace-class operator.

The monomials $\{z^n \mid n \in \mathbb{Z}\}$ form an orthogonal basis for both $H^2(A_R)$ and $H^2(A_r)$. Let us analyze the operator $T^*T : H^2(A_R) \to H^2(A_R)$. First, we observe that $T$ is equivariant with respect to rotations $R_\theta(z) = e^{i\theta}z$. For any $f, g \in H^2(A_R)$, the inner product on $H^2(A_r)$ yields
\begin{align*}
\langle T^*T (R_\theta f), g \rangle_R &= \langle T(R_\theta f), T(g) \rangle_r \\
&= \langle R_\theta (Tf), Tg \rangle_r \\
&= \langle Tf, R_{-\theta} (Tg) \rangle_r \\
&= \langle Tf, T(R_{-\theta} g) \rangle_r \\
&= \langle T^*T f, R_{-\theta} g \rangle_R \\
&= \langle R_\theta (T^*T f), g \rangle_R.
\end{align*}
Since this holds for all $g$, we have $T^*T \circ R_\theta = R_\theta \circ T^*T$, meaning $T^*T$ commutes with rotations. Because the monomials $z^n$ are exactly the eigenvectors of the rotation action, they must also be the eigenvectors of $T^*T$. Thus, $T^*T(z^n) = c_n z^n$ for some eigenvalues $c_n > 0$.

We compute $c_n$ explicitly by taking the inner product:
\[
c_n = \frac{\langle T^*T(z^n), z^n \rangle_R}{\|z^n\|^2_R} = \frac{\langle T(z^n), T(z^n) \rangle_r}{\|z^n\|^2_R} = \frac{\int_{A_r} |z|^{2n} \, dA}{\int_{A_R} |z|^{2n} \, dA}.
\]
Evaluating the area integrals in polar coordinates ($dA = s \, ds \, d\theta$), we get:
\[
c_n = \frac{\int_0^{2\pi} \int_{1-r}^{1+r} s^{2n+1} \, ds \, d\theta}{\int_0^{2\pi} \int_{1-R}^{1+R} s^{2n+1} \, ds \, d\theta} = \frac{(1+r)^{2n+2} - (1-r)^{2n+2}}{(1+R)^{2n+2} - (1-R)^{2n+2}}.
\]

To prove that $T$ is a trace-class operator, we must show that the sum of its singular values converges, i.e., $\sum_{n \in \mathbb{Z}} \sqrt{c_n} < \infty$. We apply the ratio test to the sequence $a_n = \sqrt{c_n}$ as $n \to \infty$:
\[
\lim_{n \rightarrow \infty} \frac{a_{n+1}}{a_n} = \lim_{n \rightarrow \infty} \sqrt{\frac{c_{n+1}}{c_n}} = \lim_{n \to \infty} \sqrt{\frac{(1+r)^{2n+4}}{(1+R)^{2n+4}}} = \frac{1+r}{1+R}.
\]
Because $r < R$, the ratio $\frac{1+r}{1+R}$ is strictly less than $1$. An identical symmetric limit holds for $n \to -\infty$ using the $(1-r)$ and $(1-R)$ terms. By the ratio test, the series $\sum \sqrt{c_n}$ converges absolutely, proving that $T$ is a trace-class operator.

Since $C^{\omega}(\mathbb{S}^1, \mathbb{R})'$ is a projective limit of Hilbert spaces where each linking map is trace-class, \cite[Chapter 7.3, Corollary 3]{28} proves that $C^{\omega}(\mathbb{S}^1, \mathbb{R})'$ is a nuclear Fréchet space. By \cite[Chapter IV, Theorem 9.6]{28}, the strong dual of a nuclear Fréchet space is a nuclear Silva space. 

Furthermore, because $C^{\omega}(\mathbb{S}^1, \mathbb{R})'$ is a nuclear Fréchet space, it is a Montel space. The strong dual of a Montel space is again Montel, and Montel spaces are reflexive. Therefore, the canonical evaluation map yields the topological isomorphism
\[
C^{\omega}(\mathbb{S}^1, \mathbb{R}) \cong C^{\omega}(\mathbb{S}^1, \mathbb{R})''.
\]
Thus, the original space $C^{\omega}(\mathbb{S}^1, \mathbb{R})$ is a nuclear Silva space.
\end{proof}

\textit{Examples: (3) $\widetilde{\mathrm{Diff}^\infty(M)_0}$, (4) $\widetilde{\mathrm{Diff}^\omega(M)_0}$.} \\
Let $M$ be a real compact connected manifold of dimension $n$. The exponential maps $\exp : \Gamma^\infty(TM) \to \mathrm{Diff}^\infty(M)$ and $\exp : \Gamma^\omega(TM) \to \mathrm{Diff}^\omega(M)$ are not necessarily local diffeomorphisms, but we will see that the subgroup generated by the image $\exp(\Gamma^\infty(TM))$ (respectively, $\exp(\Gamma^\omega(TM))$) equals the identity component $\mathrm{Diff}^\infty(M)_0$ (respectively, $\mathrm{Diff}^\omega(M)_0$). Let $\widetilde{\exp}: \Gamma^\infty(TM) \to \widetilde{\mathrm{Diff}^\infty(M)_0}$ be the lifting of the exponential map to the universal covering group. Then, the group generated by $\widetilde{\exp}(\Gamma^\infty(TM))$ (respectively, $\widetilde{\exp}(\Gamma^\omega(TM))$) equals $\widetilde{\mathrm{Diff}^\infty(M)_0}$ (respectively, $\widetilde{\mathrm{Diff}^\omega(M)_0}$).

\begin{remark}
Let $M$ be a compact, connected smooth manifold. Although the exponential map $\exp: \Gamma^\infty(TM) \to \mathrm{Diff}^\infty(M)$ is not locally surjective (see \cite[p. 9]{25}), the subgroup $H := \langle \exp(\Gamma^\infty(TM)) \rangle$ generated by its image equals the entire identity component $\mathrm{Diff}^\infty(M)_0$. 

To see this, note that $H$ is a normal subgroup of $\mathrm{Diff}^\infty(M)_0$. Indeed, for any diffeomorphism $\varphi \in \mathrm{Diff}^\infty(M)_0$ and any vector field $X \in \Gamma^\infty(TM)$, conjugation yields $\varphi \circ \exp(X) \circ \varphi^{-1} = \exp(\varphi_* X)$, where the pushforward $\varphi_* X$ is again a smooth vector field in $\Gamma^\infty(TM)$. By \cite[Theorem 1]{38} of Thurston, the identity component of the diffeomorphism group of a connected manifold, $\mathrm{Diff}^\infty(M)_0$, is a simple group. Since $H$ is a non-trivial normal subgroup, it cannot be a proper subgroup, and thus we must conclude that $H = \mathrm{Diff}^\infty(M)_0$.
\end{remark}

\begin{remark}
Let $M$ be a real-analytic $U(1)$-fibered compact, connected smooth manifold. Let $G = \mathrm{Diff}^\omega(M)$ be the Lie group of real-analytic diffeomorphisms on $M$, modeled on the nuclear Silva space $\mathfrak{g} = \Gamma^\omega(TM)$. Let $H = \langle \exp(\mathfrak{g}) \rangle$ denote the subgroup generated by the image of the exponential map. 

While the exponential map $\exp : \Gamma^\omega(TM) \to \mathrm{Diff}^\omega(M)$ is not a local diffeomorphism in the Silva topology of $\mathrm{Diff}^\omega(M)$, Tsuboi's Regimentation lemma implies that $H$ contains an open neighborhood of the identity.

Let $\{\xi_1, \dots, \xi_N\}$ be a finite set of analytic vector fields spanning the tangent space $T_xM$ at any point $x \in M$ (\cite[Proposition 8.1]{37}). For each choice $\kappa$ of $n$ vector fields among $\{\xi_1, \dots, \xi_N\}$, let $\Delta_\kappa$ be the determinant of the corresponding vector fields. We construct analytic partition functions $\mu_\kappa = (\Delta_\kappa)^4 / \sum_\lambda (\Delta_\lambda)^4$ (see \cite[Section 8]{37}).
    
By the Regimentation lemma \cite[Theorem 7.1]{37}, any $\varphi \in G$ sufficiently close to the identity admits a decomposition $\varphi = \prod_\kappa g_\kappa$ such that each factor $g_\kappa$ satisfies the condition that $g_\kappa - \mathrm{id}$ is divisible by $\mu_\kappa$.
    
Since $\mu_\kappa$ contains the factor $(\Delta_\kappa)^4$, the map $g_\kappa$ satisfies the hypothesis of the inverse function theorem for singular Jacobians ($g_\kappa - \mathrm{id}$ divisible by $\Delta^r$ for $r \ge 3$, \cite[Theorem 6.12]{37}). This theorem guarantees that each regimented factor $g_\kappa$ can be expressed as a finite composition of flows
\[ 
g_\kappa(x) = \left( \exp(t_1 \xi_{k_1}) \circ \cdots \circ \exp(t_n \xi_{k_n}) \right)(x), 
\]
where $t_i(x)$ are real-analytic functions. The map $x \mapsto \exp(t_i(x) \xi_{k_i})(x)$ corresponds to the time-1 flow of the global analytic vector field $X_i = t_i \cdot \xi_{k_i} \in \mathfrak{g}$.

Consequently, every $\varphi$ in a sufficiently small neighborhood of the identity belongs to $H$, implying that $H$ is an open subgroup of $G$. Since the identity component $\mathrm{Diff}^\omega(M)_0$ is connected, an open subgroup must be the entire component, meaning $H = \mathrm{Diff}^\omega(M)_0$.
\end{remark}

\begin{lemma}
Let $M$ be a compact manifold. The group of smooth diffeomorphisms $\mathrm{Diff}^{\infty}(M)$ and the group of real-analytic diffeomorphisms $\mathrm{Diff}^{\omega}(M)$ are regular Lie groups modeled on the space of smooth vector fields $\Gamma^{\infty}(TM)$ and the space of real-analytic vector fields $\Gamma^{\omega}(TM)$, respectively. Furthermore, $\Gamma^{\infty}(TM)$ is a nuclear Fréchet space, while $\Gamma^{\omega}(TM)$ is a nuclear Silva space.
\end{lemma}

\begin{proof}
If $M$ is a compact manifold of dimension $n$, $\mathrm{Diff}^{\infty}(M)$ is a Lie group with Lie algebra $\Gamma^{\infty}(TM)$. Locally, this space is isomorphic to $C^{\infty}(U, \mathbb{R}^n)$, where $U \subseteq M$ is an open subset. According to \cite[Proposition 1.7.16]{11}, $C^{\infty}(U, \mathbb{R}^n)$ is a Fréchet space. The nuclearity of $\Gamma^{\infty}(TM)$ follows from \cite[Theorem 51.5]{31} and its corollary (see also \cite[Lemma 30.3]{20}). The regularity of $\mathrm{Diff}^{\infty}(M)$ is established in \cite[Theorem 43.1]{20}.

In the real-analytic case, $\mathrm{Diff}^{\omega}(M)$ is modeled on the Lie algebra $\Gamma^{\omega}(TM)$. By \cite[Theorem 11.4]{20}, $\Gamma^{\omega}(TM)$ is a nuclear Silva space. While the regularity of analytic diffeomorphism groups can be subtle in Milnor's sense (see \cite{13}), it follows from \cite[Theorem 43.4]{20} that $\mathrm{Diff}^{\omega}(M)$ is a regular Lie group in the sense of Kriegl and Michor \cite[Definition 38.4]{20}.
\end{proof}
\section*{Acknowledgments}
First and foremost, I thank Prof. Dr. Helge Glöckner for many important discussions and hints. 
Furthermore, I am very grateful to Prof. Peter W. Michor, Prof. Alice Barbora Tumpach, Prof. Tomasz Goli{\'n}ski, Dr. Stepan Maximov, Dr. Milan Niestijl and Dr. Dawid Kern for their helpful discussions and motivation during various stages of this research.
This research was supported by the German Research Foundation (Deutsche Forschungsgemeinschaft), SFB-TRR 358/1 2023 – 491392403.

\appendix

\section{Concepts and notation of infinite-dimensional calculus}
\label{appendix:calculus}

\subsection{Topological properties of spaces}
\begin{enumerate}
    \item A Hausdorff topological space $X$ is a \textit{k}-space if a subset $A \subseteq X$ is closed if and only if $A \cap K$ is closed in $K$ for each compact subset $K \subseteq X$. Every metrizable topological space is a \textit{k}-space.
    \item A Hausdorff space $X$ is a \textit{$k_{\mathbb{R}}$-space} if a function $f: X \rightarrow \mathbb{R}$ is continuous if and only if it is $k$-continuous (continuous on compact subsets). Every \textit{k}-space is a $k_{\mathbb{R}}$-space.
    \item A space $X$ is a \textit{$k^{\infty}$-space} if the cartesian power $X^{n}$ is a \textit{k}-space for each $n \in \mathbb{N}$. A Hausdorff space $X$ is called \textit{hemicompact} if $X=\bigcup_{n\in \mathbb{N}}K_n$ for a sequence $K_1\subseteq K_2\subseteq\cdots$ of compact sets $K_n\subseteq X$ such that each compact subset of $X$ is a subset of some $K_n$. Hemicompact \textit{k}-spaces are also called $k_{\omega}$-spaces.
    \item A \textit{Silva space} (or DFS-space) is the locally convex inductive limit of a sequence of Banach spaces $E_{1} \subseteq E_{2} \subseteq \dots$ with compact bonding maps. Every Silva space is a $k_{\omega}$-space. A nuclear Silva space is the inductive limit of a sequence of Hilbert spaces with nuclear (trace-class) bonding maps.
    \item A \textit{nuclear Fréchet space} is the projective limit of a sequence of Hilbert spaces with nuclear (trace-class) bonding maps.
    \item A \textit{Montel space} is a barrelled topological vector space with the Heine-Borel property (that means every closed and bounded set is compact). A nuclear Fréchet or a nuclear Silva space is Montel.
\end{enumerate}

\subsection{Differential calculus framework}
We work in the setting of Keller's $C_{c}^{r}$-calculus (see \cite{11}) and Kriegl-Michor's convenient setting of infinite-dimensional differential calculus (see \cite{20}).

\begin{definition}
Let $\mathbb{K} \in \{\mathbb{R}, \mathbb{C}\}$ and let $E, F$ be locally convex topological vector spaces. Let $U \subseteq E$ be an open set and $r \in \mathbb{N}_{0}$. A map $f: U \rightarrow F$ is called $C_{\mathbb{K}}^{0}$ if it is continuous. A continuous map $f: U \rightarrow F$ is called a $C_{\mathbb{K}}^{r}$ map if the iterated directional derivatives
\[
d^{k}f(x, y_{1}, \dots, y_{k}) := D_{y_{k}} \cdots D_{y_{1}}f(x)
\]
exist for all $k \in \mathbb{N}$ with $k \le r$, for all $x \in U$, and all $y_{1}, \dots, y_{k} \in E$, where the maps
\[
d^{k}f: U \times E^{k} \rightarrow F, \quad (x, y_{1}, \dots, y_{k}) \mapsto d^{k}f(x, y_{1}, \dots, y_{k})
\]
are continuous. The directional derivative is defined as
\[
D_{y}f(x) = \lim_{t \rightarrow 0} \frac{f(x+ty) - f(x)}{t}.
\]
A map is smooth (or $C_{\mathbb{K}}^{\infty}$-smooth) if it is $C_{\mathbb{K}}^{r}$ for all $r \in \mathbb{N}_{0}$.
\end{definition}

\begin{definition}[Complex analytic function, \cite{11}]
Let $E$ and $F$ be complex locally convex spaces and let $U\subseteq E$ be open. A map $f: U \to F$ is called complex analytic ($\mathbb{C}$-analytic, or $C^\omega_{\mathbb{C}}$) if it is continuous and, for each $x \in U$, there exists a sequence $(p_k)_{k\in \mathbb{N}_0}$ of continuous homogeneous polynomials $p_k: E \to F$ of degree $k$ such that
\[
f(x+ y) = \sum_{k=0}^{\infty}p_k(y)
\] 
as a pointwise limit for all $y$ in a $0$-neighborhood $Y\subseteq U$ such that $x+ Y\subseteq U$.
\end{definition}

\begin{definition}[Real analytic function, \cite{11}]
Let $E, F$ be locally convex real topological vector spaces. Let $U \subseteq E$ be an open set. A map $f: U \rightarrow F$ is called real analytic ($\mathbb{R}$-analytic, or $C^\omega_{\mathbb{R}}$) if it extends to a complex analytic mapping $f: V\to E_{\mathbb{C}}$ on an open subset $V$ of $U$ in $E_{\mathbb{C}}$.
\end{definition}

\subsection{Convenient calculus}
The convenient setting due to Frölicher-Kriegl \cite{20} provides the natural framework for our work.

\begin{definition}
A locally convex topological vector space $E$ is called a \textit{convenient} space if a curve $c: \mathbb{R} \rightarrow E$ is smooth if and only if $\lambda \circ c: \mathbb{R} \rightarrow \mathbb{R}$ is smooth for all continuous linear functionals $\lambda \in E^{\prime}$.
\end{definition}

For a convenient vector space $E$, the $c^{\infty}$-topology is the final topology with respect to the set of all smooth curves $c \in C^{\infty}(\mathbb{R}, E)$. This topology may differ from the original locally convex topology, but for Fréchet (and hence Banach) and Silva spaces, both topologies coincide.

\subsection{Space of multilinear maps}
Let $E, F$ be locally convex or convenient spaces. We denote:
\begin{itemize}
    \item $L_{b}(E, F)$: the space of continuous linear maps from $E$ to $F$ with the topology of uniform convergence on bounded subsets of $E$.
    \item $L_{c}(E, F)$: the space of continuous linear maps from $E$ to $F$ with the topology of uniform convergence on compact subsets of $E$.
    \item $L^{k}(E_{1}, \dots, E_{k}, F)$: the space of continuous $k$-linear maps from $E_{1} \times \dots \times E_{k}$ to $F$.
    \item $E^{\prime}$: the space of bounded linear functionals on $E$. If $E$ is a convenient space, this coincides with the space of continuous linear maps since a convenient space is bornological.
\end{itemize}
All these spaces possess locally convex topological vector space and convenient vector space structures (see \cite[Proposition 5.6]{20}).

\subsection{Smooth vector bundle structure on $L^{2}(T'M)$}
Let
\[
L^2_{\text{skew}}(T'M) := \bigcup_{x\in M} L^2_{\text{skew}}(T'_{x}M)
\]
be a bundle over a convenient manifold $M$, where $L^2_{\text{skew}}(T'_xM)$ denotes the space of skew-symmetric bilinear forms on $T'_xM$ equipped with the topology of uniform convergence on bounded subsets of $T'_xM\times T'_xM$. Then $L^{2}(T'M)$ carries a smooth vector bundle structure over $M$ (see \cite[Section 33]{20}, \cite[Section 1.1, 3.1]{24}).

\begin{definition}[Smooth map]
For convenient locally convex spaces $E$ and $F$, and a $c^\infty$-open subset $U \subseteq E$, a map $f: U \rightarrow F$ is called smooth if for all smooth curves $c: \mathbb{R} \rightarrow E$ such that $c(\mathbb{R}) \subseteq U$, the composition $f \circ c \in C^\infty(\mathbb{R}, F)$.
\end{definition}

\begin{remark}[{\cite[Theorem 3.18]{20}}]
Let $E$ and $F$ be locally convex spaces, and let $U \subseteq E$ be $c^{\infty}$-open. Then the differential map $d: C^{\infty}(U, F) \rightarrow C^{\infty}(U, L(E, F))$, given by
\[
df(x)v := \lim_{t \rightarrow 0} \frac{f(x+tv) - f(x)}{t},
\]
exists for $x \in U, v \in E$, and is linear and bounded (smooth). The chain rule holds:
\[
d(f \circ g)(x) = df(g(x)) dg(x).
\]
\end{remark}

\subsubsection{Topologies on $C^{\infty}(E,F)$}
(see \cite[Definition 3.6, Definition 3.11]{20}). Let $E$ and $F$ be two convenient spaces, and let $U \subset E$ be an open set. The space of smooth curves on $E$, $C^{\infty}(\mathbb{R}, E)$, carries a convenient topological structure as the initial topology with respect to the following mappings:
\[
C^{\infty}(\mathbb{R}, E) \xrightarrow{d_k} C^{\infty}(\mathbb{R}, E) \xrightarrow{r_{|K}} \ell^{\infty}(K, E)
\]
for each $k \in \mathbb{N}_{0}$, where $d_{k}$ denotes the $k$-th derivative of curves, $K$ runs over the compact subsets of $\mathbb{R}$, and $r_{|K}$ is the restriction of curves to the space of bounded functions on $K$ equipped with the topology of uniform convergence on $K$.

The space of smooth maps $U \subseteq E \rightarrow F$ carries a convenient topological structure as the initial topology with respect to the following mappings:
\[
C^{\infty}(U, F) \xrightarrow{c^{*}} C^{\infty}(\mathbb{R}, F)
\]
for every $c \in C^{\infty}(\mathbb{R}, E)$ such that $c(\mathbb{R}) \subseteq U$, where $c^{*}: C^{\infty}(U, F) \rightarrow C^{\infty}(\mathbb{R}, F)$ is defined as $c^{*}(\gamma) := \gamma \circ c$ for $\gamma \in C^{\infty}(U, F)$.

\begin{remark}
The locally convex compact-open-$C^{\infty}$ topology on $C^{\infty}(U, E)$ is finer than the convenient topology (see \cite[1.7]{11}). Let $U \subseteq E$ be an open subset. Then for $k \in \mathbb{N}_{0} \cup \{\infty\}$, the compact-open-$C^{k}$ topology on $C^{k}(U, E)$ is the final topology with respect to the mapping
\[
\mathfrak{p}: C^{k}(U, E) \rightarrow \prod_{j \le k} C(U \times E^{j}, \mathbb{K}), \quad f \mapsto (d^{j}f)_{j \le k}
\]
where $C(U \times E^{j}, \mathbb{K})$ is equipped with the topology of uniform convergence on compact subsets of $U \times E^{j}$.
\end{remark}

\subsection{Convenient manifold, bundle, and Lie group}

\begin{definition}[Convenient Manifold {\cite[27.1]{20}}]
A convenient manifold $M$ is a set equipped with a smooth atlas $\{(U_{\alpha}, \phi_{\alpha})\}_{\alpha \in A}$, where $M = \bigcup_{\alpha} U_{\alpha}$ and each chart $\phi_{\alpha}: U_{\alpha} \rightarrow E_{\alpha}$ is a bijection onto an open subset of a convenient vector space $E_{\alpha}$. Transition maps $\phi_{\alpha\beta} = \phi_{\alpha} \circ \phi_{\beta}^{-1}$ are required to be smooth in the convenient sense (meaning they map smooth curves to smooth curves).
\end{definition}

Given a manifold $M$ modeled on $E$, we denote:
\begin{itemize}
    \item $T_{x}M$: the tangent space at $x \in M$, consisting of kinematic tangent vectors (the space of all derivatives $c^{\prime}(0)$ at 0 of smooth curves $c \in C^{\infty}(\mathbb{R}, M)$ with $c(0)=x$). $TM := \bigcup_{x \in M} T_{x}M$ defines a convenient smooth vector bundle over $M$ (see \cite[28.12]{20}).
    \item $T_{x}^{\prime}M$: the kinematic cotangent space at $x \in M$, defined as the dual space $(T_{x}M)'$. The kinematic $T^{\prime}M := \bigcup_{x \in M} T_{x}^{\prime}M$ defines a smooth vector bundle over $M$ (see \cite[33.1]{20}).
\end{itemize}

\begin{definition}[Smoothly regular convenient manifold]
A convenient manifold $M$ is called \textit{smoothly regular} if the initial topology on $M$ with respect to the family of smooth functions $C^{\infty}(M, \mathbb{R})$ coincides with the manifold topology. Equivalently, $M$ is smoothly regular if the smooth functions separate points and generate the topology of $M$.
\end{definition}

\begin{definition}[Smoothly regular manifold {\cite[Definition 3.5.30]{11}}]
Let $M$ be a $C^{\infty}$-manifold modeled on a real locally convex space. $M$ is called a smoothly regular manifold if the following equivalent conditions (adapted from $C^{r}$-regularity in \cite[Proposition 3.5.29]{11}) are satisfied:
\begin{enumerate}
    \item[(a)] The topology on $M$ is initial with respect to the family of smooth functions $C^{\infty}(M, \mathbb{R})$.
    \item[(b)] For every $x \in M$ and neighborhood $U \subseteq M$ of $x$, there exists a smooth function $f: M \rightarrow \mathbb{R}$ such that $f(x) \neq 0$ and $\text{supp}(f) \subseteq U$.
    \item[(c)] For every $x \in M$ and neighborhood $U \subset M$ of $x$, there exists a smooth function $f: M \rightarrow \mathbb{R}$ such that $f(M) \subseteq [0, 1]$, $\text{supp}(f) \subseteq U$, and $f|_{V} = 1$ for some neighborhood $V \subseteq M$ of $x$.
\end{enumerate}
\end{definition}

\begin{definition}[Smoothly paracompact manifold {\cite[27.4]{20}}]
A convenient manifold $M$ is smoothly paracompact if every open cover of $M$ admits a smooth partition of unity subordinate to it. Note that a smoothly paracompact manifold is automatically smoothly regular and admits smooth bump functions (see \cite[16.10]{20}).
\end{definition}

\begin{definition}[Regular Lie group]
A convenient Lie group $G$ with Lie algebra $\mathfrak{g}$ is called \textit{regular} (in the sense of Kriegl and Michor \cite[38.4]{20}) if:
\begin{enumerate}
    \item For every smooth curve $\xi \in C^{\infty}(\mathbb{R}, \mathfrak{g})$, the initial value problem for the right logarithmic derivative:
    \[
    \begin{cases}
    \dot{g}(t) = R_{g(t)*}\xi(t) \\
    g(0) = e
    \end{cases}
    \]
    has a unique solution $g = \text{Evol}(\xi) \in C^{\infty}(\mathbb{R}, G)$.
    \item The evolution map $\text{Evol}: C^{\infty}(\mathbb{R}, \mathfrak{g}) \rightarrow C^{\infty}(\mathbb{R}, G)$ is smooth.
\end{enumerate}
See \cite[Definition 6.1]{11} for the regularity of a Lie group (in the sense of Milnor \cite{25}), modeled on a locally convex topological vector space.
\end{definition}

\subsection{References for differential calculus}
The standard references for our differential calculus framework are:
\begin{itemize}
    \item \cite{20}: Kriegl-Michor, \textit{The Convenient Setting of Global Analysis} (1997).
    \item \cite{11}: Neeb-Glöckner, \textit{Infinite-Dimensional Lie groups} (2026).
\end{itemize}

\end{document}